\documentclass[a4paper,twocolumn,aps,prl,reprint,superscriptaddress,showpacs,showkeys,longbibliography,floatfix]{revtex4-2}
\usepackage[latin9]{inputenc}
\usepackage{babel}
\usepackage{mathtools}
\usepackage{amsmath}
\usepackage{amssymb}
\usepackage{amsthm}
\usepackage{mathrsfs}
\usepackage{graphicx}
\usepackage[usenames,dvipsnames]{xcolor}
\usepackage{bm}
\usepackage{bbold}
\usepackage{braket}
\usepackage{float}
\usepackage{comment}
\usepackage{dcolumn}
\usepackage{multirow}
\usepackage{enumerate}
\usepackage[shortlabels]{enumitem}
\usepackage{orcidlink}

\definecolor{bblue}{RGB}{0, 186, 255}
\definecolor{purple}{RGB}{128, 0, 128}
\definecolor{darkgreen}{RGB}{0, 155, 0}
\definecolor{green}{RGB}{0, 120, 0}

\usepackage{wrapfig}
\usepackage{hyperref, cleveref}
\hypersetup{
  colorlinks = true,
  linkcolor = blue,
  citecolor=blue,
  urlcolor=blue
}

\newcommand{\raisedtarget}[1]{%
  \raisebox{\fontcharht\font`P}[0pt][0pt]{\hypertarget{#1}{}}%
}

\usepackage[all]{hypcap} 
\newcommand{\C}[1]{ { \color{red} \footnotesize  \textsf{\textsl{#1}} } }
\newcommand{\tr}[1]{\textrm{tr}_{\rm{#1}}}

\newcommand{\posL}{\mathcal{G}(\mathbb{\Lambda})}
\newcommand{\posG}[1]{\mathcal{G}(#1)}
\newcommand{\spos}{\mathcal{G}}
\newcommand{\pqbsz}[1]{\tilde{\sigma}^{(#1)}_{z}}

\newcommand{\non}[1]{\overline{#1}}

\newcommand{\abs}[1]{\left\lvert #1\right\rvert}
\renewcommand{\set}[1]{\left\lbrace #1\right\rbrace}
\renewcommand{\setminus}{\!\mathbin{\backslash}\!}
\newcommand{\CS}{\mathcal{C}} 
\newcommand{\vw}{f}

\newcommand{\wE}{w}

\newcommand{\spec}[1]{\mathrm{Spec}\left(#1\right)}

\newcommand{\MES}{\mathcal{M}}
\newcommand{\mpme}{\left<s_{\rm{m}}\right>}

\renewcommand{\braket}[2]{\langle #1\mid #2\rangle}

\renewcommand{\C}{\mathbb{C}}
\newcommand{\CBS}[1]{\mathcal{H}_{#1}}
\newcommand{\SSP}{\mathcal{R}}
\newcommand{\MS}[1]{\mathrm{Mix}(\QS{#1})}
\newcommand{\QS}[1]{Q_{#1}}
\newcommand{\Spec}{\mathrm{Spec}}
\newcommand{\fat}{t_f}
\newcommand{\transOs}{\mathbb{T}}

\theoremstyle{definition}
\newtheorem{definition}{Definition}
\theoremstyle{remark}

\theoremstyle{lemma}
\newtheorem{lemma}{Lemma}
\theoremstyle{proposition}
\newtheorem{proposition}{Proposition}
\theoremstyle{theorem}
\newtheorem{theorem}{Theorem}
\theoremstyle{corollary}

\theoremstyle{remark}
\newtheorem{example}{Example}
\begin{document}
\title{Bundling of bipartite entanglement}
\author{Maike Drieb-Sch\"on\orcidlink{0009-0000-6849-3145}}
\email{maike.schoen@uibk.ac.at}
\affiliation{Institute for Theoretical Physics, University of Innsbruck, A-6020 Innsbruck, Austria}
\author{Florian Dreier\orcidlink{0000-0002-5674-9147}}
\email{f.dreier@parityqc.com}
\affiliation{Institute for Theoretical Physics, University of Innsbruck, A-6020 Innsbruck, Austria}
\affiliation{Parity Quantum Computing GmbH, A-6020 Innsbruck, Austria}
\author{Wolfgang Lechner\orcidlink{0000-0003-3662-1020}}
\affiliation{Institute for Theoretical Physics, University of Innsbruck, A-6020 Innsbruck, Austria}
\affiliation{Parity Quantum Computing GmbH, A-6020 Innsbruck, Austria}

\begin{abstract}
We investigate bipartite entanglement and prove that in constrained energy subspaces, the entanglement spectra of multiple bipartitions are the same across the whole subspace. We show that in quantum many-body systems the bipartite entanglement entropy is affected in such a way that it forms "bundles" under unitary time evolution. Leveraging the structure of the subspace, we present methods to verify whether the entanglement spectrum of two bipartitions is identical throughout the entire subspace. For the subspace defined by the parity embedding, we further provide an algorithm that can determine this in polynomial time.
\end{abstract}
\maketitle

\setcounter{secnumdepth}{2}

\section{Introduction}
\label{sec:introduction}
In quantum many-body systems, entanglement plays a central role in understanding the complex interplay between individual particles and the collective behavior of the system as a whole. Unlike classical systems, where correlations are typically local and additive, quantum entanglement can exhibit nonlocal and highly nontrivial structures, making it a powerful tool for probing the underlying physics of many-body systems. For example, it provides insights into a wide range of physical phenomena, from quantum phase transitions and criticality to thermalization and the emergence of classical behavior in large systems~\cite{AnnaSanperaDeChiara2018, ES_indtroLiHaldane, PhysRevX.8.021026, PhysRevB.93.174202,Geraedts_2017}. Beyond its fundamental role in condensed matter physics, entanglement is a key resource driving advancements in quantum technologies. Extending its application in quantum simulators and error-correcting codes, understanding and controlling entanglement is crucial for designing quantum algorithms capable of tackling complex optimization problems on quantum hardware. Such algorithms include quantum optimization approaches, such as adiabatic quantum computation in the form of quantum annealing (QA)~\cite{HaukeLechnerKatzgraberNishimori_2020, QuantumAnnealingOverview}, digital quantum computation like quantum approximate optimization algorithm (QAOA)~\cite{farhi2000quantum} and variational quantum algorithms (VQAs)~\cite{vqa}. A common characteristic of these methods is their reliance on initializing the computation in a product state, evolving it towards a solution state that encodes a classical problem and potentially does not contain any entanglement. Hence, a crucial open question is how much entanglement in a quantum computing device is required to reliably achieve a successful solution. Previous works examining the behavior of entanglement during annealing processes, as well as the theoretical connection between bipartite entanglement and the success probability in adiabatic quantum optimization, can be found in~\cite{PhysRevA.69.052308,LantingEntanglementInQA,haukeLechnerEntenglement}.

However, although the study of entanglement provides valuable insights into quantum optimization, several
practical limitations
must be addressed. In general, quantum hardware graphs are limited in their connectivity and offer no long range connections while the most optimization problems are presented by logical graphs with a high connectivity and long range connections. In such cases, they have to be embedded into the quantum hardware by using much more qubits and adding constraints to preserve the degrees of freedom.
In addition, many optimization problems have constraints which are typically encoded through large penalties~\cite{arute2019quantum,bernien2017probing,KochPRAChargeInsensitive2007,saffman2010quantum,henriet2020quantum,bloch2008many}.
On the other hand, recent works~\cite{a12040077, ItayHenPRApplied2016, ItayHenPRA2016, constraintpaper} have shown that some optimization problems have constraints that can be implemented in more efficient quantum annealing methods compared to penalties, where constraints can be encoded directly into specially designed driver Hamiltonians.
All these kinds of constraints lead to a restriction of the system state to a subspace ${\QS{\SSP}\subset \mathcal{H}}$ which is separated from the rest of the Hilbert space ${\mathcal{H}}$ and therefore impacts the entanglement behavior of the quantum state during the evolution process. Here, $\SSP$ denotes a set of basis states in $\mathcal{H}$ that spans the subspace $\QS{\SSP}$. For instance, analog quantum devices typically start in a pure state ${\ket{\psi}\in \mathcal{H}}$ being an equal superposition of all physical basis states and aims to drive the many-body system of qubits through a quantum spin-glass transition~\cite{RevModPhys.58.801, cambridgeQuantumSPinGlassAnnealingComputatio} to end up in a classical state or, in the case of a degenerated ground state, in a superposition of quantum states in $\SSP$.
As shown in~\cite{PhysRevLett.109.050501}, when constraints are incorporated into the optimization problem as energy penalties, initializing the quantum computation across the full Hilbert space ${\mathcal{H}}$ rather than a constrained subspace $\QS{\SSP}$ is crucial to potentially achieve a quantum speed-up in any subsequent quantum-computational process. However, the authors in~\cite{constraintpaper} demonstrate that remaining within the constrained subspace $\QS{\SSP}$ throughout the whole computation is also possible with specifically designed driver Hamiltonians without loosing performance.

In this work, we investigate the impact of these restrictions via constraints on the entanglement behavior within a quantum system. Specifically, we study bipartite entanglement~\cite{PhysRevA.69.052308, JournalofIndianInstofSciencRoleentanglemenQC, haukeLechnerEntenglement} for all bipartitions when the system state is a superposition of quantum states in $\SSP$. We observe that for quantum-computational processes, which drive the system state to a final entangled but to the subspace $\QS{\SSP}$ restricted pure state, the entanglement values of different bipartitions seem to bundle to exactly the same entanglement values at the end of the process. In the setting of~\cite{constraintpaper}, we even show that the entanglement values of different bipartitions behave equally during the whole computation process. Based on this observation, we further provide a mathematical proof of this phenomenon and show that for multiple bipartitions the spectra of the reduced density matrix, which are also called \emph{entanglement spectra}~\cite{ES_indtroLiHaldane}, are exactly the same, independent of the quantum state restricted by the subspace $\QS{\SSP}$. Since for pure states, and even for mixed states of spin-{1/2} particles, all measures of bipartite entanglement are functions of the eigenvalues of the reduced density matrix~\cite{RevModPhys.81.865, PhysRevA.65.032314}, this even shows that our result is applicable to any bipartite entanglement measure and can be utilized for the analysis of other properties such as topological order; see, for example~\cite{ES_indtroLiHaldane, PhysRevA.78.032329, Pollmann_2010, ES_PichlerZoller}. Moreover, we provide a mathematical tool that allows to determine whether any two bipartitions yield the same entanglement spectrum across the whole space spanned by $\SSP$. Finally, we demonstrate, using two types of embeddings - the parity embedding and the minor embedding - how these mathematical tools can be applied in a more efficient way. For the parity embedding, we even provide an algorithm that performs this verification in polynomial time.

The remainder of this paper is organized as follows.
In Sec.~\ref{sec:background} we briefly summarize the background knowledge used throughout this paper.
Section~\ref{sec:entanglementbackground} provides an overview of how bipartite entanglement arises within quantum optimization. From Sec.~\ref{sec:QAbackground} onward, we introduce quantum annealing and the parity embedding, which we use both for our numerical example, as well as the minor embedding, which provides another application example alongside the parity embedding. Our numerical simulations are discussed in Sec.~\ref{sec:numVis}. Section~\ref{sec:theoFrameWork} then follows with a more detailed description of our theoretical results. In Sec.~\ref{subsec:appl_embed-based_quantum_opt} application examples in quantum optimization are provided. Finally, Sec.~\ref{sec:methods} contains all the details necessary to reproduce and verify our results.

\section{Background}
\label{sec:background}
First, we provide a more detailed description of the entanglement spectrum and the entanglement entropy. Next, we briefly explain the concept of quantum annealing as an example of adiabatic quantum computing, along with two embeddings that can be used to implement quantum optimization problems on real hardware. Both adiabatic quantum computing and the embeddings are used solely to illustrate one of the possible applications of our results.

\subsection{Entanglement spectrum and entropy}
\label{sec:entanglementbackground}
The \emph{entanglement spectrum} for a bipartition $\set{A, A^c}$ of a set of qubits for a pure quantum state $\rho$ is defined as the spectrum of the reduced density matrix ${\rho_{A}=\mathrm{tr}_A{\rho}}$, i.e., ${\spec{\rho_{A}}=\set{\lambda \mid \lambda \text{ is an eigenvalue of } \rho_A}}$. It was first introduced in~\cite{ES_indtroLiHaldane} as a tool to identify topological order. In subsequent years, it was recognized that the entanglement spectrum plays a crucial role in the context of eigenstate thermalization~\cite{PhysRevX.8.021026}, many-body localization~\cite{PhysRevB.93.174202,Geraedts_2017}, quantum criticality~\cite{PhysRevA.78.032329} and tensor networks~\cite{PhysRevB.83.245134, PhysRevLett.111.090501}.
In contrast to the entanglement spectrum, which details the full distribution of entanglement values, the entanglement entropy reduces this information by mapping the entanglement spectrum to a single value. Well-known examples for entanglement entropies of pure quantum states include the von Neumann entropy, the R\'{e}nyi entropy, the linear entropy, the negativity, the Schmidt rank or the geometric measure~\cite{AnnaSanperaDeChiara2018}. In our numerical experiments, we make use of the von Neumann entanglement entropy ${S_A(\rho)=-\tr{}{(\rho_A\ln \rho_A)}=-\sum_{\lambda\in\spec{\rho_A}}\lambda\ln\lambda}$, which serves as the standard measure for quantifying entanglement between two parts of the system.

\subsection{Adiabatic quantum computing}
\label{sec:QAbackground}
In general, an optimization problem can be described by an all-to-all connected Ising spin Hamiltonian
\begin{align}\nonumber
    H_p &= \sum_{i} J_i \sigma_{z}^{(i)} + \sum_{i<j} J_{ij} \sigma_{z}^{(i)}\sigma_{z}^{(j)} \\
    &+ \sum_{i<j<k} J_{ijk} \sigma_{z}^{(i)} \sigma_{z}^{(j)}\sigma_{z}^{(k)} + ...~.
\label{eq:generalProblemH}
\end{align}
where its solution is given by the ground state of this Hamiltonian. One method of finding the ground state of a problem Hamiltonian $H_{p}$ is to evolve the system from an initial state, being a trivial ground state, adiabatically to the ground state of the problem Hamiltonian. Here we consider a quantum system which evolves from an initial time ${t=0}$ to a final time ${t=\fat}$ and its time evolution can be described by
\begin{equation}
\label{eq:ht}
    H(t) = g(t) H_0 + s(t) H_p,
\end{equation}
where ${g(t),s(t)\geq 0}$ with ${g(0) > 0}$, ${g(\fat)=0}$ and ${s(\fat)>0}$.
We call $\set{g(t), s(t)}$ the \emph{schedule} of the time evolution process.
Such a process is called \emph{adiabatic}, if the speed of the evolution is slow enough, such that for all times $0\leq t \leq \fat$ the quantum system stays in the ground state of $H(t)$. Otherwise, it is called a \emph{nonadiabatic} process.
For the numerical simulations in this paper we use a linear schedule with ${g(t) = 1 - s(t)}$ and ${s(t)=\frac{t}{t_f}}$ for which a higher process time $\fat$ means a lower speed of the evolution.

\subsection{Parity embedding \label{sec:parityembedding}}
The parity embedding~\cite{compilerpaper} is a generalization of the LHZ mapping~\cite{lechner2015quantum}, and it maps an arbitrarily connected logical problem Hamiltonian as given in Eq.~\eqref{eq:generalProblemH} to a Hamiltonian containing only local three- and four-body interactions by mapping each product in the logical Hamiltonian onto a single parity qubit. The problem Hamiltonian for the parity embedding is given by
\begin{equation}
    \label{eq:parityHamiltonian}
    H_{p} = \sum^{N_p-1}_{m = 0} \tilde{J}_{m} \pqbsz{m} - H_{C},
\end{equation}
where we denote the Pauli-$Z$ operators of the parity qubits as $\pqbsz{m}$ and
\begin{equation}
    \label{eq:parityConstraintH}
    H_{C} = \sum_{p=1}^{N_C} C_{p} \pqbsz{p_1} \pqbsz{p_2} \pqbsz{p_3} \left(\pqbsz{p_4}\right) 
\end{equation}
is the parity constraint Hamiltonian, where ${N_C = N_p - N_l + D}$ is the number of necessary parity constraints, which is calculated by the number of parity qubits $N_p$, the number of logical spins $N_l$ and the number of spin flip degeneracies $D$ of the logical problem Hamiltonian. We note that these $N_C$ parity constraints are necessary to end up in a mapping with the same degrees of freedom as the original problem. Each of the $N_C$ parity constraints consists of a product of three or four Pauli-$Z$ operators of the parity qubits and a well chosen penalty strength $C_{p}$. This set of parity constraints restricts the system state to being a superposition of states that are contained in a subspace $\SSP$, which we denote in this work by $\Pi$. An example for the parity embedding of an all-to-all connected logical problem, which we use throughout this paper, is given in Appendix~\ref{app:exampleParityEmbedding}.

\subsection{Minor embedding\label{sec:minorembedding}}
The idea of the minor embedding~\cite{choi_minor-embedding_2008} is to map the logical graph onto one of its minor graphs in such a way that it matches the hardware connectivity graph. For this purpose, a single node $v$ in the logical graph is mapped onto $N_v$ physically connected nodes, which are referred to as a \emph{chain} in the minor embedding. Similarly to the parity embedding, where penalty constraints are implemented by three- and four-body couplers, penalties are realized as couplers between each two directly connected physical qubits in a chain to force them to have the same alignment. These are the constraints in the minor embedding restricting again the system state to be an element of a corresponding subspace $\QS{\SSP}$. The specific details of the Hamiltonian for the minor embedding are not essential for the purposes of this work.

\section{Results}
\label{sec:results}
In this section, we present our main results of this work, where we first introduce in Sec.~\ref{sec:numVis} numerical experiments demonstrating that different bipartitions yield the same entanglement entropy value for all quantum states which are contained in a lower energy subspace. Subsequently, in Sec.~\ref{sec:theoFrameWork} we develop a rigorous mathematical framework and derive theoretical results that provide a formal justification for the bundling phenomenon observed in our numerical simulations. Our main theoretical result shows that the entanglement spectra of subsystems belonging to the same equivalence class, as defined by a certain equivalence relation, are identical for all quantum states in a given subspace. Therefore, in this article we suppose that the quantum states correspond to elements of the subspace
\begin{multline}
\label{def:subsystemstateset}
    \QS{\SSP}\coloneqq\Bigg\{\left.\sum_{\ket{\psi}\in\SSP}c_{\ket{\psi}}\ket{\psi} \,\right\vert \forall\ket{\psi}\in\SSP\colon c_{\ket{\psi}}\in \C\\\text{and} \sum_{\ket{\psi}\in\SSP}\abs{c_{\ket{\psi}}}^2=1\Bigg\}
\end{multline}
generated by some ${\SSP\subset \CBS{n}}$, where
$$
\CBS{n}\coloneqq\set{\ket{i_1\ldots i_n}\mid i_1,\ldots,i_n\in\set{0,1}}
$$
denotes the computational basis and ${n\in\mathbb{N}}$ the number of qubits of the quantum system.\\
In the following, for brevity, we write ${[n]\coloneqq\set{1,\ldots,n}}$, denoting the set of all integers between one and $n$.

%
%
\subsection{Numerical results\label{sec:numVis}}
\begin{figure*}[t!]
\centering
\includegraphics[width=1.9\columnwidth]{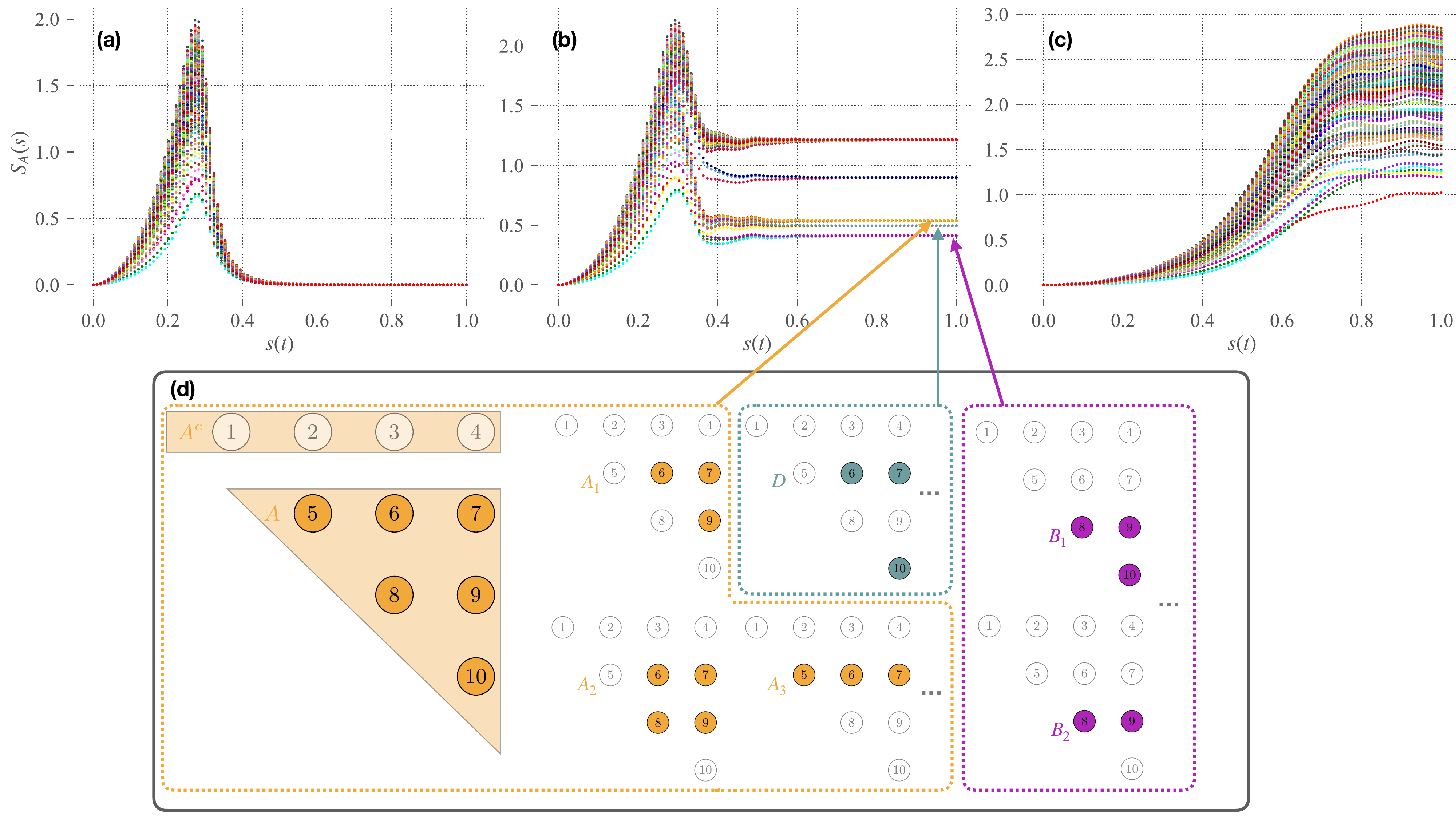}
\caption{\emph{Dynamics of von Neumann entanglement entropy $S_{A}(s)$} are presented in the top row for three different annealing processes with linear schedule and final annealing times \textbf{(a)} ${\fat=800}$, \textbf{(b)} ${\fat=100}$, and \textbf{(c)} ${\fat=11}$ of a system with 10 physical qubits laid out on a square lattice realizing a parity embedding of a complete graph with five vertices. We used the parity constraint strengths in \textbf{(a)} and \textbf{(b)} ${C=4}$ and in \textbf{(c)} ${C=1}$. Each process presents the entanglement dynamics of $266$ different unordered bipartitions $\set{A, A^c}$ of sizes ${\abs{A} \in \set{2,3, 4, 5}}$. Panel~\textbf{(d)} illustrates a few of the bipartitions of the physical qubits in a square lattice, which belong to one of the five visible bundles in Fig.~\ref{fig:GroupEntanglementDynamics}~\textbf{(b)}.
The 10 physical qubits are presented with circles labeled with the physical indices from $1$ to $10$.
Each triangle represents one possible bipartition whose qubits are filled with the color of its associated bundle's entanglement entropy.
Additionally, the bipartitions corresponding to the same final entanglement value are framed by dashed lines in the same color.
On the far right, two examples of the lowest bundle (violet) can be seen, in the middle one example of the second lowest (green-blue), and on the left four examples of the bundle in the middle (orange). All the smaller drawn bipartitions are subsets of the larger bipartition $\set{A, A^c}$ drawn on the far left.
}
\label{fig:GroupEntanglementDynamics}
\end{figure*}
Here we demonstrate numerical experiments showing how bipartite entanglement bundles to a common entanglement value.
We simulate an annealing process for a system of 10 qubits placed on a square grid by utilizing the time-dependent Hamiltonian given in Eq.~\eqref{eq:ht}.
The employed problem Hamiltonian $H_p$ is given by Eq.~\eqref{eq:parityHamiltonian} with equal energy constraints strengths ${C_{p} = C}$. The appropriate strength of these penalties allows us to control whether the state of the system remains in the subspace $\QS{\Pi}$ or not.
The problem Hamiltonian is constructed by applying the parity embedding of a complete graph with five vertices, which presents an optimization problem with all-to-all connectivity.
In this example, we use for the parity embedding the following mapping between the physical vertices (qubits) and the logical edge indices:
${1\rightarrow \set{2,3}}$, ${2\rightarrow \set{1,3}}$,
${3\rightarrow \set{0,3}}$,
${4\rightarrow \set{3,4}}$,
${5\rightarrow \set{1,2}}$,
${6\rightarrow \set{0,2}}$,
${7\rightarrow \set{2,4}}$,
${8\rightarrow \set{0,1}}$,
${9\rightarrow \set{1,4}}$,
${10\rightarrow \set{0,4}}$. For more details of this embedded problem, we refer to Appendix~\ref{app:exampleParityEmbedding}.
Figure~\ref{fig:GroupEntanglementDynamics} shows the numerical results for the dynamics of the entanglement entropy for three quantum annealing processes.
The random problem instance $H_{p}$ is encoded in the local fields $\tilde{J}_{m}$ defined by the field vector
\begin{align*}  
    \tilde{J} = (&0.58, ~-0.5, ~-0.3, ~-0.2, ~0.41,\\ -&0.53,
    ~0.48, ~-0.31, ~-0.19, ~0.39)^T.
\end{align*}
In Figs.~\ref{fig:GroupEntanglementDynamics}~{(a)} and~\ref{fig:GroupEntanglementDynamics}~{(b)} we used a sufficiently high penalty strength $C$ to restrict the system to the subspace $\QS{\Pi}$, while in Fig.~\ref{fig:GroupEntanglementDynamics}~(c) the chosen penalty strength is too low to enforce the system to stay in an energetic subspace allowing the state of the system to become an arbitrary superposition of the computational basis states in $\CBS{10}$.
Figure~\ref{fig:GroupEntanglementDynamics}~{(a)} is expected for an adiabatic annealing process. The system evolves slowly enough such that entanglement entropy reaches a maximal value before decaying to zero. In contrast, Fig.~\ref{fig:GroupEntanglementDynamics}~{(c)} depicts a highly nonadiabatic process, with an entanglement continuously growing throughout the whole annealing process. Figure~\ref{fig:GroupEntanglementDynamics}~{(b)} illustrates a nonadiabatic process for an intermediate final annealing time, where entanglement decays after its maximum to a non-zero final value. An important observation here is that the entanglement entropy of $266$ different unordered bipartitions ${\set{A, A^c}}$, which are varying in size and contain different elements, \emph{bundle} to only five different entropy values.
Because the penalty strength of the parity constraints is sufficiently high, the subspace $\QS{\Pi}$ generated by the set of all parity states ${\Pi\subset \CBS{10}}$ is energetically separated from ${\CBS{10}\setminus \Pi}$. As long as the speed of the annealing process is lower than a certain limit, which is determined by the energy gap between the lower energy subspace and the rest of the Hilbert space, the entanglement of the different bipartitions bundle to a common final entanglement value.
In Fig.~\ref{fig:GroupEntanglementDynamics}~{(d)} we illustrate some bipartitions, which belong to one of the entanglement entropy bundles. Here, we consider the reference bipartition ${\set{A, A^c}}$ with ${A=\set{5, 6, 7, 8, 9, 10}}$ of the 10 physical qubits [lower orange sub triangle on the left-hand side in Fig.~\ref{fig:GroupEntanglementDynamics}~{(d)}] and take several subsets from $A$. This reference bipartition corresponds to the entanglement value of the orange bundle in Fig.~\ref{fig:GroupEntanglementDynamics}~{(b)}. The other three orange colored subsets ${A_1, A_2}$ and $A_3$ belong to an entanglement value in the same bundle, while the two violet colored subsystems $B_1$ and $B_2$ correspond to the lowest and the green-blue one $D$ to the second lowest entropy bundle.
\begin{figure}[t!]
\centering
\includegraphics[width=1.\columnwidth]{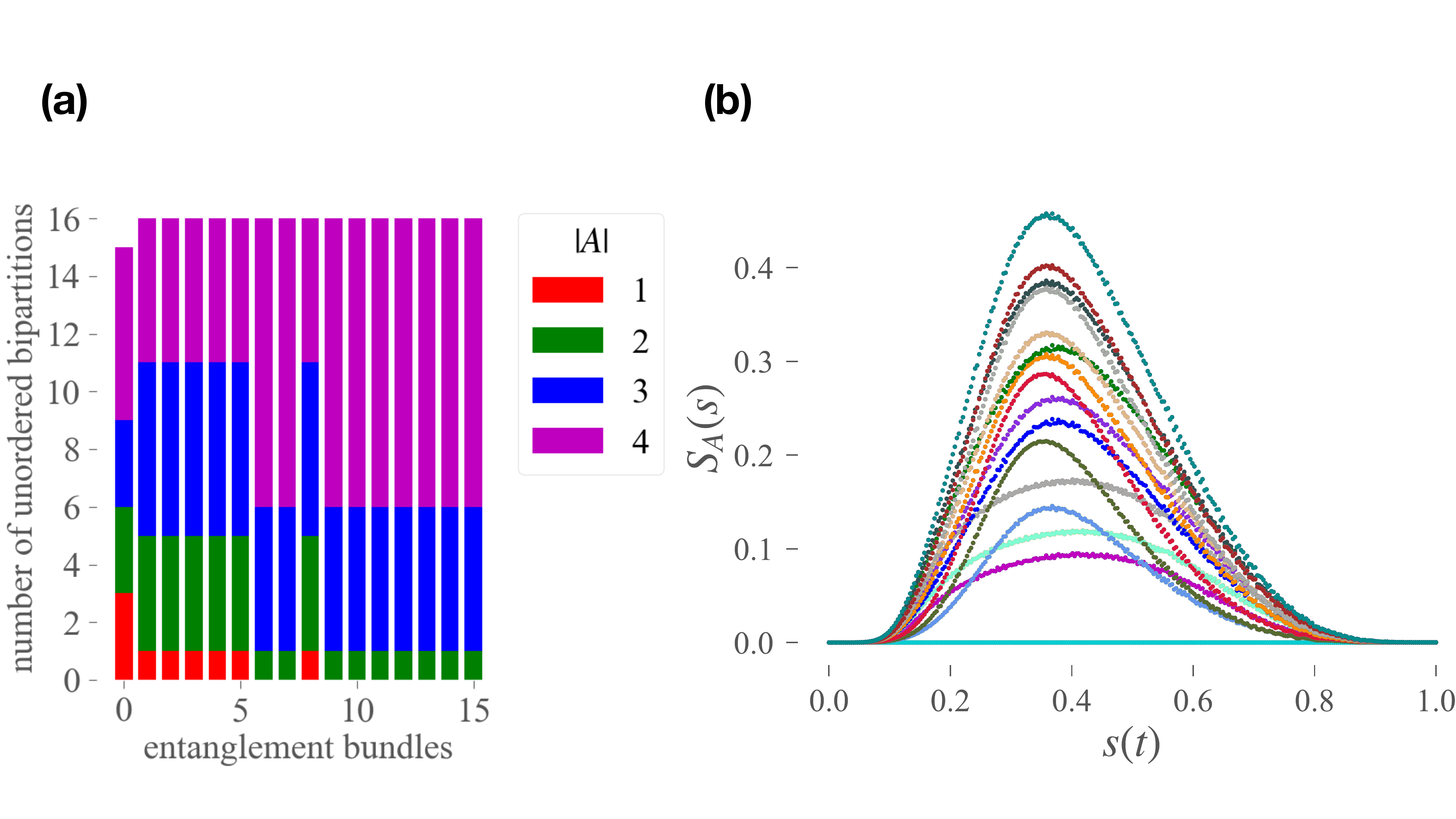}
\caption{\emph{Bundles and dynamics of the von Neumann entanglement entropy $S_{A}(s)$} for the example with nine physical qubits given in~\cite{constraintpaper}. As sum constraint we chose ${\langle\tilde{\sigma}_{z}^{(3,4)} + \tilde{\sigma}_{z}^{(1,4)} + \tilde{\sigma}_{z}^{(1,6)} \rangle = 3}$ and the initial state is set to ${\ket{\Psi(0)} = \ket{0 1 1 1 1 1 0 1 0}}$. Panel~{(a)} illustrates the entanglement bundles containing all $255$ bipartitions $\set{A,A^c}$ of size ${\abs{A}\in\set{1,2,3,4}}$. Each bar corresponds to one entanglement bundle enumerated from $0$ to $15$ where the bundles are sorted by their maximal entanglement entropy value in ascending order. The $y$ axis shows the number of bipartitions of each bundle and the colors correspond to the size of $A$.
Panel~(b) shows the dynamics of the von Neumann entanglement entropy of all $255$ unordered bipartitions. The final annealing time of the adiabatic annealing process is ${t_f=500}$. Further details of the annealing process and the driver Hamiltonian are available in~\cite{constraintpaper}. 
}
\label{fig:GroupEntanglementDynamicsConstrainedPaper}
\end{figure}

 In Sec.~\ref{sec:embeddingExamples} we present an example justifying the existence of bundles that include the orange bundle shown in Fig.~\ref{fig:GroupEntanglementDynamics}~{(b)}, based on our theoretical results. Section~\ref{section:applicationExample} verifies the existence of the remaining four bundles in this figure.
Additionally, we estimate for each of the classes how many bipartitions contribute to each of the entanglement bundles.
In Appendix~\ref{app:bundlesEntanglement} we provide an evaluation of the bundles for all existing bipartitions. It is important to note that the number of bundles is invariant to the specific values of the local fields and interaction strengths of the problem instance. However, as the bundle cardinality depends on the embedding and problem graph, any modification to either component requires a recalculation.
For the sake of completeness, we present in Appendix~\ref{app:entanglementSpectrum} the instantaneous entanglement spectrum for all bipartitions in the bundle with the highest entropy value (red) in Fig.~\ref{fig:GroupEntanglementDynamics}~{(b)}.

Following the approach in~\cite{constraintpaper}, 
the second numerical experiment in Fig.~\ref{fig:GroupEntanglementDynamicsConstrainedPaper} presents the results of a quantum computational process where the system is restricted to a subspace throughout the whole evolution process. In this example, we observe in total $16$ bundles~[Fig.~\ref{fig:GroupEntanglementDynamicsConstrainedPaper}~{(a)}], one bundle with $15$ and $15$ bundles with $16$ different bipartitions. Remarkably, the entanglement dynamics are identical for all bipartitions belonging to the same bundle. To illustrate this, we have visualized the dynamics of all bipartitions in Fig.~\ref{fig:GroupEntanglementDynamicsConstrainedPaper}~{(b)}. The same result applies to a QAOA process as defined in Ref.~\cite{constraintpaper}, which also restricts the system states to the subspace throughout the entire evolution.

In general, these entanglement bundles are expected to occur in any quantum algorithm, and therefore also for any annealing schedule, in which the system is constrained to such a subspace $\QS{\SSP}$ over an extended period of time.
The bundling emerges once the system becomes confined to a particular subspace $\QS{\SSP}$, which depends strongly on the specific algorithm and the corresponding time protocol.
For annealing instances such as the second example shown in Fig.~\ref{fig:GroupEntanglementDynamicsConstrainedPaper}, where the system is initialized within the constrained subspace and remains there throughout the evolution, the bundles persist over the entire annealing process from the very beginning.
In contrast, for standard annealing processes (see Fig.~\ref{fig:GroupEntanglementDynamics}), bipartite and multipartite entanglement typically become significant near the intermediate annealing regime where the driver Hamiltonian and problem Hamiltonian have comparable strength~\cite{LantingEntanglementInQA}.
Since these constraints dominate only after this intermediate annealing regime, the onset of bundling occurs after the entanglement maximum, such that a good reference position for the onset of bundling is given by the mean position ${\mpme = \left<s^A_{\rm{m}}\right>_{A}}$ (averaged over all bipartitions $A$) of the maximal entanglement entropy ${S_{A}(s^A_{\rm{m}}) = \max_{s}S_{A}(s)}$.
 This means that the point $\mpme$ and therefore the onset of bundling depends on the problem Hamiltonian, the strength of the embedding penalties and the annealing schedule.
In Appendix~\ref{app:bundlingOnsetDependency} we examine in more detail how the bundling onset changes when the penalty strength $C$ for the parity embedding is varied and the annealing schedule function is slightly modified.
Across $500$ random instances, the values of $\mpme$
 exhibit substantial scatter for both different penalty strengths and modified annealing schedules. Remarkably, however, the corresponding shifts in $\mpme$ take only certain values. Larger penalty strengths $C$ move $\mpme$ to smaller values, whereas delaying the point at which the driver and problem Hamiltonian become competitive shifts $\mpme$ to larger values.
Additionally, in Appendix~\ref{app:thermalNoise} we demonstrate that thermal noise induces a dispersive broadening of the entanglement bundles. Depending on the noise strength, the bundles that form after the point $\mpme$ progressively separate into distinct entanglement values in a fan-like manner as the annealing process evolves.

Our simulations of the annealing processes have been implemented with the Python package~\texttt{qutip}. For the analysis of observable bundles in Fig.~\ref{fig:GroupEntanglementDynamics}~(b) and Fig.~\ref{fig:GroupEntanglementDynamicsConstrainedPaper} we used~\texttt{dbscan}.

\subsection{Theoretical framework\label{sec:theoFrameWork}}
As already mentioned in the introduction, our main approach to capturing the observed bundling of entanglement mathematically relies on the examination of the spectra of the reduced density matrices given by quantum states in the subspace $\QS{\SSP}$. Consequently, if the spectra of the reduced density matrices of different subsystems are equal, all spectral-based entanglement measures will yield the same entanglement value for all of these different subsystems. As we observe in Fig.~\ref{fig:GroupEntanglementDynamics}, the von Neumann entropy of different subsystems are equal from that time point on where the evolved state lies in the subspace $\QS{\SSP}$ and stays in the subspace $\QS{\SSP}$. Therefore, we seek for a condition that implies identical spectra for different subsystems independent of the quantum state in $\QS{\SSP}$. As all states in $\QS{\SSP}$ are a superposition of elements in $\SSP$ and, by definition of the partial trace, the reduced density matrix of a state in $\QS{\SSP}$ with respect to a subsystem $A$ removes the subsystem $A^c$, it seems likely that such a condition is solely based on the remaining basis states in $\SSP$ restricted on $A$. Note that for a computational basis state 
${\ket{\psi} = |a_1 \ldots a_n\rangle \in \CBS{n}}$
and a nonempty subset ${A}$ of indices ${1\leq i_1 < \cdots <i_m\leq n}$, we denote by
$$
\ket{\psi_A}\coloneqq\ket{\psi}_A\coloneqq\ket{a_{i_1}\ldots a_{i_m}}
$$
the restriction of $\ket{\psi}$ to the subsystem $A$. The following definition addresses this point and will be utilized to establish a condition, which will be introduced subsequently, implying equal spectra for all quantum states in $\QS{\SSP}$ for different subsystems. 
\begin{definition}
\label{def:equvalenzPartition}
    Let ${A\subset[n]}$ be a nonempty subsystem. Then we set
    \begin{equation}
        \label{eq:equi_rel_subset_qubits}
\ket{\psi}\sim_A\ket{\phi}:\Longleftrightarrow\ket{\psi_A}=\ket{\phi_A},\quad \ket{\psi},\ket{\phi}\in\SSP.
    \end{equation}
    Moreover, if ${A=\emptyset}$, then we write ${\ket{\psi}\sim_{\emptyset}\ket{\phi}}$ for the equivalence relation where all states in $\SSP$ are equivalent.
\end{definition}
We emphasize that Eq.~\eqref{eq:equi_rel_subset_qubits} defines an equivalence relation on $\SSP$ and therefore the quotient set of $\SSP$ induced by the equivalence relation $\sim_A$, which we denote by ${\SSP/{\sim_A}}$, is a partition of $\SSP$.

Based on Definition~\ref{def:equvalenzPartition}, we now present the key relation between two subsystems $A_1$ and $A_2$, which compares the quotient sets induced by $\sim_{A_1}$ and $\sim_{A_1^c}$ against the quotient sets induced by $\sim_{A_2}$ and $\sim_{A_2^c}$ on $\SSP$. Notably, this relation proves sufficient to establish the equality of spectra across the whole subspace $\QS{\SSP}$.
\begin{definition}
\label{def:equivalenzrelation}
    We define
    \begin{equation*}
        \label{eq:equi_rel_subspace}
        A_1\sim_{\SSP} A_2:\Longleftrightarrow \set{\SSP/{\sim_{A_1}},\SSP/{\sim_{A_1^c}}}=\set{\SSP/{\sim_{A_2}},\SSP/{\sim_{A_2^c}}}
    \end{equation*}
    for two subsystems ${A_1,A_2\subset[n]}$.
\end{definition}
Similarly to~\eqref{eq:equi_rel_subset_qubits}, $\sim_{\SSP}$ defines an equivalence relation on the power set of ${[n]}$, which we denote by ${\mathcal{P}([n])}$.

In the following, we set ${\mathcal{B}\coloneqq \mathcal{P}([n])/{\sim_{\SSP}}}$, the quotient set of $\sim_{\SSP}$, and call each element in $\mathcal{B}$ a \emph{bundle}. Moreover, we call $\mathcal{B}$ the \emph{set of all bundles in $\SSP$}. The subsequent theorem presents one of our main results in this article.
As we will later see in Sec.~\ref{subsec:appl_embed-based_quantum_opt}, the bundles of the subspace in the example of Fig.~\ref{fig:GroupEntanglementDynamics} coincide with the bundles we actually observe in Fig.~\ref{fig:GroupEntanglementDynamics}~(b), thereby justifying the term \emph{bundle}.
\begin{theorem}\label{thm:suff_cond_equivalence_spectrum_subspace_main}
    Let ${A_1,A_2\subset[n]}$ be two nonempty subsystems with ${A_1,A_2\neq [n]}$ satisfying ${A_1\sim_{\SSP} A_2}$. Then, for all ${\ket{\Psi}\in \QS{\SSP}}$ the relation ${\Spec(\rho_{A_1})=\Spec(\rho_{A_2})}$ holds, where $\rho_{A_1}$ and $\rho_{A_2}$ denote the reduced density matrices of $\ket{\Psi}$ with respect to $A_1$ and $A_2$. Conversely, the equality of the spectra of the reduced density matrices of all quantum states in $\QS{\SSP}$ with respect to two subsystems ${A_1, A_2\in\mathcal{P}([n])}$ does not necessarily imply that ${A_1\sim_\SSP A_2}$.
\end{theorem}
A proof of Theorem~\ref{thm:suff_cond_equivalence_spectrum_subspace_main} is given in Sec.~\ref{sec:methods}.
Note that the above theorem can also be interpreted as an extension of the well-known result in quantum computation and information theory stating that the entanglement spectra of a quantum state with respect to some subsystem $A$ and $A^c$ are identical.

In Sec.~\ref{subsec:bundl_spectr-based_entgl_measures} we will show that equivalent subsystems under the relation $\sim_{\SSP}$ implies not only equal spectra for all pure states in $\QS{\SSP}$, but also equal spectra for all mixed states in $\QS{\SSP}$. Moreover, in Sec.~\ref{subsec:bundl_spectr-based_entgl_measures}  we will show that there exist subspaces ${\SSP\subset\CBS{n}}$ of arbitrarily large size greater than three with growing ${n\in\mathbb{N}}$ satisfying the conditions in the second statement of Theorem~\ref{thm:suff_cond_equivalence_spectrum_subspace_main}. We further note that the time complexity to verify for two subsystems ${A_1,A_2\subset[n]}$ the condition ${A_1\sim_{\SSP}A_2}$ is of order ${O(n\abs{\SSP}^2)}$ as validating condition~\eqref{eq:equi_rel_subset_qubits} takes ${O(n)}$ computations and determining the quotient sets requires ${O(n\abs{\SSP}^2)}$ evaluations. As such subspaces are typical exponentially large, this verification can be extremely costly.
However, as demonstrated by the parity embedding, in Sec.~\ref{subsec:appl_embed-based_quantum_opt} we derive an algorithm applicable for Hamiltonians, which are transformed via embeddings, reducing this exponential-time verification to polynomial time.
Note that the same entanglement spectrum does not imply that the entanglement structure is the same, which would require that also the eigenvectors of the reduced density matrices are equal.

We conclude this section with a simple example to further familiarize the reader with our notation and compute the set of bundles of the subspace spanned by ${\SSP=\set{\ket{\psi_1}, \ket{\psi_2}, \ket{\psi_3}}}$, where
\begin{equation*}
    \ket{\psi_1}\coloneqq\ket{000},\quad  \ket{\psi_2}\coloneqq\ket{100}\quad\text{and}\quad\ket{\psi_3}\coloneqq\ket{111}.
\end{equation*}
Defining
$${A_1\coloneqq\set{1}},\ {A_2\coloneqq\set{2}},\ {A_3\coloneqq\set{3}}\ \text{and}\  {A_4\coloneqq\set{1,2,3}},$$
we deduce the following quotient groups
\begin{center}
\begin{tabular}{c| c }
    $A$ & $\SSP/{\sim_{A}}$\\\hline
    $A_1$ & $\set{\set{\ket{\psi_1}},\set{\ket{\psi_2},\ket{\psi_3}}}$\\
    $A_1^c$ & $\set{\set{\ket{\psi_1},\ket{\psi_2}},\set{\ket{\psi_3}}}$\\\hline
    $A_2$ & $\set{\set{\ket{\psi_1},\ket{\psi_2}},\set{\ket{\psi_3}}}$\\
    $A_2^c$ & $\set{\set{\ket{\psi_1}},\set{\ket{\psi_2}},\set{\ket{\psi_3}}}$\\\hline
    $A_3$ & $\set{\set{\ket{\psi_1},\ket{\psi_2}},\set{\ket{\psi_3}}}$\\
    $A_3^c$ & $\set{\set{\ket{\psi_1}},\set{\ket{\psi_2}},\set{\ket{\psi_3}}}$\\\hline
    $A_4$ & $\set{\set{\ket{\psi_1}},\set{\ket{\psi_2}},\set{\ket{\psi_3}}}$\\
    $A_4^c$ & $\set{\set{\ket{\psi_1},\ket{\psi_2},\ket{\psi_3}}}$
\end{tabular}
\end{center}
and thus, we obtain the bundles ${B_1\coloneqq\set{A_1,A_1^c}}$, ${B_2\coloneqq \set{A_2,A_2^c,A_3,A_3^c}}$ and ${B_3\coloneqq\set{A_4,A_4^c}}$, yielding the set of all bundles ${\mathcal{B}=\set{B_1,B_2,B_3}}$.
Hence, applying Theorem~\ref{thm:suff_cond_equivalence_spectrum_subspace_main} shows that for each ${i=1,2,3}$ the entanglement spectrum with respect to any subsystem $A$ in the bundle $B_i$ is identical for every state in $\QS{\SSP}$.
\section{Applications to embedding-based quantum optimization\label{subsec:appl_embed-based_quantum_opt}}
In this section we demonstrate several examples illustrating how the results presented in the previous section can be applied to quantum optimization. Specifically, we examine cases where an original problem graph is transformed via an embedding to a new problem graph. This allows to formulate an equivalent problem that can be implemented on contemporary quantum devices with limited qubit connectivity, circumventing the hardware's connectivity constraints. In order to obtain an equivalent optimization problem, such embeddings typically add constraints to the new problem that have to be satisfied, yielding naturally an energetic subspace spanned by $\SSP$ which can be uniquely mapped back to the space of the logical problem.

Given an embedding and the resulting basis set $\SSP$, we first introduce in Sec.~\ref{sec:opBasedForm} an operator-based formulation of Definition~\ref{def:equivalenzrelation}.
In Sec.~\ref{sec:embeddingExamples} we investigate concrete embeddings and derive efficient algorithms based on Sec.~\ref{sec:opBasedForm} for the parity embedding and for the minor embedding, which verify whether two subsystems are equivalent under $\sim_{\SSP}$. In case of the parity embedding, we will even observe for this verification a reduction in time complexity from exponential to polynomial scaling. Moreover, in Sec.~\ref{sec:embeddingExamples} we utilize the operator-based formulation to characterize certain bundles and determine their elements with respect to sets $\SSP$, which results from the application of the parity embedding on complete problem graphs.

\subsection{Operator-based formulation 
of equivalent subsystems\label{sec:opBasedForm}}
We begin by introducing an operator-based condition and we will show that this is equivalent to the equivalence relation ${\sim_{\SSP}}$ presented in Definition~\ref{def:equivalenzrelation}. To formulate this alternative condition, we require the existence of a collection $\transOs$ of operators that act on the spanning set ${\SSP}$, each mapping elements of  ${\SSP}$ to other elements of ${\SSP}$, and satisfies certain properties we will specify subsequently. As we will later see in Sec.~\ref{sec:embeddingExamples}, such operator sets can be found by employing the individual characteristics of the embedding.

Note that in the following, given a set $M$ and a subset $F$ of the set of all functions from $M$ to $M$, we denote the closure of $F$ under composition by $\posG{F}$.
\begin{definition}[Generator set]
\label{def:generatorSet}
We call $\transOs$ a \emph{generator set of $\SSP$} if the following properties are satisfied:
\begin{enumerate}[(i)]
    \item \label{item:self_inv_commutativity_prop} For all ${P, T\in \posG{\transOs}}$ it holds ${PT=TP}$, i.e., operators in $\posG{\transOs}$ commute.
    \item \label{item:generatingAllStates} For all ${\ket{\psi}, \ket{\phi} \in \SSP}$ there exists an operator ${P\in \posG{\transOs}}$ such that
    ${P\ket{\psi} = \ket{\phi}}$.
    \item \label{item:local_invariance} (local invariance) Let ${P\in\transOs}$, ${A\subset[n]}$ be nonempty and suppose that there exists ${\ket{\psi}\in \SSP}$ such that ${\left(P\ket{\psi}\right)_{A} = \ket{\psi_A}}$. Then ${\left(P\ket{\phi}\right)_{A} = \ket{\phi_A}}$ for all ${\ket{\phi}\in \SSP}$.
\end{enumerate}
Moreover, we call each element in ${\posG{\transOs}}$ a product operator over $\transOs$.
\end{definition}
Next, we introduce so-called bipartite operator sets, which will define the operator-based condition given on the left-hand side in Eq.~\eqref{eq:operator_based_cond}.
\begin{definition}[Bipartite operator set]
\label{def:biOpset}
Let ${A\subset [n]}$. If $A$ is nonempty, we define
    \begin{equation*}
        \mathcal{O}_A \coloneqq \set{P\in \posG{\transOs}\mid \forall \ket{\psi}\in \SSP\colon \left(P\ket{\psi}\right)_{A} = \ket{\psi_A}}
    \end{equation*}
    and for ${A=\emptyset}$ we set ${\mathcal{O}_{\emptyset}\coloneqq \posG{\transOs}}$. We call $\mathcal{O}_A$ \emph{operator set of the subsystem $A$} and ${\set{\mathcal{O}_{A}, \mathcal{O}_{A^c}}}$ a \emph{bipartite operator set}.
\end{definition}
Finally, we present the main result of this section, showing that the operator-based condition is equivalent to $\sim_{\SSP}$.
\begin{theorem}
\label{theorem:equivalenzOperatorbipartitions}
Let $\transOs$ be a generator set of $\SSP$ and ${A_1,A_2 \subset [n]}$. Then, we have
\begin{equation}
\label{eq:operator_based_cond}
\set{\mathcal{O}_{A_1},\mathcal{O}_{A_1^c}} = \set{\mathcal{O}_{A_2},\mathcal{O}_{A_2^c}}
        \Longleftrightarrow
        A_1\sim_{\SSP} A_2.
\end{equation}
\end{theorem}
A proof of Theorem~\ref{theorem:equivalenzOperatorbipartitions} is given in Sec.~\ref{subsection:methods:applicationEmbedding}. As we observe in Eq.~\eqref{eq:operator_based_cond}, the operator-based condition has the identical structure as the original one, which defines the equivalence relation in Definition~\ref{def:equivalenzrelation}, except that the quotient sets are replaced by their corresponding operator sets. The reason is that for a subsystem $A\subset [n]$ the operator set $\mathcal{O}_A$ uniquely determines the quotient set $\SSP/{\sim_A}$ and hence, relating operator sets (and the resulting bipartite operator sets) to the equivalence relation $\sim_{A}$. This relation is stated in Proposition~\ref{prop:quotient_set_equality}. Moreover, supposing the existence of a generator set of $\SSP$ which, in addition, is \emph{pointwise-disjoint}, i.e.,  for all ${P,T\in\posG{\transOs}}$ and some ${\ket{\psi}\in\SSP}$ with ${P\ket{\psi}=T\ket{\psi}}$ it follows ${P=T}$, and where each element is self-inverse, will also show that the number of elements of each equivalence class in $\SSP/{\sim_A}$ is constant and corresponds to $\abs{\mathcal{O}_A}$, highlighting another feature of the operator-based formulation.
\subsection{Embedding examples\label{sec:embeddingExamples}}
In this section we demonstrate how the preceding results can be applied to two embeddings used in quantum optimization.

The first example we examine in more detail is the parity embedding, which we introduced in Sec.~\ref{sec:parityembedding}. As already mentioned, the penalty terms~\eqref{eq:parityConstraintH} of the parity Hamiltonian~\eqref{eq:parityHamiltonian} create an energy gap that isolates a subset of physical states from the remaining Hilbert space. We define this subset of physical states as $\Pi$, the set of all \emph{parity states}~(Definition~\ref{def:parityStateSpace}). For the parity embedding, we choose the set of all \emph{logical line operators}~(Definition~\ref{def:logicalLineOperatorProducts}), denoted by $\mathbb{\Lambda}$, as a generator set where each logical line operator $\Lambda_v$ maps the change of the state of a single logical qubit (weight of a given vertex $v\in V$) to the change of the states of the respective set of physical qubits (weight of all edges containing the vertex $v$). In Definition~\ref{def:logicalLine} we introduce this set of physical qubits as a \emph{logical line of $v$} and denote it by $L_v(H)$. The reason why this set is a natural choice and serves as a candidate for a generator set lies in the main characteristics of the parity embedding: Each parity state defined by a parity weight on $E$ can be mapped back to its corresponding logical state defined by some weight on $V$. Thus, given two different parity states and their corresponding logical states, applying the logical line operators of all vertices where the state of the qubit differs between the two logical states transforms one parity state into the other. As logical line operators change the states of qubits independently of their current state, $\mathbb{\Lambda}$ forms a generator set~(for more details, see Lemma~\ref{lemma:generatorSetParityEmbedding}).

Recalling Definition~\ref{def:biOpset}, we observe that an operator set on a subsystem $A$ covers all operators in $\posG{\mathbb{\Lambda}}$ which leave all states restricted on $A$ unchanged. In the next lemma, we derive generators for such operator sets, hence simplifying the verification whether two operator sets on two subsystems are equal. Moreover, it enables a systematic way of constructing operator sets of $A$ and thus, determining all equivalence classes with respect to $\sim_A$ on $\Pi$~(see also Proposition~\ref{prop:quotient_set_equality}).
\begin{lemma}
\label{lemma:constructBipOperatorSetsForAllBips}
Let ${H=(V,E)}$ be a hypergraph and ${A\subset E}$. Furthermore, let ${V_1,\ldots,V_N\subset V}$ be the connected components of ${H\vert_A=(V_A,E_A)}$ and ${H_k\coloneqq H\vert_{E_k}}$ the induced sub-hypergraphs of ${H\vert_{A}}$ where ${E_k\coloneqq\set{e\in E\mid e\subset V_k}}$ for ${k=1,\ldots,N}$. Define
\begin{equation*}
    \mathcal{U}_k \coloneqq \set{\left.\prod_{v\in \tilde{V}} \Lambda_{v}\,\right\vert \tilde{V}\in\mathcal{V}_k},
\end{equation*}
where
\begin{equation}
    \label{eq:generators_bip_op_sets}
    \mathcal{V}_k\coloneqq\set{\tilde{V}\subset V_k\left\vert\, \tilde{V}\neq\emptyset\land\forall e\in E_k\colon \abs{\tilde{V}\cap e}\text{ even}\right.}.
\end{equation}
Then, we have
    $$
    \mathcal{O}_{A}= \spos\left(\mathcal{U} \cup \mathcal{Q}\right),
    $$
    where
    $$
    \mathcal{U}\coloneqq\bigcup_{k=1}^N\mathcal{U}_k
    $$
    and
    $$
    \mathcal{Q} \coloneqq \set{\Lambda_v\in \mathbb{\Lambda} \left\vert\, v\in V\setminus\bigcup_{k=1}^{N}V_k\right.}.
    $$
    Moreover, if $H$ is a graph, then ${\mathcal{V}_k=\set{V_k}}$ for ${k=1,\ldots,N}$.
\end{lemma}
By Lemma~\ref{lemma:constructBipOperatorSetsForAllBips}, comparing all elements of two operator sets of two subsystems ${A_1,A_2}$ is equivalent to verifying whether all generators of one operator set is contained in the other and vice versa. Thus, by Theorem~\ref{theorem:equivalenzOperatorbipartitions}, we can derive for the verification of ${A_1\sim_{\Pi} A_2}$ the following algorithm: Let ${V=\set{v_1,\ldots,v_n}}$ and ${E=\set{e_1,\ldots,e_n}}$. Moreover, for ${A\subset E}$ let ${v_{k,1}^{(A)},\ldots,v_{k,m_k}^{(A)}\in\mathbb{Z}_2^{n}}$ with ${m_k\leq \abs{V_k}}$ be a basis of $\mathcal{V}_k$, where we identified subsets ${W\subset V}$ by elements in $\mathbb{Z}_2^{n}$ via ${b_W=(b_1,\ldots,b_{n})}$ with ${b_i=1}$ if ${v_i\in W}$ and ${b_i=0}$ otherwise for all ${i=1,\ldots,n}$. Note that a basis of $\mathcal{V}_k$ can be determined using ${O(\abs{V_k}\abs{E_k}\min(\abs{V_k},\abs{E_k}))}$ calculations by applying Gaussian elimination on the matrix ${B^{(k)}\in\mathbb{Z}_2^{\abs{E_k}\times \abs{V_k}}}$ defined by
\begin{equation*}
    B_{i,j}^{(k)}\coloneqq\begin{cases}
        1,\quad &w_{k,j}\in f_{k,i},\\
        0,\quad &\text{else,}\\
    \end{cases}
\end{equation*}
where we use the enumerations ${V_k=\set{w_{k,1},\ldots,w_{k,\abs{V_k}}}}$ and ${E_k=\set{f_{k,1},\ldots,f_{k,\abs{E_k}}}}$.
Hence, applying Lemma~\ref{lem:prop_product_operators} and using the properties of product operators, we deduce from~\eqref{eq:generators_bip_op_sets} that $\Lambda_W \in \mathcal{O}_{A_1}$ for some $W\subset V$, where ${\Lambda_W\coloneqq\prod_{w\in W}\Lambda_w}$, if and only if the linear system with coefficient matrix
\begin{equation*}
B_A\coloneqq\left[v_{1,1}^{(A)},\ldots,v_{1,m_1}^{(A)},v_{2,1}^{(A)},\ldots,,v_{N,m_N}^{(A)},\delta_{i_1},\ldots,\delta_{i_M}\right]
\end{equation*}
 and right-hand side $b_W$ has a solution over $\mathbb{Z}_2$. Here ${M\coloneqq\abs{V\setminus\bigcup_{k=1}^{N}V_k}}$, ${i_1,\ldots,i_M}$ are chosen such that
 $$
 \set{v_{i_1},\ldots,v_{i_M}}=V\setminus\bigcup_{k=1}^{N}V_k
 $$
 and $\delta_{i_l}$ denotes the $i_l$th standard basis vector for ${l=1,\ldots,M}$. This implies that the condition ${\mathcal{O}_{A_1}\subset \mathcal{O}_{A_2}}$ can be verified by applying Gaussian elimination on the extended coefficient matrix ${[B_{A_1}\mid B_{A_2}]}$. Since the number of columns of each coefficient matrix is less than or equal to $n$, this verification needs ${O(n^3)}$ computations. Using that calculating a basis of a $\mathcal{V}_k$ for all ${k=1,\ldots,N}$ requires in total ${O(\abs{V_A}\abs{E_A}\min(\abs{V_A},\abs{E_A}))}$ computations, we deduce that verifying ${A_1\sim_{\Pi} A_2}$ yields a total cost of ${O(\abs{V}\abs{E}\min(\abs{V},\abs{E}))}$. As the dimension of the parity state space $\Pi$ is polynomial in $\abs{V}$ and $\abs{E}$~(see~\cite{compilerpaper}), and thus the cardinality of $\Pi$ is exponential, this reduces the cost of verifying ${A_1\sim_{\Pi} A_2}$ from exponential time to polynomial time.

Aside from the derivation of this algorithm, Lemma~\ref{lemma:constructBipOperatorSetsForAllBips} can also be used to establish the existence of bundles in $\Pi$ containing multiple bipartitions. As an example, we investigate bundles that are equivalent to logical lines $L_v(H)$ for $v\in V$ in a complete graph ${H=(V,E)}$ with ${\abs{V}>3}$. For a subsystem $A\subset L_v(H)^c$ with ${\abs{A}>\abs{L_v(H)^c}-(\abs{V}-2)}$ we observe that the restricted graph ${H\vert_{A}}$ is connected and has the vertex set $V\setminus\set{v}$, since removing less than ${\abs{V}-2}$ edges from ${H\vert_{L_v(H)^c}}$ does not split ${H\vert_{L_v(H)^c}}$ into two disconnected graphs and shrink the vertex set of ${H\vert_{L_v(H)^c}}$. Therefore, using Lemma~\ref{lemma:constructBipOperatorSetsForAllBips} and the identity
 $$
 {\prod_{u\in V\setminus \set{v}}\Lambda_u = \Lambda_v},
 $$
 we obtain ${\mathcal{O}_A=\set{\mathbb{1}, \Lambda_v}}$. Moreover, since $L_v(H)\subset A^c$, we have ${\mathcal{O}_{A^c}=\set{\mathbb{1}}}$, which, by Theorem~\ref{theorem:equivalenzOperatorbipartitions}, shows that ${A\in [L_v(H)]_{\Pi}}$ and thus, $[L_v(H)]_{\Pi}$ has at least
 \begin{equation}
\label{eq:fullBisSpanningClass}
    \sum_{p=0}^{\abs{V}-3} \binom{\abs{L_v(H)^c}}{\abs{L_v(H)^c}-p}
\end{equation}
bipartitions. To determine whether a subset belongs to the same class $[L_{v}(H)]_{\Pi}$, it is required verifying that its associated graph preserves the same connectivity structure,  which means that it still contains the same set of vertices as ${H\vert_{L_v(H)^c}}$. We also want to remark that from Proposition~\ref{prop:quotient_set_equality} it follows that the parity state space can be partitioned into pairs of parity states which are identical in their weights for all edges in $L_v(H)^c$.

For a complete graph, there are $\abs{V}$ different bipartition classes $[L_v(H)]_{\Pi}$ of this type, with one class associated to each ${v\in V}$.
In Fig.~\ref{fig:GroupEntanglementDynamics}~{(b)}, the orange data points correspond to the von Neumann entropy of the bundle $[L_3(H)]_{\Pi}$ with ${v=3}$ and ${L_{3}(H) = \set{\set{0,3},\set{1,3},\set{2,3},\set{3,4}}}$, where $H$ is the complete graph on five vertices. Therefore, we can deduce that the orange colored subsystems in Fig.~\ref{fig:GroupEntanglementDynamics}~{(d)} ${A=L_3^c(H)}$, ${A_1=\set{\set{0,2},\set{2,4},\set{1,4}}}$, ${A_2=\set{\set{0,2},\set{2,4},\set{1,4},\set{0,1}}}$ and ${A_3=\set{\set{0,2},\set{2,4},\set{1,2}}}$ belong to the equivalence class $[L_{3}(H)]_{\Pi}$ and bundle to the same final entanglement value in Fig.~\ref{fig:GroupEntanglementDynamics}~{(b)}, while the violet and green-blue subsystems ${B_1=\set{\set{0,1},\set{1,4},\set{0,4}}}$, ${B_2=\set{\set{0,1},\set{1,4}}}$ and ${D=\set{\set{0,2},
\set{2,4},\set{0,4}}}$ belong to other classes and do not bundle to the same final entanglement value as ${L_3(H)}$.
Furthermore, we now can prove that the total number of unordered bipartitions contributing to this bundle is $38$.
With Eq.~\eqref{eq:fullBisSpanningClass} we have for all subsystems of size ${3<\abs{A}\leq 6}$ in total $22$ unordered bipartitions. For all subsets ${A\subset L_3(H)^c}$ with ${\abs{A} = 3}$ we count in total $16$ for which the corresponding subgraph ${H\vert_A = (V\setminus \set{3}, A)}$ is connected and has the same vertex set. For all smaller subsets it is not possible to find any connected subgraph containing all four vertices and thus we have in total $38$ unordered bipartitions contributing to this bundle.

Finally, we demonstrate that, analogous to the parity embedding, a generator set also exists for the minor embedding and show the existence of multiple bipartitions yielding the same entanglement spectrum.
Similar to the parity embedding, we can define a physical space $\mathcal{H}(V_M)$ and a minor state space ${\MES\subset \mathcal{H}(V_H)}$  for a minor graph ${M=(V_M,E_M)}$ and chains $C_v$ for ${v\in V}$~(Definition~\ref{def:minorState}). Furthermore, analogous to the logical line operator in the parity embedding, we can define the \emph{chain operator}~(Definition~\ref{def:chainOperator}) $\Gamma_{v}$ of a vertex ${v\in V}$ for the minor embedding, which changes the weights of all vertices ${w\in V_M}$ that are contained in the chain $C_v$. From this, it can be shown that the set of all chain operators $\mathbb{\Gamma}$ is a generator set of the minor state space $\MES$.
Analogous to the parity embedding, we can determine a bipartite operator set for each bipartition and use it to verify whether two bipartitions belong to the same class by applying Theorem~\ref{theorem:equivalenzOperatorbipartitions}. To conclude, we show by Theorem~\ref{theorem:equivalenzOperatorbipartitions} the existence of multiple bipartitions which are equivalent under $\sim_{\MES}$. In the following, we assume that there exists a chain $C_v$ in the minor embedding that has at least length ${\abs{C_v}\geq 3}$. Then, we observe that for all ${A\subset C_v}$ with ${\abs{A}<\abs{C_v}}$, the bipartite operator sets are identical and given by ${\mathcal{O}_A=\mathcal{G}(\set{\Gamma_w\mid w \in V\setminus\set{v}})}$ and ${\mathcal{O}_{A^c}  = \set{\mathbb{1}}}$.
To demonstrate another example, we consider subsystems ${A\subset V_M}$ consisting of exactly 
$\abs{V}$ elements, where each element in $A$ associated with a different chain, i.e., ${\abs{A\cap C_v} =  1}$ for all ${v\in V_M}$. Again, all bipartite operator sets are equal and correspond to $\set{\set{\mathbb{1}},\set{\mathbb{1}}}$.
\section{Methods\label{sec:methods}}
In the final section of this article, we introduce the remaining mathematical tools and deliver the proofs of the theoretical claims established in the earlier sections.
\subsection{Bundling of spectral-based entanglement measures\label{subsec:bundl_spectr-based_entgl_measures}}
First, we present a full proof of our main result presented in Theorem~\ref{thm:suff_cond_equivalence_spectrum_subspace_main} which will be refined in Theorem~\ref{thm:equivalence_spectrum_sim_subspace}. In the following we denote by
\begin{multline*}
\label{eq:subsystemmixedstateset}
 \MS{\SSP}\coloneqq\Bigg\{\left.\sum_{i=1}^m p_{i} \ket{\psi_{i}}\bra{\psi_{i} }\,\right\vert m\in\mathbb{N},\\ \forall i\in[m]\colon \ket{\psi_{i}}\in\QS{\SSP}\text{ and }p_{i} \geq 0,\, \sum_{i=1}^m p_{i} = 1\Bigg\}
\end{multline*}
the set of mixed states whose pure components lie in the subspace $\QS{\SSP}$.
\begin{theorem}\label{thm:equivalence_spectrum_sim_subspace}
    Let ${A_1,A_2\subset [n]}$ be two nonempty subsystems with ${A_1,A_2\neq [n]}$ satisfying ${A_1\sim_{\SSP} A_2}$. Then, for all ${\rho\in  \MS{\SSP}}$ we have that ${\Spec(\rho_{A_1})=\Spec(\rho_{A_2})}$. Conversely, there exists ${\SSP\subset\CBS{n}}$ of arbitrarily large size greater than three with growing ${n\in\mathbb{N}}$ and ${A_1,A_2\subset [n]}$ with ${A_1\not\sim_{\SSP} A_2}$ such that ${\Spec(\rho_{A_1})=\Spec(\rho_{A_2})}$ for all  pure states $\rho$ in ${\MS{\SSP}}$.
\end{theorem}
In order to prove Theorem~\ref{thm:equivalence_spectrum_sim_subspace}, we need the following two lemmas.
\begin{lemma}
    \label{lem:equal_spectrum}
    Let ${A_1,A_2\subset\set{1,\ldots,n}}$ be two nonempty sets. Furthermore, let ${S\subset \CBS{n}}$ such that for all
    $$
    \forall \ket{\psi},\ket{\phi}\in S\colon \ket{\psi_{A_1}}=\ket{\phi_{A_1}}\Longleftrightarrow \ket{\psi_{A_2}}=\ket{\phi_{A_2}}.
    $$
    Moreover, let ${d_{\ket{\psi},\ket{\phi}}\in\C}$ for ${\ket{\psi},\ket{\phi}\in S}$. Then the spectrum of the operators
    \begin{equation*}
    \rho_1\coloneqq\sum_{\ket{\psi},\ket{\phi}\in S}d_{\ket{\psi},\ket{\phi}}\ket{\psi_{A_1}}\bra{\phi_{A_1}}
    \end{equation*}
    and
    \begin{equation*}
        \rho_2\coloneqq\sum_{\ket{\psi},\ket{\phi}\in S}d_{\ket{\psi},\ket{\phi}}\ket{\psi_{A_2}}\bra{\phi_{A_2}}
    \end{equation*}
    are equal.
\end{lemma}
\begin{proof}
    In the following we choose ${\ket{\psi_i}\in S}$ for ${i=1,\ldots,m}$ such that
    $$
S/{\sim_{A_1}}=\set{[\ket{\psi_1}]_{A_1},\ldots,[\ket{\psi_m}]_{A_1}}
    $$
where ${m\coloneqq \abs{S/{\sim_{A_1}}}}$.
    \begin{enumerate}[label=(\roman*),wide=\parindent,leftmargin=0pt,align=left]
        \item \label{item:proof_spec_rel_step_1} First, we observe that
        \begin{align}
            \rho_1&=\sum_{1\leq i,j\leq m}\sum_{\substack{\ket{\psi}\in [\ket{\psi_i}]_{A_1}\\\ket{\phi}\in[\ket{\psi_j}]_{A_1}}}d_{\ket{\psi},\ket{\phi}}\ket{\psi_{A_1}}\bra{\phi_{A_1}}\nonumber\\
            &=\sum_{1\leq i,j\leq m}d_{i,j}\ket{{\psi_i}}_{A_1}\bra{{\phi_j}}_{A_1}\label{eq:matrix_formula_spec_proof},
        \end{align}
        where
        \begin{equation*}
            d_{i,j}\coloneqq\sum_{\substack{\ket{\psi}\in [\ket{\psi_i}]_{A_1}\\ \ket{\phi}\in[\ket{\psi_j}]_{A_1}}}d_{\ket{\psi},\ket{\phi}}
        \end{equation*}
        for ${1\leq i,j\leq m}$.
        \item \label{item:proof_spec_rel_step_2} Since ${\ket{{\psi_1}}_{A_1},\ldots,\ket{{\psi_m}}_{A_1}}$ are linearly independent in ${\CBS{\abs{A_1}}}$, there exist ${\ket{\psi_{m+1}},\ldots,\ket{\psi_{2^{\abs{A_1}}}}\in \CBS{{\abs{A_1}}}}$ such that \begin{equation}
            \label{eq:basis_spec_proof}
            \set{\ket{{\psi_1}}_{A_1},\ldots,\ket{{\psi_m}}_{A_1},\ket{\psi_{m+1}},\ldots,\ket{\psi_{2^{\abs{A_1}}}}}
        \end{equation}
        constitute an orthonormal basis in $\CBS{{\abs{A_1}}}$. From~\eqref{eq:matrix_formula_spec_proof} it can be easily deduced that
        \begin{equation*}
            \begin{pmatrix}
                M & 0\\
                0 & 0
            \end{pmatrix}\in \mathbb{C}^{2^{\abs{A_1}}\times 2^{\abs{A_1}}},\ M\coloneqq                 \begin{pmatrix}
d_{1,1} & d_{1,2} & \cdots & d_{1,m} \\
d_{2,1} & d_{2,2} & \cdots & d_{2,m} \\
\vdots  & \vdots  & \ddots & \vdots  \\
d_{m,1} & d_{m,2} & \cdots & d_{m,m}
                \end{pmatrix},
        \end{equation*}
        is the matrix of $\rho_1$ with respect to the basis~\eqref{eq:basis_spec_proof}. Therefore
        \begin{equation}
            \label{eq:spec_rel_spec_proof}
            \mathrm{Spec}(\rho_1)=\mathrm{Spec}(M).
        \end{equation}
        \item Finally, by our assumption, we observe that ${S/{\sim_{A_1}}=S/{\sim_{A_2}}}$ and hence, by repeating the steps~\ref{item:proof_spec_rel_step_1} and~\ref{item:proof_spec_rel_step_2} where $A_1$ is replaced by $A_2$, we deduce again ${\mathrm{Spec}(\rho_2)=\mathrm{Spec}(M)}$, which together with~\eqref{eq:spec_rel_spec_proof} proves the desired statement.\qedhere
    \end{enumerate}
\end{proof}
\begin{lemma}
    \label{lemma:notEquivalentBipsEqualSpec}
Suppose ${\SSP\subset \CBS{n}}$ and ${A_1, A_2\in\mathcal{P}([n])}$ such that
\begin{equation}
    \label{eq:rel_quotient_set_not_sim}
    \begin{aligned}
        \SSP/{\sim_{A_1}} &= \set{B, D},\\
    \SSP/{\sim_{A_1^c}} &= \bigcup_{\ket{\psi}\in B\cup D}\set{\set{\ket{\psi}}},\\ \SSP/{\sim_{A_2}} &= \set{B}\cup \big(\bigcup_{\ket{\psi}\in D}\set{\set{\ket{\psi}}}\big),\\
    \SSP/{\sim_{A_2^c}} &= \big(\bigcup_{\ket{\psi}\in B}\set{\set{\ket{\psi}}}\big)\cup\set{D},
    \end{aligned}
\end{equation}
where ${B,D\subset \SSP}$ with $B\cap D=\emptyset$ and ${\abs{D}=2}$.
Then ${A_1\not\sim_{\SSP} A_2}$ and for all $\ket{\Psi}\in\QS{\SSP}$ the entanglement spectra of the pure state $\ket{\Psi}$ with respect to  $A_1$ and $A_2$ are equal. The size of $\SSP$ is limited by
\begin{equation}
\label{eq:limitedSizeR}
    4\leq\abs{\SSP} \leq 2^{n - \abs{A_1\cup A_2}} + 2.
\end{equation}
\end{lemma}
\begin{proof}
    \begin{enumerate}[label=(\roman*),wide=\parindent,leftmargin=0pt,align=left]
        \item First, we show that the spectra are equal. In the following we denote by $\ket{\psi_1}$ some element in $B$ and by $\ket{\psi_2}$ and $\ket{\psi_3}$ the two distinct elements in $D$. Moreover, let ${\ket{\Psi}\in \QS{\SSP}}$ with
        $$
        \ket{\Psi}=\sum_{\ket{\psi}\in\SSP} c_{\ket{\psi}}\ket{\psi}
        $$
        with ${c_{\ket{\psi}}\in\C}$ for ${\ket{\psi}\in\SSP}$. Using Eq.~\eqref{eq:reduced_density_matrix_formula} and the relations of the quotients sets as depicted in~\eqref{eq:rel_quotient_set_not_sim}, it can be deduced that reduced density matrices of the subsystems $A_1$ and $A_2$ correspond to
            \begin{equation*}
        \rho_{1} =\begin{pmatrix}
                M_1 & 0\\
                0 & 0
            \end{pmatrix}\in \mathbb{C}^{2^{\abs{A_1}}\times 2^{\abs{A_1}}}
    \end{equation*}
    and
    \begin{equation*}
        \rho_{2}\coloneqq\begin{pmatrix}
                M_2 & 0\\
                0 & 0
            \end{pmatrix}\in \mathbb{C}^{2^{\abs{A_2}}\times 2^{\abs{A_2}}},
    \end{equation*}
    where
    \begin{equation*}
        M_1\coloneqq
\begin{pmatrix}
\sum_{\ket{\psi}\in B}\abs{c_{\ket{\psi}}}^2 & 0 \\
0 & \abs{c_{\ket{\psi_2}}}^2+\abs{c_{\ket{\psi_3}}}^2
\end{pmatrix}
    \end{equation*}
    and
    \begin{equation*}
       M_2=\begin{pmatrix}
\sum_{\ket{\psi}\in B}\abs{c_{\ket{\psi}}}^2 & 0 & 0 \\
0 & \abs{c_{\ket{\psi_2}}}^2 & c_{\ket{\psi_2}}c_{\ket{\psi_3}}^* \\
0 & c_{\ket{\psi_2}}^*c_{\ket{\psi_3}} & \abs{c_{\ket{\psi_3}}}^2 
\end{pmatrix},
    \end{equation*}
    where have chosen orthonormal bases such that the first two basis vectors in  $\mathbb{C}^{2^{\abs{A_1}}}$ are ${\ket{\psi_{1}}_{A_1}}$ and ${\ket{\psi_{2}}_{A_1}}$, and the first three basis vectors in $\mathbb{C}^{2^{\abs{A_2}}}$ are ${\ket{\psi_{1}}_{A_2}}$, ${\ket{\psi_{2}}_{A_2}}$ and ${\ket{\psi_{3}}_{A_2}}$. Thus, we see that the spectra of both matrices equal $$\set{0,\sum_{\ket{\psi}\in B}\abs{c_{\ket{\psi}}}^2,\abs{c_{\ket{\psi_2}}}^2+\abs{c_{\ket{\psi_3}}}^2}.$$
    \item Next, we prove the bounds in Eq.~\eqref{eq:limitedSizeR}. As a first step, we derive the upper bound in Eq.~\eqref{eq:limitedSizeR}. Without loss of generality, suppose ${B\neq\emptyset}$ and choose ${\ket{\psi}\in B}$. Thus, we have for all ${\ket{\phi}\in B}$ that ${\ket{\phi}_{A_1}=\ket{\psi}_{A_1}}$ and ${\ket{\phi}_{A_2}=\ket{\psi}_{A_2}}$, and therefore, ${B\subset\set{\ket{\phi}\in \CBS{n}\mid \ket{\phi}_{A_1\cup A_2}=\ket{\psi}_{A_1\cup A_2}}}$. This implies ${\abs{B}\leq 2^{n-\abs{A_1\cup A_2}}}$, which, by using ${\abs{D}=2}$, shows the upper bound in Eq.~\eqref{eq:limitedSizeR}.
            Now, we prove the lower bound of $\abs{\SSP}$. For ${\abs{B}=\emptyset}$ we deduce
\begin{align*}
    \SSP/{\sim_{A_2}} &= \SSP/{\sim_{A_1^c}} = \bigcup_{\ket{\psi}\in D}\set{\set{\ket{\psi}}},\\
    \SSP/{\sim_{A_2^c}} &= \SSP/{\sim_{A_1}} = \set{D},
\end{align*}
and hence ${A_1\sim_{\SSP}A_2}$. For ${\abs{B}=1}$ we similarly deduce
\begin{align*}
    \SSP/{\sim_{A_2}} &= \SSP/{\sim_{A_1^c}} = \set{B}\cup \bigcup_{\ket{\psi}\in D}\set{\set{\ket{\psi}}},\\
    \SSP/{\sim_{A_2^c}} &= \SSP/{\sim_{A_1}} = \set{B, D},
\end{align*}
which shows the desired statement.\qedhere
\end{enumerate}
\end{proof}
\begin{proof}[Proof of Theorem~\ref{thm:equivalence_spectrum_sim_subspace}]
    \begin{enumerate}[label=(\roman*),wide=\parindent,leftmargin=0pt,align=left]
    \item 
    Without loss of generality we may assume that ${\SSP/{\sim_{A_1}}=\SSP/{\sim_{A_2}}}$. Now, let ${\ket{\Psi}\in\QS{\SSP}}$ and ${c_{\ket{\psi}}\in\C}$ for ${\ket{\psi}\in\SSP}$ such that
    \begin{equation*}
        \ket{\Psi}=\sum_{\ket{\psi}\in\SSP}c_{\ket{\psi}}\ket{\psi}.
    \end{equation*}
    Furthermore, we choose ${\ket{\psi_i}\in \SSP}$ for ${i=1,\ldots,m}$ such that ${\SSP/{\sim_{A_1}}=\set{[\ket{\psi_i}]_{A_1},\ldots,[\ket{\psi_m}]_{A_1}}}$ where ${m\coloneqq \abs{\SSP/{\sim_{A_1}}}}$. Using that ${\SSP/{\sim_{A_1}}}$ is a partition of $\SSP$, we observe that
    \begin{equation*}
        \ket{\Psi}=\sum_{i=1}^m\sum_{\ket{\phi}\in[\ket{\psi_i}]_{A_1}}c_{\ket{\phi}}\ket{\phi},
    \end{equation*}
    Thus, we obtain
    \begin{equation*}
        \rho=\sum_{i,j=1}^m\sum_{\substack{\ket{\phi}\in [\ket{\psi_i}]_{A_1}\\\ket{\chi}\in[\ket{\psi_j}]_{A_1}}}c_{\ket{\phi}}c_{\ket{\chi}}^*\ket{\phi}\bra{\chi}.
    \end{equation*}
    Next, we apply the properties of the partial trace to deduce
    \begin{align}
        \rho_{A_1}&=\sum_{i,j=1}^m\sum_{\substack{\ket{\phi}\in [\ket{\psi_i}]_{A_1}\\\ket{\chi}\in[\ket{\psi_j}]_{A_1}}}c_{\ket{\phi}}c_{\ket{\chi}}^*\braket{\phi_{A_1^c}}{\chi_{A_1^c}}\ket{\phi_{A_1}}\bra{\chi_{A_1}}\nonumber\\
        &=\sum_{i,j=1}^m\Bigg(\sum_{\substack{\ket{\phi}\in [\ket{\psi_i}]_{A_1}\\\ket{\chi}\in[\ket{\psi_j}]_{A_1}}}c_{\ket{\phi}}c_{\ket{\chi}}^*\braket{\phi_{A_1^c}}{\chi_{A_1^c}}\Bigg)\ket{{\phi_i}}_{A_1}{\bra{\phi_j}_{A_1}}\nonumber\\
        &=\sum_{i,j=1}^m d_{i,j}\ket{{\phi_i}}_{A_1}{\bra{\phi_j}}_{A_1},\label{eq:reduced_density_matrix_formula}
    \end{align}
    where
    \begin{equation}
    \label{eq:coefficientsRedDens}
    d_{i,j}\coloneqq \sum_{\substack{\ket{\phi}\in [\ket{\psi_i}]_{A_1}\\\ket{\chi}\in[\ket{\psi_j}]_{A_1}}} c_{\ket{\phi}}c_{\ket{\chi}}^*\braket{\phi_{A_1^c}}{\chi_{A_1^c}}.
    \end{equation}
    Since $A_1\sim_{\SSP} A_2$ we have that $\braket{\phi_{A_1^c}}{\chi_{A_1^c}}=0$ if and only if $\braket{\phi_{A_2^c}}{\chi_{A_2^c}}=0$ for $\ket{\phi}\in [\ket{\psi_i}]_{A_1}$ and $\ket{\chi}\in [\ket{\psi_j}]_{A_1}$. This implies that $$\rho_{A_2}=\sum_{i,j=1}^m d_{i,j}\ket{{\phi_i}}_{A_2}\bra{{\phi_j}}_{A_2}$$ which together with Lemma~\ref{lem:equal_spectrum} shows the desired statement for pure states. Note that the logic of the proof for mixed states remains identical. As the partial trace is linear, the only difference is that calculating the coefficients $d_{i,j}$ in Eq.~\eqref{eq:coefficientsRedDens} requires an additional sum over all components of the mixed state.
        \item To prove the second statement, we first define
    \begin{align*}
        \ket{\psi_{4,1}}\coloneqq\ket{0000},\quad\ket{\psi_{4,2}}\coloneqq\ket{0001},\\
        \ket{\psi_{4,3}}\coloneqq\ket{1010},\quad\ket{\psi_{4,4}}\coloneqq\ket{1100}
    \end{align*}
    as well as ${\SSP_4\coloneqq\{\ket{\psi_{4,1}},\ket{\psi_{4,2}},\ket{\psi_{4,3}},\ket{\psi_{4,4}}\}\subset\CBS{4}}$. For ${n>4}$, we denote by $\ket{\psi_{n,1}}$ the state in $\CBS{n}$ which satisfies ${\ket{\psi_{n,1}}_{\set{1,\ldots,n-1}}=\ket{0}^{\otimes(n-1)}}$ and ${\ket{\psi_{n,1}}_{\set{n}}=\ket{1}}$. Furthermore, we recursively set ${\ket{\psi_{n,i}}}$ as the state in $\CBS{n}$ defined by $$\ket{\psi_{n,i}}_{\set{1,\ldots,n-1}}=\ket{\psi_{n-1,i-1}}\quad\text{and}\quad \ket{\psi_{n,i}}_{\set{n}}=\ket{0}$$ for ${i=2,\ldots,n}$ and $${\SSP_n\coloneqq\{\ket{\psi_{n,1}},\ket{\psi_{n,2}},\ldots,\ket{\psi_{n,n}}\}\subset\CBS{n}}.$$ Next, we show for each ${n\geq 4}$ that ${A_1\coloneqq\set{1}}$ and ${A_2\coloneqq\set{2,3}}$ are the right candidates for proving our statement. By Lemma~\ref{lemma:notEquivalentBipsEqualSpec}, we are left to show that
        \begin{align*}
            &\SSP_n/{\sim_{A_1}}=\set{\set{\ket{\psi_{n,1}},\ldots,\ket{\psi_{n,n-2}}},\set{\ket{\psi_{n,n-1}},\ket{\psi_{n,n}}}},\\
            &\SSP_n/{\sim_{A_1^c}}=\set{\set{\ket{\psi_{n,1}}},\ldots,\set{\ket{\psi_{n,n}}}},\\
            &\SSP_n/{\sim_{A_2}}=\set{\set{\ket{\psi_{n,1}},\ldots,\ket{\psi_{n,n-2}}},\set{\ket{\psi_{n,n-1}}},\set{\ket{\psi_{n,n}}}},\\
            &\SSP_n/{\sim_{A_2^c}}=\set{\set{\ket{\psi_{n,1}}},\ldots,\set{\ket{\psi_{n,n-2}}},\set{\ket{\psi_{n,n-1}},\ket{\psi_{n,n}}}}.
        \end{align*}
        For $n=4$, the relation can be easily verified. Now, suppose that the statement holds for some $n\geq 4$. By the recursive definition of states ${\ket{\psi_{n,1}},\ldots,\ket{\psi_{n,n}}\in\SSP_n}$, we observe that
        \begin{equation*}
            \ket{\psi_{n+1,i}}\sim_{B}\ket{\psi_{n+1,j}}\Longleftrightarrow\ket{\psi_{n,i-1}}\sim_{B\setminus\set{n+1}}\ket{\psi_{n,j-1}}
        \end{equation*}
        for all ${i=2,\ldots,n+1}$ and ${B\in\set{A_1,A_1^c,A_2,A_2^c}}$. Thus, by using our induction hypothesis and the definition $\ket{\psi_{n,1}}$, the claim for $n+1$ follows.\qedhere
    \end{enumerate}
\end{proof}

\subsection{Applications: Proofs and further details\label{subsection:methods:applicationEmbedding}}
This section presents proofs of the claims made in Sec.~\ref{subsec:appl_embed-based_quantum_opt} along with further details on the embeddings discussed there. Moreover, we extend these results and present additional theoretical results.
\subsubsection{Proof of Theorem~\ref{theorem:equivalenzOperatorbipartitions}}
In order to prove Theorem~\ref{theorem:equivalenzOperatorbipartitions}, we introduce the following notation: For a nonempty set ${A\subset [n]}$ we set ${\spos_0 \coloneqq\mathcal{O}_{A}}$. Then we choose ${\ket{\psi_0}\in\SSP}$
and for each ${q\in\mathbb{N}}$
$$
P_q\in \posG{\transOs}\setminus\bigcup_{i=0}^{q-1}\spos_i
$$
as long as ${\posG{\transOs}\setminus\bigcup_{i=0}^{q-1}\spos_i\neq\emptyset}$ and define ${\spos_q\coloneqq\set{P_qP\mid P\in\spos_0}}$. Let $m$ be the largest ${q\in\mathbb{N}}$ such that ${\posG{\transOs}\setminus\bigcup_{i=0}^{q-1}\spos_i\neq\emptyset}$. Then we set
$$
\SSP_i\coloneqq\SSP_i(\mathcal{O}_A,\ket{\psi_0})\coloneqq\set{P\ket{\psi_0}\mid P\in\spos_i}
$$
for ${i=0,\ldots,m}$.
\begin{proposition}
    \label{prop:quotient_set_equality}
    It holds
    \begin{equation}
        \label{eq:quotient_set_equality}
        \SSP/{\sim_A}=\set{\SSP_0,\ldots,\SSP_m}.
    \end{equation}
    Thus, each $\SSP_i$ corresponds to an equivalence class of $\sim_A$. Moreover, if all elements in $\transOs$ are self-inverse and $\transOs$ is pointwise-disjoint, then  every equivalence class in ${\SSP/{\sim_A}}$ has the same size and its cardinality corresponds to $\abs{\mathcal{O}_A}$.
\end{proposition}
\begin{proof}
    \begin{enumerate}[label=(\roman*),wide=\parindent,leftmargin=0pt,align=left]
        \item To show Eq.~\eqref{eq:quotient_set_equality}, it is sufficient to verify that ${\SSP_i=[\ket{\psi_i}]_A}$ for all ${0\leq i\leq m}$ where ${\ket{\psi_i}\coloneqq P_i\ket{\psi_0}}$. The remaining statement follows from Definition~\ref{def:generatorSet}~\ref{item:generatingAllStates} and the definition of $m$. Suppose ${\ket{\psi}\in\SSP_i}$. Then, there exists ${P\in\spos_0}$ such that
    \begin{equation*}
        \ket{\psi}=P_iP\ket{\psi_0}=PP_i\ket{\psi_0}=P\ket{\psi_i}
    \end{equation*}
    where we used the commutativity of the product operators in the second equality. Since ${P\in\spos_0}$, we can imply ${\ket{\psi_A}=(P\ket{\psi_i})_A=\ket{\psi_i}_A}$, and hence, ${\SSP_i\subset[\ket{\psi_i}]_A}$. Conversely, suppose that ${\ket{\psi}\in [\ket{\psi_i}]_A}$. Then Definition~\ref{def:generatorSet}~\ref{item:generatingAllStates} implies that there exists ${P\in\posG{\transOs}}$ such that
    \begin{equation}
        \label{eq:rel_states_equiv_class}
        \ket{\psi}=P\ket{\psi_i}=PP_i\ket{\psi_0}=P_i P\ket{\psi_0}
    \end{equation}
    where we again used the commutativity of the product operators. Moreover, from our assumption we have
    $$
    (P\ket{\psi_i})_A=\ket{\psi}_A=\ket{\psi_i}_A.
    $$
    Thus, from Definition~\ref{def:generatorSet}~\ref{item:local_invariance}, we have ${P\in\spos_0}$ and hence, from Eq.~\eqref{eq:rel_states_equiv_class} we deduce ${\ket{\psi}\in\SSP_i}$.
    \item To prove the final statement, it suffices to show that the mapping ${r_i\colon \SSP_0\to\SSP_i\colon P\ket{\psi_0}\mapsto P_iP\ket{\psi_0}}$ for ${1\leq i\leq m}$ is well-defined and bijective. Using that generator set is pointwise-disjoint, we have for ${P,Q\in\spos_0}$ with ${P\ket{\psi_0}=Q\ket{\psi_0}}$ that ${P=Q}$ which shows that $r_i$ is well defined and hence ${\abs{\SSP_0}=\abs{\mathcal{O}_A}}$. The surjectivity follows from definition of the set $\SSP_i$. Since all elements in $\transOs$ are self-inverse and $\transOs$ is pointwise-disjoint, we obtain for ${P,Q\in\spos_0}$ with ${P_iP\ket{\psi_0}=P_iQ\ket{\psi_0}}$ the relation ${P=Q}$, implying the injectivity of $r_i$.\qedhere
    \end{enumerate}
\end{proof}
\begin{proof}[Proof of Theorem~\ref{theorem:equivalenzOperatorbipartitions}]
    The "only if" direction follows immediately from Proposition~\ref{prop:quotient_set_equality} and the fact that the sets $\SSP_i$ depend only on operator sets. For the other direction, it is sufficient to verify that for two distinct nonempty sets ${A,B\subset[n]}$ with ${\mathcal{O}_A\neq\mathcal{O}_B}$ it holds ${\SSP/{\sim_A}\neq \SSP/{\sim_B}}$. Without loss of generality suppose that there exists ${P\in\mathcal{O}_A}$ such that ${P\notin\mathcal{O}_B}$. Then there exists ${i>0}$ such that $P\ket{\psi_0}\in \SSP_i(\mathcal{O}_B,\ket{\psi_0})$. From Proposition~\ref{prop:quotient_set_equality} we deduce  $$\SSP_0(\mathcal{O}_B,\ket{\psi_0})\cap\SSP_i(\mathcal{O}_B,\ket{\psi_0})=\emptyset,$$ which implies ${P\ket{\psi_0}\notin\SSP_0(\mathcal{O}_B,\ket{\psi_0})}$. Hence, ${\SSP_0(\mathcal{O}_A,\ket{\psi_0})\neq \SSP_0(\mathcal{O}_B,\ket{\psi_0})}$, which together with Proposition~\ref{prop:quotient_set_equality} shows the relation.\qedhere
\end{proof}

\subsubsection{Parity embedding: Technical details and results\label{sec:mathParityEmbedding}}
In this section, we give a technical introduction to the parity embedding and recap some of the definitions presented in~\cite{dreier2024uniqueness}. Furthermore, we introduce some new definitions and derive new results required to prove the existence of a generator set for the parity embedding.

A hypergraph $H$ is a pair $(V,E)$, where $V$ is the set of all vertices and ${E\subset\mathcal{P}(V)}$ is the edge set of $H$. Here, ${\mathcal{P}(V)}$ denotes the power set of $V$, and ${\mathcal{P}(E)}$ denotes the edge space of $E$, which is defined as the power set of $E$.
An element of the edge space $C\in \mathcal{P}(E)$ is called a \emph{constraint} of $H$ if it satisfies the following condition:
\begin{equation}
\label{def:parityConstraint}
    \forall v \in V_{C} : \abs{\set{e\in C \mid v \in e}} \text{ is even, }
\end{equation}
where the set $V_C$ of all vertices contained in the constraint $C$ is given by ${V_C \coloneqq \set{v\in V\mid \exists e \in C\mid v\in e}}$. In this article, in addition to the terminology used in~\cite{dreier2024uniqueness}, we define a constraint 
$C$ with $|C|=p$ as a \emph{$p$-body constraint}. The \emph{constraint space} of $H$ is defined by ${\mathcal{C}_H\coloneqq \set{C\in \mathcal{P}(E)\mid C \text{ is a constraint}}}$, which forms a vector space over $\mathbb{Z}_2$.

Now, let us introduce some new definitions. In the following, we denote by $H$ a hypergraph with vertex set $V$ and edge set $E$. Moreover, for ${p\in \mathbb{Z}_2}$, we define
${\non{p} \coloneqq (p+1)\mod 2 = p \oplus 1}$ the negation of ${p}$. For a function $f\colon M\to\mathbb{Z}_2$ on some nonempty set $M$ we denote by ${\non{\vw}\colon M\to\mathbb{Z}_2\colon v\mapsto \overline{f(v)}}$ the \emph{negation of $f$}.
\begin{definition}[Weight of a vertex and edge set]
    We call a function $\vw \colon V\to\mathbb{Z}_2$ a \emph{weight of $V$} and a function ${\wE\colon E\to\mathbb{Z}_2}$ a \emph{weight of $E$}.
\end{definition}
\begin{definition}[Logical configuration state]
For a weight $\vw$ of $V$ we define ${\ket{\phi_{\vw}}\coloneqq\ket{\vw (v_1) \cdots \vw(v_{\abs{V}})}}$ as a \textit{logical (configuration) state} and denote by ${\mathcal{H}(V) = \set{\ket{\phi_{\vw}}\mid \text{$\vw$ is a weight of $V$}}}$ the set of logical configuration states.
\end{definition}
\begin{definition}[Physical configuration state]
\label{def:phys_conf_states}
For a weight $\wE$ of $E$ we define ${\ket{\psi_{\wE}}\coloneqq\ket{w(e_1) \cdots w(e_{\abs{E}})}}$ as a \textit{physical (configuration) state} and define ${\mathcal{H}(E)\coloneqq \set{\ket{\psi_{\wE}} \mid \wE\text{ is a weight of $E$}}}$ as the set of physical configuration states.
\end{definition}
\begin{definition}[Parity states]
\label{def:parityStateSpace}
    We call a weight $w\colon E\to\mathbb{Z}_2$ of $E$ a \emph{parity weight} if and only if there exists a weight $\vw$ of $V$ such that for all $e\in E$
    \begin{equation}
        \label{eq:parityweight}
        \non{w}(e)=\bigoplus_{v\in e} \vw(v).
    \end{equation}
    Moreover, for a parity weight $w$ we call the physical state $\ket{\psi_{w}}$ a \emph{parity state} and the set of all parity states ${\Pi := \set{\ket{\psi_w}\mid w \text{ is a parity weight}}}$ the \emph{parity state space}.
\end{definition}
\begin{lemma}
\label{lemma:parityStateConstr0}
    A weight $w$ of $E$ is a \textit{parity weight} if and only if 
\begin{equation}
\label{eq:parityStateConstraints}
\forall C \in \CS_H\colon |\set{e\in C\mid w_{e}=0}| \text{ is even,}
\end{equation}
or equivalently, for all $C \in \CS_H$ we have $\bigoplus_{e\in C} \overline{w}_e=0$.
\end{lemma}
\begin{proof}
    First, suppose that $w$ is a parity weight and $f$ a weight of $V$ satisfying~\eqref{eq:parityweight}. Then, for $C\in\CS_H$, we have
    \begin{equation*}
        \bigoplus_{e\in C} \non{w}(e)=\bigoplus_{e\in C}\bigoplus_{v\in e} \vw(v)=\bigoplus_{v\in V_C}\bigoplus_{\substack{e\in C\\ v\in e}} \vw(v).
    \end{equation*}
    By definition of a constraint [Eq.~\eqref{def:parityConstraint}], for each $v\in V_C$ the inner sum is an even number of additions modulo two, i.e., we have that
    \begin{equation*}
        \bigoplus_{v\in V_C}\bigoplus_{\substack{e\in C\\ v\in e}} \vw(v)
         = 0.
    \end{equation*}
    Thus, ${\bigoplus_{e\in C} \non{w}(e)}$ is equal to zero for ${C\in \CS_{H}}$.\\
    Now, suppose that $\wE$ is a weight of $E$ such that~\eqref{eq:parityStateConstraints} holds. In order to show that $\wE$ is a parity weight, we see from~\eqref{eq:parityweight} that this is equivalent to show that the linear system with coefficient matrix $A\in\mathbb{Z}_2^{m\times n}$ and right-hand side $b\in\mathbb{Z}_2^{m}$ defined by
    \begin{equation*}
        A_{ij}\coloneqq\begin{cases}
            1,\quad&\text{if $v_j\in e_i$},\\
            0,\quad&\text{else},\\
        \end{cases}\quad \text{and}\quad b_i\coloneqq \non{\wE}_{e_i},
    \end{equation*}
    has a solution. Let ${P\in \mathbb{Z}_2^{m\times m}}$ an invertible matrix such that ${A'\coloneqq PA}$ has reduced row echelon form, and denote ${b'\coloneqq Pb}$. Note that the linear system with coefficient matrix ${A'\in\mathbb{Z}_2^{m\times n}}$ and right-hand side ${b'\in\mathbb{Z}_2^{m}}$ has a solution if and only if for all ${i=1,\ldots,m}$ for which the {$i$-th} row of $A'$ is the zero vector, the {$i$-th} entry of $b'$ is zero. Since the matrix $P$ performs row operations, we can find for such a row with index ${1\leq i\leq m}$ row indices ${1\leq j_1<\ldots<j_l\leq m}$ satisfying ${\bigoplus_{k=1}^lA_{j_k p}
    =0}$ for all ${1\leq p\leq n}$. We are left to prove that ${b_i'=\bigoplus_{k=1}^lb_{j_k}}$ is zero. By definition of the matrix $A$ and the modulo two addition, we see that ${\bigoplus_{k=1}^lA_{j_k p}
    =0}$ for all ${1\leq p\leq n}$ if and only if for all ${v\in V}$ the number ${\abs{\set{k\in\set{1,\ldots,l}\mid v\in e_{j_k}}}}$ is even. This implies that ${C\coloneqq\set{e_{j_1},\ldots,e_{j_l}}}$ is a constraint and hence, by our assumption, ${\abs{\set{e\in C\mid w_{e}=0}}}$ is even. Therefore, ${\abs{\set{e\in C\mid \non{w}_{e}=1}}}$ is even, implying ${\bigoplus_{k=1}^lb_{j_k}=\bigoplus_{e\in C}\non{\wE}_e=0}$.\qedhere
\end{proof}
\begin{definition}[Logical negation]
\label{def:logicalNegation}
For a weight $\wE$ of $E$ we call the mapping
$$
l_w\colon E \times V \rightarrow \mathbb{Z}_2\colon (e,v) \mapsto l_w(e, v)
$$
with
\begin{equation*}
l_w(e, v) \coloneqq \begin{cases}
     \non{w}_e,\quad &\text{if } v\in e,\\
     w_e,\quad & \text{else,}
\end{cases}
\end{equation*}
for ${v\in V}$ and ${e\in E}$ a \textit{logical negation} of the weight $w$.
For brevity, we write ${l_v(w_e)\coloneqq l_w(e, v)}$.
\end{definition}
Next, we define the \textit{logical line} $L_v(H)$ and the corresponding \emph{logical line operator} $\Lambda_v$ of a vertex ${v\in V}$. As we will later see in Lemma~\ref{lemma:generatorSetParityEmbedding}, logical line operators can be used to define a generator set of the parity state space.
\begin{definition}[Logical line]
\label{def:logicalLine}
For $v\in V$ we call
$$L_{v}(H)\coloneqq \set{e\in E\mid v\in e}$$
the \emph{logical line of the vertex $v$ in the hypergraph ${H}$}.
\end{definition}
To define logical line operators, we need the following lemma.
\begin{lemma}
    \label{lemma:lgocialLinePitoPi}
    Let ${w}$ be a weight of $E$, ${v\in V}$ and define
    \begin{equation}
        \label{eq:def_logical_line_op}
        \Lambda_v(\ket{\psi_w}) \coloneqq \ket{l_v(w_{e_1}) \cdots l_v(w_{e_{\abs{E}}})}.
    \end{equation}
    \begin{enumerate}[(i)]
    \item If ${\ket{\psi_w}\in  \Pi}$, then ${\Lambda_v(\ket{\psi_w})\in \Pi}$.
        \item If ${\ket{\psi_w}\in \mathcal{H}(E)\setminus \Pi}$, then ${\Lambda_v(\ket{\psi_w})\in \mathcal{H}(E)\setminus \Pi}$.
    \end{enumerate}
\end{lemma}
\begin{proof}
We only prove the first statement. The second statement can be proven analogously. Suppose that ${v\in V}$ and let $w$ be a parity weight. Then we observe for a constraint $C\in\CS_H$ that
\begin{align*}
        \bigoplus_{e \in C} \overline{l_v(w_e)}&=\bigoplus_{\substack{e \in C\\v\in e}} \overline{l_v(w_e)}\oplus\bigoplus_{\substack{e \in C\\v\notin e}} \overline{l_v(w_e)}\\
        &=\bigoplus_{\substack{e \in C\\v\in e}} (\overline{w}_e\oplus 1)\oplus\bigoplus_{\substack{e \in C\\v\notin e}} \overline{w}_e        =\bigoplus_{\substack{e \in C\\v\in e}} 1\oplus\bigoplus_{e \in C} \overline{w}_e.
\end{align*}
Hence, since $C$ is constraint we deduce from Lemma~\ref{lemma:parityStateConstr0} that the above sum equals zero. Applying Lemma~\ref{lemma:parityStateConstr0} again shows that ${\Lambda_v(\ket{\psi_w})}$ is a parity state.
\end{proof}
\begin{definition}[Logical line operator]
\label{def:logicalLineOperatorProducts}
For ${v\in V}$ we call the mapping
${\Lambda_v \colon \Pi \rightarrow \Pi}$ defined by Eq.~\eqref{eq:def_logical_line_op} a \emph{logical line operator of $v$}  
and  ${\mathbb{\Lambda} \coloneqq \set{\Lambda_v \mid v \in V} \cup \set{\mathbb{1}}}$ the \emph{set of all logical operators}, where  ${\mathbb{1}}$ denotes the identity operator on $\Pi$. Note that $\Lambda_v$ is a well-defined operator as shown in Lemma~\ref{lemma:lgocialLinePitoPi}.
\end{definition}
We conclude this section by showing that $\mathbb{\Lambda}$ forms a generator set of the parity state space. In the following, we write ${\Lambda_v \Lambda_p}$ as the composition of two logical line operators for $p,v\in V$.
\begin{lemma}
    \label{lemma:generatorSetParityEmbedding}
    The set of all logical line operators $\mathbb{\Lambda}$ is a pointwise-disjoint generator set of $\Pi$ whose elements are self-inverse.
\end{lemma}
\begin{proof}
    The first property in Definition~\ref{def:generatorSet} and that each logical line operator is self-inverse can be easily verified.
    Next, we show that for all ${\ket{\psi}, \ket{\phi} \in \Pi}$ there exists an operator $P\in \posL$, such that
    $$
    P\ket{\psi} = \ket{\phi}.
    $$
    By definition of the parity states, we can find two weights $f_1,f_2$ of $V$ such that Eq.~\eqref{eq:parityweight} holds, replacing $w$ and $f$ with $w_1$ and $f_1$, and $w_2$ and $f_2$ respectively. Define ${\hat{V}\coloneqq\set{v\in V\mid f_1(v)\neq f_2(v)}}$ and set
\begin{equation*}
 \ket{\hat{w}_{e_1}\ldots \hat{w}_{e_{\abs{E}}}}\coloneqq P\ket{\psi_{w_1}},
\end{equation*}
where ${P\coloneqq\prod_{v\in \hat{V}}\Lambda_v\in\posL}$. Using the the definition of logical line operators, we deduce for ${1\leq i\leq \abs{E}}$
\begin{equation*}
    \overline{\hat{w}}_{e_i}=\overline{w_1(e_i)\oplus\bigoplus_{\substack{v\in\hat{V}\cap e_i}}1}=\overline{w_1}(e_i)\oplus \bigoplus_{\substack{v\in\hat{V}\cap e_i}}1.
\end{equation*}
Hence, applying Eq.~\eqref{eq:parityweight} on $w_1$ shows that the right-hand side of the above equation equals
\begin{align*}
    \bigoplus_{v\in e_i}f_1(v)\oplus \bigoplus_{v\in\hat{V}\cap e_i}1&=\bigoplus_{v\in e_i\setminus\hat{V}}f_1(v)\oplus \bigoplus_{\substack{v\in\hat{V}\cap e_i}}\overline{f_1(v)}\\
    &=\bigoplus_{v\in e_i}f_2(v),
\end{align*}
where we used in the last equality the relation ${\hat{V}=\{v\in V\mid f_1(v)=\overline{f_2(v)}\}}$. Thus, using Eq.~\eqref{eq:parityweight} with weight $w_2$ proves the desired statement.\\
Finally, we prove that for ${P,T\in\posL}$ and ${\ket{\psi}\in\Pi}$, if ${P\ket{\psi}=T\ket{\psi}}$, then ${P=T}$. The local invariance follows from the definition of the logical line operators. Since all elements in $\posL$ are self-inverse, it is sufficient to prove that ${Q\coloneqq PT=\mathbb{1}}$. Let ${\tilde{V}\subset V}$ such that ${Q=\prod_{v\in \tilde{V}}\Lambda_v}$. Since ${Q \ket{\psi}=\ket{\psi}}$
we deduce from the definition of the logical line operators that ${\abs{\tilde{V}\cap e}}$ is even for all ${e\in E}$, which immediately shows that ${Q \ket{\phi}=\ket{\phi}}$ for all ${\ket{\phi}\in \Pi}$.\qedhere
\end{proof}
\subsubsection{Proof of Lemma~\ref{lemma:constructBipOperatorSetsForAllBips}\label{proof:constructBipOperatorSetsForAllBips}}
Here, we give a proof of Lemma~\ref{lemma:constructBipOperatorSetsForAllBips}.
\begin{proof}
    \begin{enumerate}[label=(\roman*),wide=\parindent,leftmargin=0pt,align=left]
    \item Let ${P\in\mathcal{O}_A}$ and ${\tilde{V}\subset V}$ such that ${P=\prod_{v\in\tilde{V}}\Lambda_v}$.
    First, we show that there exists ${K\subset \set{1, \dots, N}}$ such that 
    \begin{equation}
    \label{eq:vertixsetOA}
        \tilde{V}=\bigcup_{l\in K}\tilde{V}_{l}\cup\left(\tilde{V}\setminus\bigcup_{k=1}^{N}V_k\right),
    \end{equation}
where ${\tilde{V}_l\in\mathcal{V}_l}$ for all ${l\in K}$, which implies ${P\in \spos\left(\mathcal{U} \cup \mathcal{Q}\right)}$. Define ${\tilde{V}_l\coloneqq V_l\cap \tilde{V}}$ for ${l=1,\ldots,N}$ and set ${K\coloneqq\set{ l\in[N]\mid \tilde{V}_l\neq\emptyset}}$.
Since ${P\in\mathcal{O}_A}$ we deduce that ${\abs{e\cap\tilde{V}}}$ is even for all ${e\in A}$. Hence, we have for ${l\in K}$ that ${\abs{e\cap\tilde{V}}=\abs{e\cap\tilde{V}_l}}$ is even for all ${e\in E_l}$, showing that ${\tilde{V}_l\in\mathcal{V}_l}$. This shows that ${P\in \spos\left(\mathcal{U} \cup \mathcal{Q}\right)}$. Note that the definition of $\mathcal{U}$ and $\mathcal{Q}$ implies that each ${P\in \posG{\mathcal{U}\cup \mathcal{Q}}}$ is also contained in $\mathcal{O}_A$.
    \item Second, we prove ${\mathcal{V}_k=\set{V_k}}$ when $H$ is a graph. Suppose by contradiction that there exists $\tilde{V}\subset V_k$ such that ${\tilde{V}\neq\emptyset}$ and ${\tilde{V}\neq V_k}$. Thus, we can choose ${v_0\in \tilde{V}}$ and ${w\in V_k}$ with ${w\notin\tilde{V}}$. Since ${H\vert_{E_k}}$ is connected, we can find a path ${(v_0,e_1,v_1,\ldots,e_m,v_m)}$ in $H\vert_{E_k}$ such that ${v_l\in V_k}$, ${e_l\in E_k}$ for ${l=1,\ldots,m}$ and ${v_m=w}$. Define ${p\coloneqq\min\set{l\in[m]\mid e_l\not\subset\tilde{V}}}$. If ${p=1}$, then ${\abs{e_1\cap\tilde{V}}=\abs{\set{v_0}}=1}$. For $p>1$ we deduce ${e_{p-1}\subset\tilde{V}}$ and hence ${\abs{e_p\cap\tilde{V}}=\abs{\set{v_{p-1}}}=1}$, showing that both cases lead to a contradiction.\qedhere
\end{enumerate}
\end{proof}

\subsubsection{Characterization of further bundles under the parity embedding\label{section:applicationExample}}
In Sec.~\ref{sec:embeddingExamples}, we investigated bundles of logical lines and derived a lower bound on the total number of subsystems they contain.

In the following, we analyze the remaining types of bundles which appear in Fig.~\ref{fig:GroupEntanglementDynamics} and determine for some cases the total number of subsystems inside these bundles.

For the subsequent analysis, we make use of the following lemma.
\begin{lemma}
    \label{lem:prop_product_operators}
    Let $H=(V,E)$ be a hypergraph.
    \begin{enumerate}[label=(\roman*)]
        \item For ${V_1,V_2\subset V}$ we have ${\Lambda_{V_1}\Lambda_{V_2}=\Lambda_{V_1\triangle V_2}}$.
        \item Let $W\subset V$. If $H$ is a connected graph and ${\Lambda_W=\mathbb{1}}$, then ${W=V}$ or ${W=\emptyset}$. Here, we set ${\Lambda_\emptyset\coloneqq\mathbb{1}}$. 
    \end{enumerate}
\end{lemma}
\begin{proof}
    The first property can be easily verified. The second property can be proven analogously to the second statement in Lemma~\ref{lemma:constructBipOperatorSetsForAllBips}.\qedhere
\end{proof}
The first type of bundle we investigate contains all subsystems which are equivalent to a three-cycle. Two of these bundles are depicted as the lowest (violet) and the next higher (green-blue) line in Fig.~\ref{fig:GroupEntanglementDynamics}~(b).

\begin{example}[Three-cycles]
\label{example:threeCycles}
  For a complete graph ${H=(V,E)}$ with ${\abs{V}>3}$ each of the closed loops ${A\coloneqq \set{\set{v,u},\set{u,w},\set{w,v}}}$ with three different vertices ${v,u,w\in V}$, the equivalence class ${[A]_{\Pi}}$ consists of eight elements.
\end{example}
\begin{proof}
    \begin{enumerate}[label=(\roman*),wide=\parindent,leftmargin=0pt,align=left]
        \item \label{label:first_step_3-cycles}First, since $H\vert_A$ is connected and $H$ is a complete graph, implying that $H\vert_{A^c}$ is connected and its vertex set equals $V$, we deduce from Lemma~\ref{lemma:constructBipOperatorSetsForAllBips} that
        \begin{equation*}
            \mathcal{O}_A=\posG{\set{\Lambda_{A}}\cup \mathcal{Q}}\quad\text{and}\quad \mathcal{O}_{A^c}=\set{\mathbb{1}},
        \end{equation*}
        where ${\mathcal{Q} \coloneqq \set{\Lambda_v\in \mathbb{\Lambda} \left\vert\, v\in V\setminus\set{v,u,w}\right.}}$. Furthermore, for $B\subset A$ with $\abs{B}=2$, we similarly obtain ${\mathcal{O}_B=\mathcal{O}_A}$ and ${\mathcal{O}_{B^c}=\mathcal{O}_{A^c}}$, showing ${\abs{[A]_{\Pi}}\geq 8}$.
        \item Next, we have to prove that for all subsystems other than in~\ref{label:first_step_3-cycles} (and their complements) do not belong to ${[A]_{\Pi}}$.\\
        First, suppose ${B\subset A}$ with ${\abs{B}=1}$, and without loss of generality, let ${B=\set{\set{v,u}}}$. Define ${W\coloneqq\set{v,u}}$ and assume by contradiction that ${\Lambda_W\in\mathcal{O}_B}$ is an element of $\mathcal{O}_A$. Hence, using Lemma~\ref{lemma:constructBipOperatorSetsForAllBips} we can find ${c_A,c_v\in\set{0,1}}$ for ${v\in V\setminus{\set{v,u,w}}}$ such that
        \begin{equation}
            \label{eq:proof_second_step_three-cycles}
            \Lambda_W=\Lambda_{\set{v,u,w}}^{c_A}\prod_{v\in V\setminus{\set{v,u,w}}}\Lambda_{v}^{c_v},
        \end{equation}
        which, by Lemma~\ref{lemma:generatorSetParityEmbedding} and~\ref{lem:prop_product_operators}, implies that
        \begin{equation*}
            W\triangle c_A\set{v,u,w}\triangle\bigcup_{v\in V\setminus{\set{v,u,w}}}c_v\set{v}\in \set{V,\emptyset}.
        \end{equation*}
        However, for ${c_A=0}$ as well as ${c_A=1}$, it can be easily seen that the above set does not correspond to $V$ or $\emptyset$, which is a contradiction. Since ${\mathcal{O}_{B^c}=\set{\mathbb{1}}}$, we conclude by Theorem~\ref{theorem:equivalenzOperatorbipartitions} that ${B\notin [A]_{\Pi}}$.\\
        Finally, we assume ${B\subset E}$ with ${B\setminus A\neq\emptyset}$. If ${H\vert_B}$ is disconnected, we can find ${v'\in V\setminus{\set{v,u,w}}}$ which is contained in the vertex set of $H\vert_B$. Then, using similar arguments as above (see Eq.~\eqref{eq:proof_second_step_three-cycles}), we can deduce that ${\Lambda_{v'}\in\mathcal{O}_A}$ and ${\Lambda_{v'}\notin\mathcal{O}_B}$, which shows that ${B\notin [A]_{\Pi}}$. If ${H\vert_B}$ is disconnected, it is sufficient to consider the case where the vertex set of  ${H\vert_{B}}$ equals $V$ and the vertex set of ${H\vert_{B^c}}$ is not equal to $V$. The case where the vertex set of  ${H\vert_{B}}$ is not equal to $V$ and the vertex set of ${H\vert_{B^c}}$ equals $V$ can be proven analogously. Excluding the previous choices of $B$ (and their complements), we find ${v\in V}$ in vertex set of ${H\vert_{B^c}}$ such that $v$ is not in the vertex set of ${H\vert_{A}}$. This shows that ${\Lambda_v\in\mathcal{O}_A}$ and ${\Lambda_v\notin\mathcal{O}_{B^c}}$, implying the desired statement.\qedhere
    \end{enumerate}
\end{proof}

We note that the lowest (violet) bundle in Fig.~\ref{fig:GroupEntanglementDynamics}~(b) comprises the subsystems ${B_1= \set{\set{0,1}, \set{1, 4}, \set{0, 4}}}$ and ${B\subset A}$ with ${\abs{B}=2}$. The next higher bundle in green-blue consists of the subsystems ${D=\set{\set{0,2},\set{0,4},\set{2,4}}}$ and all subsets of $D$ with size two. For a complete graph of five vertices the parity embedding produces 10 three-cycles, which are represented by the green-blue bars in Fig.~\ref{fig:groupsOfFinalEntanglementValues}.\\

From Lemma~\ref{lemma:constructBipOperatorSetsForAllBips} it follows that for a graph ${H=(V, E)}$ and a subsystem ${A\subset E}$ which
we can be decomposed as ${A=\cup_{i=1}^{n} B_i}$, where  ${B_i\cap B_j=\emptyset}$ for all ${1\leq i<j\leq n}$, and where the corresponding graphs ${H\vert_{B_i}}$ are connected for all ${i\in[n]}$, that ${\mathcal{O}_A = \posG{\set{\Lambda_{V_{B_i}}\mid i\in[n]}\cup \mathcal{Q}}}$ with ${\mathcal{Q}=\set{\Lambda_v\in \mathbb{\Lambda} \mid v\in V\setminus \cup_{i=1}^n V_{B_i}}}$.
Here we give an example for such a subset $A$ with ${n=2}$ in a complete graph with ${\abs{V}=5}$.
\begin{example}[Disconnected subgraphs]
    \label{lemma:SameDisconnectedSubgraphs}
    Let ${H=(V,E)}$ be a complete graph with ${\abs{V}=5}$ and ${C\coloneqq \set{\set{v,u},\set{u,w},\set{w,v}}}$ be a three-cycle with three different vertices ${v,u,w\in V}$. Furthermore, define ${A\coloneqq C\cup \set{V_C^c}}$, where ${V_C\coloneqq\set{u,v,w}}$. Then the equivalence class ${[A]_{\Pi}}$ consists of eight elements.
\end{example}
\begin{proof}
    \begin{enumerate}[label=(\roman*),wide=\parindent,leftmargin=0pt,align=left]
        \item Using Lemma~\ref{lemma:constructBipOperatorSetsForAllBips} and that $H$ is complete, we again see that ${\mathcal{O}_A = \posG{\set{\Lambda_{V_C}, \Lambda_{V_C^c}}}}$ and ${\mathcal{O}_{A^c} = \set{\mathbb{1}}}$. Moreover, for each
        \begin{equation}
            \label{eq:proof_disconnected_subgraphs}
            B\in \set{D\cup \set{V_C^c} \mid \abs{D} = 2 \text{ and } D\subset C},
        \end{equation}
        we similarly deduce that ${\mathcal{O}_B = \posG{\set{\Lambda_{V_C}, \Lambda_{V_C^c}}}}$ and ${\mathcal{O}_{B^c} = \set{\mathbb{1}}}$ (see also proof of Example~\ref{example:threeCycles}), showing that ${B\in [A]_{\Pi}}$ and thus ${\abs{[A]_{\Pi}}\geq 8}$.
        \item 
        Suppose by contradiction that there exists ${B\in [A]_{\Pi}}$ such that ${B\cap (E\setminus{A})\neq\emptyset}$. Thus, we can find ${e\in B\cap (E\setminus{A})}$ and a vertex set $V'$ which is a connected component of ${H\vert_{B}}$ such that ${e\subset V'}$. Using similar arguments as in the proof of Example~\ref{example:threeCycles} shows that ${\Lambda_{V'}\notin\mathcal{O}_A}$ but ${\Lambda_{V'}\in\mathcal{O}_B}$, which is a contradiction. Therefore, any subsystem $B$ has to be contained in $A$. For the case ${\abs{B}=3}$ we have in total four subsets ${B\subset A}$, where three of them are already an element of the set in Eq.~\eqref{eq:proof_disconnected_subgraphs}. For ${B=C}$ and ${q\in V^c}$ we have ${\Lambda_{q}\in \mathcal{O}_B}$ but ${\Lambda_{q}\notin \mathcal{O}_A}$ and hence, ${B\notin[A]_{\Pi}}$.
        If ${\abs{B}=2}$, we see that both $\mathcal{O}_B$ and $\mathcal{O}_{B^c}$ are not equal to ${\set{\mathbb{1}}}$, implying that ${\abs{[A]_{\Pi}}=8}$.\qedhere
    \end{enumerate}
\end{proof}
The (dark blue) bundle with the second highest entanglement value in Fig.~\ref{fig:GroupEntanglementDynamics}~(b) consists of the four unordered bipartitions contained in the equivalence class $[A]_{\Pi}$ where ${A=C\cup\set{\set{0,2}}}$ and ${C = \set{\set{1,3},\set{3,4},\set{1,4}}}$. We note that there exist exactly 10 bundles of this type, which are illustrated in Fig.~\ref{fig:groupsOfFinalEntanglementValues} as the blue-violet bars.\\

%
%
Finally, we provide an example of a bundle that includes the bundle in Fig.~\ref{fig:GroupEntanglementDynamics} with the highest final entanglement value.
\begin{example}[Double spanning tree bipartition]
    \label{lemma:highestEntanglementGroup}
    Let ${H=(V, E)}$ be a graph and ${A\subset E}$.
    \begin{enumerate}[(i)]
        \item\label{item:bipart_op_set_db_spanning_tree} If ${H\vert_{A}}$ and ${H\vert_{A^c}}$ are connected and their vertex sets correspond to $V$, then ${\set{\mathcal{O}_{A}, \mathcal{O}_{A^c}} = \set{\set{\mathbb{1}},\set{\mathbb{1}}}}$.
        \item Under the same assumptions as in \ref{item:bipart_op_set_db_spanning_tree}, for all ${c_{\ket{\psi}}\in\mathbb{C}}$, ${\ket{\psi}\in\Pi}$ with ${\sum_{\ket{\psi}\in\Pi}\abs{c_{\ket{\psi}}}^2=1}$ and ${\ket{\Psi}\coloneqq\sum_{\ket{\psi}\in\Pi}c_{\ket{\psi}}}$ we have 
        \begin{equation}
        \label{eq:specSingleHighestEnt}
        \spec{\rho_A} =  \set{\left.\abs{c_{\ket{\psi}}}^2 \,\right\vert \ket{\psi}\in\Pi},
        \end{equation}
        where $\rho_A$ denotes the reduced density matrix of $\ket{\Psi}$ of the subsystem $A$.
    \end{enumerate}
\end{example}
\begin{proof}
\begin{enumerate}[(i)]
    \item The first statement follows immediately from Lemma~\ref{lemma:constructBipOperatorSetsForAllBips}.
    \item To prove the second statement, we deduce from Proposition~\ref{prop:quotient_set_equality} that each equivalence in ${\SSP/{\sim_A}}$ has size one. Using Eq.~\eqref{eq:reduced_density_matrix_formula} and~\eqref{eq:coefficientsRedDens} in the proof of Theorem~\ref{thm:equivalence_spectrum_sim_subspace}, we derive
    $$
    \rho_A = \sum_{\ket{\psi},\ket{\phi}\in \Pi} d_{\ket{\psi},\ket{\phi}} \ket{\psi_{A}}\bra{\phi_{A}}
    $$
    with
    $$
    d_{\ket{\psi},\ket{\phi}} = \delta_{\ket{\psi},\ket{\phi}} c_{\ket{\psi}}c_{\ket{\phi}}^*,
    $$
    which shows Eq.~\eqref{eq:specSingleHighestEnt}.\qedhere
\end{enumerate}
\end{proof}
We note that determining the number of ways to partition the edge set $E$ of a graph ${H=(V,E)}$ into two connected spanning subgraphs ${H\vert_A = (V,A)}$ and ${H\vert_{A^c} = (V, A^c)}$ is known to be a hard combinatorial enumeration problem. There is no simple closed formula but it is related to the \emph{spanning tree packing problem} and \emph{Tutte's theorem}. 
In our experiments, we used for the example in Fig.~\ref{fig:GroupEntanglementDynamics} a brute force counting algorithm to count all these bipartitions where $120$ unordered bipartition with ${\abs{A}=4}$ and $96$ unordered bipartition with ${\abs{A}=5}$ has been detected. Thus, in total we found $216$ unordered bipartitions and hence, $432$ subsystems inside this bundle. Moreover, we observe that for all subsystems ${B\not\sim_{\SSP}A}$ that the von Neumann entropy calculated with respect to $B$ is smaller than the von Neumann entropy calculated with respect to the subsystem $A$. We also emphasize that for all subsystems in ${[A]_{\Pi}}$, the entanglement spectrum provides information about the composition of the entangled system state. Since the entanglement spectrum is experimentally measurable~(see~\cite{ES_PichlerZoller}), our theory could be used to identify a subset $A$ that lies in the equivalence class with ${\set{\mathcal{O}_A, \mathcal{O}_{A^c}}=\set{\set{\mathbb{1}}, \set{\mathbb{1}}}}$ which, by Eq.~\eqref{eq:specSingleHighestEnt}, then allows to experimentally determine how many basis states are present in the superposition of the system state as well as how their weights are distributed. 

\subsubsection{Minor embedding: Technical details and results}
In the final section we introduce some notations to describe the minor embedding and define the generator set of this embedding similar to the parity embedding.
Let ${M=(V_M,E_M)}$ be a minor graph of the logical graph $H=(V, E)$ and
for each vertex ${v\in V}$
$$
    C_v = \set{({v, i)}\mid i = 1, \dots, N_v} \subset V_M,
$$
the corresponding chain of $v$, where $N_v$ is the number of physical vertices in the chain $C_v$. Moreover, we define for a weight $f$ of $V_M$ the \emph{physical (configuration) state} as ${\ket{\psi_f}\coloneqq \ket{f(n_1)\cdots f(n_{\abs{V_M}})}}$, where $$V_M=\set{(v,i)\mid v\in V,\,1\leq i\leq N_v}=\set{n_1,\ldots,n_{\abs{V_M}}}$$ is an enumeration of $V_M$, and the set of all physical configuration states ${\mathcal{H}(V_M)\coloneqq \set{\ket{\psi_f}\mid f \text{ is a weight of } V_M}}$.
\begin{definition}[Minor states]
    \label{def:minorState}
    We call a weight ${f\colon V_M \to \mathbb{Z}_2}$ of $V_M$ a \emph{minor weight} if and only if for all ${v\in V}$ and for all ${i,j\in \set{1, \dots, N_v}}$ ${f((v,i)) = f((v,j))}$, i.e., $f$ is constant on each chain $C_v$ for ${v\in V}$.
    Moreover, for a minor weight $f$, we call the physical state $\ket{\psi_f}$ a \emph{minor state} and the set of all minor states ${\MES\coloneqq\set{\ket{\psi_f} \mid f \text{ is a minor weight}}}$ the \emph{minor state space}.
\end{definition}
Analogous to the parity embedding, we can define the logical negation in the minor embedding along a chain $C_v$.
\begin{definition}[Logical negation in the minor graph]
    \label{def:minorLogicalNegation}
    For a weight $f$ of $V_M$ we call the mapping
    $$
    l_f\colon V \times V_M \to \mathbb{Z}_2 \colon (v, n) \mapsto l_f(v,n)
    $$
    with
    \begin{equation*}
l_f(v, n) \coloneqq \begin{cases}
     \non{f}(n),\quad &\text{if } n\in C_v, \\
     f(n),\quad & \text{else,} 
\end{cases}
\end{equation*}
for ${v\in V}$ and ${n \in V_M}$ a \emph{logical negation} of the weight $f$. Furthermore, we write ${l_v(f_n) \coloneqq l_f(v, n)}$.
\end{definition}
Similar to Definition~\ref{def:logicalLineOperatorProducts}, we make use of the following lemma to define the \emph{chain operator} $\Gamma_v$ of a vertex ${v\in V}$ as follows.
\begin{lemma}
    \label{lemma:minorStateMapping}
    Let ${f\colon V_M \to \mathbb{Z}_2}$ be a weight of $V_M$ and ${v\in V}$. Moreover, define 
    \begin{equation}
    \label{eq:minorMapOperator}
 \Gamma_v(\ket{\psi_f})\coloneqq \ket{l_v(f_{n_1}) \cdots l_{v}(f_{n_{\abs{V_M}}})}.
    \end{equation}
    \begin{enumerate}[(i)]
        \item If ${\ket{\psi_w}\in \MES}$, then ${\Gamma_v(\ket{\psi_w})\in \MES}$.
        \item If ${\ket{\psi_w}\in \mathcal{H}(V_M)\setminus \MES}$, then ${\Gamma_v(\ket{\psi_w})\in \mathcal{H}(V_M)\setminus \MES}$.
    \end{enumerate}
\end{lemma}
\begin{proof}
The proof proceeds analogously to that of Lemma~\ref{lemma:lgocialLinePitoPi} and is therefore omitted.\qedhere    
\end{proof}
\begin{definition}[Chain operator]
    \label{def:chainOperator}
    For ${v\in V}$ we call the mapping ${\Gamma_v\colon \MES \to \MES}$ defined by Eq.~\eqref{eq:minorMapOperator} the \emph{chain operator of $v$} and ${\mathbb{\Gamma}\coloneqq \set{\Gamma_v \mid v\in V}\cup \set{\mathbb{1}}}$ the \emph{set of all chain operators}, where $\mathbb{1}$ denotes the identity operator on $\MES$. As for product operators, $\Gamma_v$ is a well-defined operator, as shown in Lemma~\ref{lemma:minorStateMapping}.
\end{definition}
We note that, analogous to the parity embedding (see Lemma~\ref{lemma:generatorSetParityEmbedding}), it can be shown that the set of all chain operators $\mathbb{\Gamma}$ is a pointwise-disjoint generator set of $\MES$ whose elements are all self-inverse.

\section{Conclusion}

This work reveals a previously unobserved equivalence in entanglement behavior, demonstrating through numerical simulations that different bipartitions can exhibit equal entanglement entropy within a constrained subspace. We develop a theoretical framework to describe this bundling effect in such subspaces and derive a sufficient condition implying equal spectra for all mixed states within that subspace. Beyond these theoretical insights, our results also provide crucial practical applications. Utilizing our theoretical framework, an a prior classification of the bipartition classes enables to selectively measure the entanglement of bipartitions which belong to different bipartitions classes. Depending on the quantum states that generate the constrained subspace, this can reduce the experimental time required to determine the total entanglement significantly as measuring the entanglement of just a few bipartitions already suffices to obtain the total entanglement. Moreover, based on our theoretical results, we derive efficient algorithms that are applicable for the a prior determination of the bipartitions equivalences classes for constrained subspaces that are generated through embeddings used in quantum optimization, verifying the equivalence of two subsystems and can even run in polynomial time. Finally, our analysis of the full entanglement spectrum highlights the broad applicability of our derivations, extending beyond entanglement entropy analysis to quantum many-body physics in general.

\section*{Acknowledgments}

The numerical simulations throughout this manuscript were performed using the \texttt{qutip} package~\cite{Qutip2}. This research was funded in part by the European Union program Horizon 2020 under Grant No.~817482 (PASQuanS), the Austrian Science Fund (FWF)~10.55776/I6011 through a START grant under Project No. Y1067-N27, the SFB BeyondC Project No. F7108-N38 and the Austrian Research Promotion Agency (FFG Project No.~FO999924030 (FFG Basisprogramm), FFG Project No.~FO999925691).
\bibliographystyle{apsrev4-1}

\begin{thebibliography}{38}%
\makeatletter
\providecommand \@ifxundefined [1]{%
 \@ifx{#1\undefined}
}%
\providecommand \@ifnum [1]{%
 \ifnum #1\expandafter \@firstoftwo
 \else \expandafter \@secondoftwo
 \fi
}%
\providecommand \@ifx [1]{%
 \ifx #1\expandafter \@firstoftwo
 \else \expandafter \@secondoftwo
 \fi
}%
\providecommand \natexlab [1]{#1}%
\providecommand \enquote  [1]{``#1''}%
\providecommand \bibnamefont  [1]{#1}%
\providecommand \bibfnamefont [1]{#1}%
\providecommand \citenamefont [1]{#1}%
\providecommand \href@noop [0]{\@secondoftwo}%
\providecommand \href [0]{\begingroup \@sanitize@url \@href}%
\providecommand \@href[1]{\@@startlink{#1}\@@href}%
\providecommand \@@href[1]{\endgroup#1\@@endlink}%
\providecommand \@sanitize@url [0]{\catcode `\\12\catcode `\$12\catcode
  `\&12\catcode `\#12\catcode `\^12\catcode `\_12\catcode `\%12\relax}%
\providecommand \@@startlink[1]{}%
\providecommand \@@endlink[0]{}%
\providecommand \url  [0]{\begingroup\@sanitize@url \@url }%
\providecommand \@url [1]{\endgroup\@href {#1}{\urlprefix }}%
\providecommand \urlprefix  [0]{URL }%
\providecommand \Eprint [0]{\href }%
\providecommand \doibase [0]{https://doi.org/}%
\providecommand \selectlanguage [0]{\@gobble}%
\providecommand \bibinfo  [0]{\@secondoftwo}%
\providecommand \bibfield  [0]{\@secondoftwo}%
\providecommand \translation [1]{[#1]}%
\providecommand \BibitemOpen [0]{}%
\providecommand \bibitemStop [0]{}%
\providecommand \bibitemNoStop [0]{.\EOS\space}%
\providecommand \EOS [0]{\spacefactor3000\relax}%
\providecommand \BibitemShut  [1]{\csname bibitem#1\endcsname}%
\let\auto@bib@innerbib\@empty
\bibitem [{\citenamefont {De~Chiara}\ and\ \citenamefont
  {Sanpera}(2018)}]{AnnaSanperaDeChiara2018}%
  \BibitemOpen
  \bibfield  {author} {\bibinfo {author} {\bibfnamefont {G.}~\bibnamefont
  {De~Chiara}}\ and\ \bibinfo {author} {\bibfnamefont {A.}~\bibnamefont
  {Sanpera}},\ }\bibfield  {title} {\bibinfo {title} {Genuine quantum
  correlations in quantum many-body systems: a review of recent progress},\
  }\href {https://doi.org/10.1088/1361-6633/aabf61} {\bibfield  {journal}
  {\bibinfo  {journal} {Reports on Progress in Physics}\ }\textbf {\bibinfo
  {volume} {81}},\ \bibinfo {pages} {074002} (\bibinfo {year}
  {2018})}\BibitemShut {NoStop}%
\bibitem [{\citenamefont {Li}\ and\ \citenamefont
  {Haldane}(2008)}]{ES_indtroLiHaldane}%
  \BibitemOpen
  \bibfield  {author} {\bibinfo {author} {\bibfnamefont {H.}~\bibnamefont
  {Li}}\ and\ \bibinfo {author} {\bibfnamefont {F.~D.~M.}\ \bibnamefont
  {Haldane}},\ }\bibfield  {title} {\bibinfo {title} {Entanglement spectrum as
  a generalization of entanglement entropy: Identification of topological order
  in non-abelian fractional quantum hall effect states},\ }\href
  {https://doi.org/10.1103/PhysRevLett.101.010504} {\bibfield  {journal}
  {\bibinfo  {journal} {Phys. Rev. Lett.}\ }\textbf {\bibinfo {volume} {101}},\
  \bibinfo {pages} {010504} (\bibinfo {year} {2008})}\BibitemShut {NoStop}%
\bibitem [{\citenamefont {Garrison}\ and\ \citenamefont
  {Grover}(2018)}]{PhysRevX.8.021026}%
  \BibitemOpen
  \bibfield  {author} {\bibinfo {author} {\bibfnamefont {J.~R.}\ \bibnamefont
  {Garrison}}\ and\ \bibinfo {author} {\bibfnamefont {T.}~\bibnamefont
  {Grover}},\ }\bibfield  {title} {\bibinfo {title} {Does a single eigenstate
  encode the full hamiltonian?},\ }\href
  {https://doi.org/10.1103/PhysRevX.8.021026} {\bibfield  {journal} {\bibinfo
  {journal} {Phys. Rev. X}\ }\textbf {\bibinfo {volume} {8}},\ \bibinfo {pages}
  {021026} (\bibinfo {year} {2018})}\BibitemShut {NoStop}%
\bibitem [{\citenamefont {Geraedts}\ \emph {et~al.}(2016)\citenamefont
  {Geraedts}, \citenamefont {Nandkishore},\ and\ \citenamefont
  {Regnault}}]{PhysRevB.93.174202}%
  \BibitemOpen
  \bibfield  {author} {\bibinfo {author} {\bibfnamefont {S.~D.}\ \bibnamefont
  {Geraedts}}, \bibinfo {author} {\bibfnamefont {R.}~\bibnamefont
  {Nandkishore}},\ and\ \bibinfo {author} {\bibfnamefont {N.}~\bibnamefont
  {Regnault}},\ }\bibfield  {title} {\bibinfo {title} {Many-body localization
  and thermalization: Insights from the entanglement spectrum},\ }\href
  {https://doi.org/10.1103/PhysRevB.93.174202} {\bibfield  {journal} {\bibinfo
  {journal} {Phys. Rev. B}\ }\textbf {\bibinfo {volume} {93}},\ \bibinfo
  {pages} {174202} (\bibinfo {year} {2016})}\BibitemShut {NoStop}%
\bibitem [{\citenamefont {Geraedts}\ \emph {et~al.}(2017)\citenamefont
  {Geraedts}, \citenamefont {Regnault},\ and\ \citenamefont
  {Nandkishore}}]{Geraedts_2017}%
  \BibitemOpen
  \bibfield  {author} {\bibinfo {author} {\bibfnamefont {S.~D.}\ \bibnamefont
  {Geraedts}}, \bibinfo {author} {\bibfnamefont {N.}~\bibnamefont {Regnault}},\
  and\ \bibinfo {author} {\bibfnamefont {R.~M.}\ \bibnamefont {Nandkishore}},\
  }\bibfield  {title} {\bibinfo {title} {Characterizing the many-body
  localization transition using the entanglement spectrum},\ }\href
  {https://doi.org/10.1088/1367-2630/aa93a5} {\bibfield  {journal} {\bibinfo
  {journal} {New Journal of Physics}\ }\textbf {\bibinfo {volume} {19}},\
  \bibinfo {pages} {113021} (\bibinfo {year} {2017})}\BibitemShut {NoStop}%
\bibitem [{\citenamefont {Hauke}\ \emph {et~al.}(2020)\citenamefont {Hauke},
  \citenamefont {Katzgraber}, \citenamefont {Lechner}, \citenamefont
  {Nishimori},\ and\ \citenamefont
  {Oliver}}]{HaukeLechnerKatzgraberNishimori_2020}%
  \BibitemOpen
  \bibfield  {author} {\bibinfo {author} {\bibfnamefont {P.}~\bibnamefont
  {Hauke}}, \bibinfo {author} {\bibfnamefont {H.~G.}\ \bibnamefont
  {Katzgraber}}, \bibinfo {author} {\bibfnamefont {W.}~\bibnamefont {Lechner}},
  \bibinfo {author} {\bibfnamefont {H.}~\bibnamefont {Nishimori}},\ and\
  \bibinfo {author} {\bibfnamefont {W.~D.}\ \bibnamefont {Oliver}},\ }\bibfield
   {title} {\bibinfo {title} {Perspectives of quantum annealing: methods and
  implementations},\ }\href {https://doi.org/10.1088/1361-6633/ab85b8}
  {\bibfield  {journal} {\bibinfo  {journal} {Reports on Progress in Physics}\
  }\textbf {\bibinfo {volume} {83}},\ \bibinfo {pages} {054401} (\bibinfo
  {year} {2020})}\BibitemShut {NoStop}%
\bibitem [{\citenamefont {Rajak}\ \emph {et~al.}(2022)\citenamefont {Rajak},
  \citenamefont {Suzuki}, \citenamefont {Dutta},\ and\ \citenamefont
  {Chakrabarti}}]{QuantumAnnealingOverview}%
  \BibitemOpen
  \bibfield  {author} {\bibinfo {author} {\bibfnamefont {A.}~\bibnamefont
  {Rajak}}, \bibinfo {author} {\bibfnamefont {S.}~\bibnamefont {Suzuki}},
  \bibinfo {author} {\bibfnamefont {A.}~\bibnamefont {Dutta}},\ and\ \bibinfo
  {author} {\bibfnamefont {B.~K.}\ \bibnamefont {Chakrabarti}},\ }\bibfield
  {title} {\bibinfo {title} {Quantum annealing: an overview},\ }\href
  {https://doi.org/10.1098/rsta.2021.0417} {\bibfield  {journal} {\bibinfo
  {journal} {Phil. Trans. R. Soc. A.}\ }\textbf {\bibinfo {volume} {381}},\
  \bibinfo {pages} {20210417} (\bibinfo {year} {2022})}\BibitemShut {NoStop}%
\bibitem [{\citenamefont {Farhi}\ \emph {et~al.}(2000)\citenamefont {Farhi},
  \citenamefont {Goldstone}, \citenamefont {Gutmann},\ and\ \citenamefont
  {Sipser}}]{farhi2000quantum}%
  \BibitemOpen
  \bibfield  {author} {\bibinfo {author} {\bibfnamefont {E.}~\bibnamefont
  {Farhi}}, \bibinfo {author} {\bibfnamefont {J.}~\bibnamefont {Goldstone}},
  \bibinfo {author} {\bibfnamefont {S.}~\bibnamefont {Gutmann}},\ and\ \bibinfo
  {author} {\bibfnamefont {M.}~\bibnamefont {Sipser}},\ }\bibfield  {title}
  {\bibinfo {title} {Quantum computation by adiabatic evolution},\ }\href
  {https://arxiv.org/abs/quant-ph/0001106} {\bibfield  {journal} {\bibinfo
  {journal} {arXiv preprint quant-ph/0001106}\ } (\bibinfo {year}
  {2000})}\BibitemShut {NoStop}%
\bibitem [{\citenamefont {Cerezo}\ \emph {et~al.}()\citenamefont {Cerezo},
  \citenamefont {Arrasmith}, \citenamefont {Babbush}, \citenamefont {Benjamin},
  \citenamefont {Endo}, \citenamefont {Fujii}, \citenamefont {{McClean}},
  \citenamefont {Mitarai}, \citenamefont {Yuan}, \citenamefont {Cincio},\ and\
  \citenamefont {Coles}}]{vqa}%
  \BibitemOpen
  \bibfield  {author} {\bibinfo {author} {\bibfnamefont {M.}~\bibnamefont
  {Cerezo}}, \bibinfo {author} {\bibfnamefont {A.}~\bibnamefont {Arrasmith}},
  \bibinfo {author} {\bibfnamefont {R.}~\bibnamefont {Babbush}}, \bibinfo
  {author} {\bibfnamefont {S.~C.}\ \bibnamefont {Benjamin}}, \bibinfo {author}
  {\bibfnamefont {S.}~\bibnamefont {Endo}}, \bibinfo {author} {\bibfnamefont
  {K.}~\bibnamefont {Fujii}}, \bibinfo {author} {\bibfnamefont {J.~R.}\
  \bibnamefont {{McClean}}}, \bibinfo {author} {\bibfnamefont {K.}~\bibnamefont
  {Mitarai}}, \bibinfo {author} {\bibfnamefont {X.}~\bibnamefont {Yuan}},
  \bibinfo {author} {\bibfnamefont {L.}~\bibnamefont {Cincio}},\ and\ \bibinfo
  {author} {\bibfnamefont {P.~J.}\ \bibnamefont {Coles}},\ }\bibfield  {title}
  {\bibinfo {title} {Variational quantum algorithms},\ }\href
  {https://doi.org/https://doi.org/10.1038/s42254-021-00348-9} {\ \textbf
  {\bibinfo {volume} {3}},\ \bibinfo {pages} {625}}\BibitemShut {NoStop}%
\bibitem [{\citenamefont {Or\'us}\ and\ \citenamefont
  {Latorre}(2004)}]{PhysRevA.69.052308}%
  \BibitemOpen
  \bibfield  {author} {\bibinfo {author} {\bibfnamefont {R.}~\bibnamefont
  {Or\'us}}\ and\ \bibinfo {author} {\bibfnamefont {J.~I.}\ \bibnamefont
  {Latorre}},\ }\bibfield  {title} {\bibinfo {title} {Universality of
  entanglement and quantum-computation complexity},\ }\href
  {https://doi.org/10.1103/PhysRevA.69.052308} {\bibfield  {journal} {\bibinfo
  {journal} {Phys. Rev. A}\ }\textbf {\bibinfo {volume} {69}},\ \bibinfo
  {pages} {052308} (\bibinfo {year} {2004})}\BibitemShut {NoStop}%
\bibitem [{\citenamefont {Lanting}\ \emph {et~al.}(2014)\citenamefont
  {Lanting}, \citenamefont {Przybysz}, \citenamefont {Smirnov}, \citenamefont
  {Spedalieri}, \citenamefont {Amin}, \citenamefont {Berkley}, \citenamefont
  {Harris}, \citenamefont {Altomare}, \citenamefont {Boixo}, \citenamefont
  {Bunyk}, \citenamefont {Dickson}, \citenamefont {Enderud}, \citenamefont
  {Hilton}, \citenamefont {Hoskinson}, \citenamefont {Johnson}, \citenamefont
  {Ladizinsky}, \citenamefont {Ladizinsky}, \citenamefont {Neufeld},
  \citenamefont {Oh}, \citenamefont {Perminov}, \citenamefont {Rich},
  \citenamefont {Thom}, \citenamefont {Tolkacheva}, \citenamefont {Uchaikin},
  \citenamefont {Wilson},\ and\ \citenamefont
  {Rose}}]{LantingEntanglementInQA}%
  \BibitemOpen
  \bibfield  {author} {\bibinfo {author} {\bibfnamefont {T.}~\bibnamefont
  {Lanting}}, \bibinfo {author} {\bibfnamefont {A.~J.}\ \bibnamefont
  {Przybysz}}, \bibinfo {author} {\bibfnamefont {A.~Y.}\ \bibnamefont
  {Smirnov}}, \bibinfo {author} {\bibfnamefont {F.~M.}\ \bibnamefont
  {Spedalieri}}, \bibinfo {author} {\bibfnamefont {M.~H.}\ \bibnamefont
  {Amin}}, \bibinfo {author} {\bibfnamefont {A.~J.}\ \bibnamefont {Berkley}},
  \bibinfo {author} {\bibfnamefont {R.}~\bibnamefont {Harris}}, \bibinfo
  {author} {\bibfnamefont {F.}~\bibnamefont {Altomare}}, \bibinfo {author}
  {\bibfnamefont {S.}~\bibnamefont {Boixo}}, \bibinfo {author} {\bibfnamefont
  {P.}~\bibnamefont {Bunyk}}, \bibinfo {author} {\bibfnamefont
  {N.}~\bibnamefont {Dickson}}, \bibinfo {author} {\bibfnamefont
  {C.}~\bibnamefont {Enderud}}, \bibinfo {author} {\bibfnamefont {J.~P.}\
  \bibnamefont {Hilton}}, \bibinfo {author} {\bibfnamefont {E.}~\bibnamefont
  {Hoskinson}}, \bibinfo {author} {\bibfnamefont {M.~W.}\ \bibnamefont
  {Johnson}}, \bibinfo {author} {\bibfnamefont {E.}~\bibnamefont {Ladizinsky}},
  \bibinfo {author} {\bibfnamefont {N.}~\bibnamefont {Ladizinsky}}, \bibinfo
  {author} {\bibfnamefont {R.}~\bibnamefont {Neufeld}}, \bibinfo {author}
  {\bibfnamefont {T.}~\bibnamefont {Oh}}, \bibinfo {author} {\bibfnamefont
  {I.}~\bibnamefont {Perminov}}, \bibinfo {author} {\bibfnamefont
  {C.}~\bibnamefont {Rich}}, \bibinfo {author} {\bibfnamefont {M.~C.}\
  \bibnamefont {Thom}}, \bibinfo {author} {\bibfnamefont {E.}~\bibnamefont
  {Tolkacheva}}, \bibinfo {author} {\bibfnamefont {S.}~\bibnamefont
  {Uchaikin}}, \bibinfo {author} {\bibfnamefont {A.~B.}\ \bibnamefont
  {Wilson}},\ and\ \bibinfo {author} {\bibfnamefont {G.}~\bibnamefont {Rose}},\
  }\bibfield  {title} {\bibinfo {title} {Entanglement in a quantum annealing
  processor},\ }\href {https://doi.org/10.1103/PhysRevX.4.021041} {\bibfield
  {journal} {\bibinfo  {journal} {Phys. Rev. X}\ }\textbf {\bibinfo {volume}
  {4}},\ \bibinfo {pages} {021041} (\bibinfo {year} {2014})}\BibitemShut
  {NoStop}%
\bibitem [{\citenamefont {Hauke}\ \emph {et~al.}(2015)\citenamefont {Hauke},
  \citenamefont {Bonnes}, \citenamefont {Heyl},\ and\ \citenamefont
  {Lechner}}]{haukeLechnerEntenglement}%
  \BibitemOpen
  \bibfield  {author} {\bibinfo {author} {\bibfnamefont {P.}~\bibnamefont
  {Hauke}}, \bibinfo {author} {\bibfnamefont {L.}~\bibnamefont {Bonnes}},
  \bibinfo {author} {\bibfnamefont {M.}~\bibnamefont {Heyl}},\ and\ \bibinfo
  {author} {\bibfnamefont {W.}~\bibnamefont {Lechner}},\ }\bibfield  {title}
  {\bibinfo {title} {Probing entanglement in adiabatic quantum optimization
  with trapped ions},\ }\href
  {https://www.frontiersin.org/articles/10.3389/fphy.2015.00021} {\bibfield
  {journal} {\bibinfo  {journal} {Frontiers in Physics}\ }\textbf {\bibinfo
  {volume} {3}} (\bibinfo {year} {2015})}\BibitemShut {NoStop}%
\bibitem [{\citenamefont {Arute}\ \emph {et~al.}(2019)\citenamefont {Arute},
  \citenamefont {Arya}, \citenamefont {Babbush}, \citenamefont {Bacon},
  \citenamefont {Bardin}, \citenamefont {Barends}, \citenamefont {Biswas},
  \citenamefont {Boixo}, \citenamefont {Brandao}, \citenamefont {Buell} \emph
  {et~al.}}]{arute2019quantum}%
  \BibitemOpen
  \bibfield  {author} {\bibinfo {author} {\bibfnamefont {F.}~\bibnamefont
  {Arute}}, \bibinfo {author} {\bibfnamefont {K.}~\bibnamefont {Arya}},
  \bibinfo {author} {\bibfnamefont {R.}~\bibnamefont {Babbush}}, \bibinfo
  {author} {\bibfnamefont {D.}~\bibnamefont {Bacon}}, \bibinfo {author}
  {\bibfnamefont {J.~C.}\ \bibnamefont {Bardin}}, \bibinfo {author}
  {\bibfnamefont {R.}~\bibnamefont {Barends}}, \bibinfo {author} {\bibfnamefont
  {R.}~\bibnamefont {Biswas}}, \bibinfo {author} {\bibfnamefont
  {S.}~\bibnamefont {Boixo}}, \bibinfo {author} {\bibfnamefont {F.~G.}\
  \bibnamefont {Brandao}}, \bibinfo {author} {\bibfnamefont {D.~A.}\
  \bibnamefont {Buell}}, \emph {et~al.},\ }\bibfield  {title} {\bibinfo {title}
  {Quantum supremacy using a programmable superconducting processor},\ }\href
  {https://www.nature.com/articles/s41586-019-1666-5} {\bibfield  {journal}
  {\bibinfo  {journal} {Nature}\ }\textbf {\bibinfo {volume} {574}},\ \bibinfo
  {pages} {505} (\bibinfo {year} {2019})}\BibitemShut {NoStop}%
\bibitem [{\citenamefont {Bernien}\ \emph {et~al.}(2017)\citenamefont
  {Bernien}, \citenamefont {Schwartz}, \citenamefont {Keesling}, \citenamefont
  {Levine}, \citenamefont {Omran}, \citenamefont {Pichler}, \citenamefont
  {Choi}, \citenamefont {Zibrov}, \citenamefont {Endres}, \citenamefont
  {Greiner}, \citenamefont {Vuleti{\'c}},\ and\ \citenamefont
  {Lukin}}]{bernien2017probing}%
  \BibitemOpen
  \bibfield  {author} {\bibinfo {author} {\bibfnamefont {H.}~\bibnamefont
  {Bernien}}, \bibinfo {author} {\bibfnamefont {S.}~\bibnamefont {Schwartz}},
  \bibinfo {author} {\bibfnamefont {A.}~\bibnamefont {Keesling}}, \bibinfo
  {author} {\bibfnamefont {H.}~\bibnamefont {Levine}}, \bibinfo {author}
  {\bibfnamefont {A.}~\bibnamefont {Omran}}, \bibinfo {author} {\bibfnamefont
  {H.}~\bibnamefont {Pichler}}, \bibinfo {author} {\bibfnamefont
  {S.}~\bibnamefont {Choi}}, \bibinfo {author} {\bibfnamefont {A.~S.}\
  \bibnamefont {Zibrov}}, \bibinfo {author} {\bibfnamefont {M.}~\bibnamefont
  {Endres}}, \bibinfo {author} {\bibfnamefont {M.}~\bibnamefont {Greiner}},
  \bibinfo {author} {\bibfnamefont {V.}~\bibnamefont {Vuleti{\'c}}},\ and\
  \bibinfo {author} {\bibfnamefont {M.~D.}\ \bibnamefont {Lukin}},\ }\bibfield
  {title} {\bibinfo {title} {Probing many-body dynamics on a 51-atom quantum
  simulator},\ }\href {https://doi.org/10.1038/nature24622} {\bibfield
  {journal} {\bibinfo  {journal} {Nature}\ }\textbf {\bibinfo {volume} {551}},\
  \bibinfo {pages} {579 EP } (\bibinfo {year} {2017})}\BibitemShut {NoStop}%
\bibitem [{\citenamefont {Koch}\ \emph {et~al.}(2007)\citenamefont {Koch},
  \citenamefont {Yu}, \citenamefont {Gambetta}, \citenamefont {Houck},
  \citenamefont {Schuster}, \citenamefont {Majer}, \citenamefont {Blais},
  \citenamefont {Devoret}, \citenamefont {Girvin},\ and\ \citenamefont
  {Schoelkopf}}]{KochPRAChargeInsensitive2007}%
  \BibitemOpen
  \bibfield  {author} {\bibinfo {author} {\bibfnamefont {J.}~\bibnamefont
  {Koch}}, \bibinfo {author} {\bibfnamefont {T.~M.}\ \bibnamefont {Yu}},
  \bibinfo {author} {\bibfnamefont {J.}~\bibnamefont {Gambetta}}, \bibinfo
  {author} {\bibfnamefont {A.~A.}\ \bibnamefont {Houck}}, \bibinfo {author}
  {\bibfnamefont {D.~I.}\ \bibnamefont {Schuster}}, \bibinfo {author}
  {\bibfnamefont {J.}~\bibnamefont {Majer}}, \bibinfo {author} {\bibfnamefont
  {A.}~\bibnamefont {Blais}}, \bibinfo {author} {\bibfnamefont {M.~H.}\
  \bibnamefont {Devoret}}, \bibinfo {author} {\bibfnamefont {S.~M.}\
  \bibnamefont {Girvin}},\ and\ \bibinfo {author} {\bibfnamefont {R.~J.}\
  \bibnamefont {Schoelkopf}},\ }\bibfield  {title} {\bibinfo {title}
  {Charge-insensitive qubit design derived from the cooper pair box},\ }\href
  {https://doi.org/10.1103/PhysRevA.76.042319} {\bibfield  {journal} {\bibinfo
  {journal} {Phys. Rev. A}\ }\textbf {\bibinfo {volume} {76}},\ \bibinfo
  {pages} {042319} (\bibinfo {year} {2007})}\BibitemShut {NoStop}%
\bibitem [{\citenamefont {Saffman}\ \emph {et~al.}(2010)\citenamefont
  {Saffman}, \citenamefont {Walker},\ and\ \citenamefont
  {M\o{}lmer}}]{saffman2010quantum}%
  \BibitemOpen
  \bibfield  {author} {\bibinfo {author} {\bibfnamefont {M.}~\bibnamefont
  {Saffman}}, \bibinfo {author} {\bibfnamefont {T.~G.}\ \bibnamefont
  {Walker}},\ and\ \bibinfo {author} {\bibfnamefont {K.}~\bibnamefont
  {M\o{}lmer}},\ }\bibfield  {title} {\bibinfo {title} {Quantum information
  with rydberg atoms},\ }\href {https://doi.org/10.1103/RevModPhys.82.2313}
  {\bibfield  {journal} {\bibinfo  {journal} {Rev. Mod. Phys.}\ }\textbf
  {\bibinfo {volume} {82}},\ \bibinfo {pages} {2313} (\bibinfo {year}
  {2010})}\BibitemShut {NoStop}%
\bibitem [{\citenamefont {Henriet}\ \emph {et~al.}(2020)\citenamefont
  {Henriet}, \citenamefont {Beguin}, \citenamefont {Signoles}, \citenamefont
  {Lahaye}, \citenamefont {Browaeys}, \citenamefont {Reymond},\ and\
  \citenamefont {Jurczak}}]{henriet2020quantum}%
  \BibitemOpen
  \bibfield  {author} {\bibinfo {author} {\bibfnamefont {L.}~\bibnamefont
  {Henriet}}, \bibinfo {author} {\bibfnamefont {L.}~\bibnamefont {Beguin}},
  \bibinfo {author} {\bibfnamefont {A.}~\bibnamefont {Signoles}}, \bibinfo
  {author} {\bibfnamefont {T.}~\bibnamefont {Lahaye}}, \bibinfo {author}
  {\bibfnamefont {A.}~\bibnamefont {Browaeys}}, \bibinfo {author}
  {\bibfnamefont {G.-O.}\ \bibnamefont {Reymond}},\ and\ \bibinfo {author}
  {\bibfnamefont {C.}~\bibnamefont {Jurczak}},\ }\bibfield  {title}
  {{\selectlanguage {english}\bibinfo {title} {Quantum computing with neutral
  atoms}},\ }\href {https://quantum-journal.org/papers/q-2020-09-21-327/}
  {\bibfield  {journal} {\bibinfo  {journal} {Quantum}\ }\textbf {\bibinfo
  {volume} {4}} (\bibinfo {year} {2020})}\BibitemShut {NoStop}%
\bibitem [{\citenamefont {Bloch}\ \emph {et~al.}(2008)\citenamefont {Bloch},
  \citenamefont {Dalibard},\ and\ \citenamefont {Zwerger}}]{bloch2008many}%
  \BibitemOpen
  \bibfield  {author} {\bibinfo {author} {\bibfnamefont {I.}~\bibnamefont
  {Bloch}}, \bibinfo {author} {\bibfnamefont {J.}~\bibnamefont {Dalibard}},\
  and\ \bibinfo {author} {\bibfnamefont {W.}~\bibnamefont {Zwerger}},\
  }\bibfield  {title} {\bibinfo {title} {Many-body physics with ultracold
  gases},\ }\href {https://doi.org/10.1103/RevModPhys.80.885} {\bibfield
  {journal} {\bibinfo  {journal} {Rev. Mod. Phys.}\ }\textbf {\bibinfo {volume}
  {80}},\ \bibinfo {pages} {885} (\bibinfo {year} {2008})}\BibitemShut
  {NoStop}%
\bibitem [{\citenamefont {Vyskocil}\ and\ \citenamefont
  {Djidjev}(2019)}]{a12040077}%
  \BibitemOpen
  \bibfield  {author} {\bibinfo {author} {\bibfnamefont {T.}~\bibnamefont
  {Vyskocil}}\ and\ \bibinfo {author} {\bibfnamefont {H.}~\bibnamefont
  {Djidjev}},\ }\bibfield  {title} {\bibinfo {title} {Embedding equality
  constraints of optimization problems into a quantum annealer},\ }\bibfield
  {journal} {\bibinfo  {journal} {Algorithms}\ }\textbf {\bibinfo {volume}
  {12}},\ \href {https://doi.org/10.3390/a12040077} {10.3390/a12040077}
  (\bibinfo {year} {2019})\BibitemShut {NoStop}%
\bibitem [{\citenamefont {Hen}\ and\ \citenamefont
  {Spedalieri}(2016)}]{ItayHenPRApplied2016}%
  \BibitemOpen
  \bibfield  {author} {\bibinfo {author} {\bibfnamefont {I.}~\bibnamefont
  {Hen}}\ and\ \bibinfo {author} {\bibfnamefont {F.~M.}\ \bibnamefont
  {Spedalieri}},\ }\bibfield  {title} {\bibinfo {title} {Quantum annealing for
  constrained optimization},\ }\href
  {https://doi.org/10.1103/PhysRevApplied.5.034007} {\bibfield  {journal}
  {\bibinfo  {journal} {Phys. Rev. Applied}\ }\textbf {\bibinfo {volume} {5}},\
  \bibinfo {pages} {034007} (\bibinfo {year} {2016})}\BibitemShut {NoStop}%
\bibitem [{\citenamefont {Hen}\ and\ \citenamefont
  {Sarandy}(2016)}]{ItayHenPRA2016}%
  \BibitemOpen
  \bibfield  {author} {\bibinfo {author} {\bibfnamefont {I.}~\bibnamefont
  {Hen}}\ and\ \bibinfo {author} {\bibfnamefont {M.~S.}\ \bibnamefont
  {Sarandy}},\ }\bibfield  {title} {\bibinfo {title} {Driver hamiltonians for
  constrained optimization in quantum annealing},\ }\href
  {https://doi.org/10.1103/PhysRevA.93.062312} {\bibfield  {journal} {\bibinfo
  {journal} {Phys. Rev. A}\ }\textbf {\bibinfo {volume} {93}},\ \bibinfo
  {pages} {062312} (\bibinfo {year} {2016})}\BibitemShut {NoStop}%
\bibitem [{\citenamefont {Drieb-Sch{\"{o}}n}\ \emph {et~al.}(2023)\citenamefont
  {Drieb-Sch{\"{o}}n}, \citenamefont {Ender}, \citenamefont {Javanmard},\ and\
  \citenamefont {Lechner}}]{constraintpaper}%
  \BibitemOpen
  \bibfield  {author} {\bibinfo {author} {\bibfnamefont {M.}~\bibnamefont
  {Drieb-Sch{\"{o}}n}}, \bibinfo {author} {\bibfnamefont {K.}~\bibnamefont
  {Ender}}, \bibinfo {author} {\bibfnamefont {Y.}~\bibnamefont {Javanmard}},\
  and\ \bibinfo {author} {\bibfnamefont {W.}~\bibnamefont {Lechner}},\
  }\bibfield  {title} {\bibinfo {title} {Parity {Q}uantum {O}ptimization:
  {E}ncoding {C}onstraints},\ }\href
  {https://doi.org/10.22331/q-2023-03-17-951} {\bibfield  {journal} {\bibinfo
  {journal} {{Quantum}}\ }\textbf {\bibinfo {volume} {7}},\ \bibinfo {pages}
  {951} (\bibinfo {year} {2023})}\BibitemShut {NoStop}%
\bibitem [{\citenamefont {Binder}\ and\ \citenamefont
  {Young}(1986)}]{RevModPhys.58.801}%
  \BibitemOpen
  \bibfield  {author} {\bibinfo {author} {\bibfnamefont {K.}~\bibnamefont
  {Binder}}\ and\ \bibinfo {author} {\bibfnamefont {A.~P.}\ \bibnamefont
  {Young}},\ }\bibfield  {title} {\bibinfo {title} {Spin glasses: Experimental
  facts, theoretical concepts, and open questions},\ }\href
  {https://doi.org/10.1103/RevModPhys.58.801} {\bibfield  {journal} {\bibinfo
  {journal} {Rev. Mod. Phys.}\ }\textbf {\bibinfo {volume} {58}},\ \bibinfo
  {pages} {801} (\bibinfo {year} {1986})}\BibitemShut {NoStop}%
\bibitem [{\citenamefont {Chakrabarti}\ \emph {et~al.}(2017)\citenamefont
  {Chakrabarti}, \citenamefont {Inoue}, \citenamefont {Tamura},\ and\
  \citenamefont {Tanaka}}]{cambridgeQuantumSPinGlassAnnealingComputatio}%
  \BibitemOpen
  \bibfield  {author} {\bibinfo {author} {\bibfnamefont {B.~K.}\ \bibnamefont
  {Chakrabarti}}, \bibinfo {author} {\bibfnamefont {J.-i.}\ \bibnamefont
  {Inoue}}, \bibinfo {author} {\bibfnamefont {R.}~\bibnamefont {Tamura}},\ and\
  \bibinfo {author} {\bibfnamefont {S.}~\bibnamefont {Tanaka}},\ }\href@noop {}
  {\emph {\bibinfo {title} {Quantum Spin Glasses, Annealing and Computation}}}\
  (\bibinfo  {publisher} {Cambridge University Press},\ \bibinfo {address}
  {Cambridge},\ \bibinfo {year} {2017})\BibitemShut {NoStop}%
\bibitem [{\citenamefont {Somma}\ \emph {et~al.}(2012)\citenamefont {Somma},
  \citenamefont {Nagaj},\ and\ \citenamefont
  {Kieferov\'a}}]{PhysRevLett.109.050501}%
  \BibitemOpen
  \bibfield  {author} {\bibinfo {author} {\bibfnamefont {R.~D.}\ \bibnamefont
  {Somma}}, \bibinfo {author} {\bibfnamefont {D.}~\bibnamefont {Nagaj}},\ and\
  \bibinfo {author} {\bibfnamefont {M.}~\bibnamefont {Kieferov\'a}},\
  }\bibfield  {title} {\bibinfo {title} {Quantum speedup by quantum
  annealing},\ }\href {https://doi.org/10.1103/PhysRevLett.109.050501}
  {\bibfield  {journal} {\bibinfo  {journal} {Phys. Rev. Lett.}\ }\textbf
  {\bibinfo {volume} {109}},\ \bibinfo {pages} {050501} (\bibinfo {year}
  {2012})}\BibitemShut {NoStop}%
\bibitem [{\citenamefont {Pati}\ and\ \citenamefont
  {Braunstein}(2009)}]{JournalofIndianInstofSciencRoleentanglemenQC}%
  \BibitemOpen
  \bibfield  {author} {\bibinfo {author} {\bibfnamefont {A.~K.}\ \bibnamefont
  {Pati}}\ and\ \bibinfo {author} {\bibfnamefont {S.~L.}\ \bibnamefont
  {Braunstein}},\ }\bibfield  {title} {\bibinfo {title} {Role of entanglement
  in quantum computation},\ }\href
  {https://journal.iisc.ac.in/index.php/iisc/article/view/105} {\bibfield
  {journal} {\bibinfo  {journal} {Journal of the Indian Institute of Science}\
  }\textbf {\bibinfo {volume} {89}},\ \bibinfo {pages} {295} (\bibinfo {year}
  {2009})}\BibitemShut {NoStop}%
\bibitem [{\citenamefont {Horodecki}\ \emph {et~al.}(2009)\citenamefont
  {Horodecki}, \citenamefont {Horodecki}, \citenamefont {Horodecki},\ and\
  \citenamefont {Horodecki}}]{RevModPhys.81.865}%
  \BibitemOpen
  \bibfield  {author} {\bibinfo {author} {\bibfnamefont {R.}~\bibnamefont
  {Horodecki}}, \bibinfo {author} {\bibfnamefont {P.}~\bibnamefont
  {Horodecki}}, \bibinfo {author} {\bibfnamefont {M.}~\bibnamefont
  {Horodecki}},\ and\ \bibinfo {author} {\bibfnamefont {K.}~\bibnamefont
  {Horodecki}},\ }\bibfield  {title} {\bibinfo {title} {Quantum entanglement},\
  }\href {https://doi.org/10.1103/RevModPhys.81.865} {\bibfield  {journal}
  {\bibinfo  {journal} {Rev. Mod. Phys.}\ }\textbf {\bibinfo {volume} {81}},\
  \bibinfo {pages} {865} (\bibinfo {year} {2009})}\BibitemShut {NoStop}%
\bibitem [{\citenamefont {Vidal}\ and\ \citenamefont
  {Werner}(2002)}]{PhysRevA.65.032314}%
  \BibitemOpen
  \bibfield  {author} {\bibinfo {author} {\bibfnamefont {G.}~\bibnamefont
  {Vidal}}\ and\ \bibinfo {author} {\bibfnamefont {R.~F.}\ \bibnamefont
  {Werner}},\ }\bibfield  {title} {\bibinfo {title} {Computable measure of
  entanglement},\ }\href {https://doi.org/10.1103/PhysRevA.65.032314}
  {\bibfield  {journal} {\bibinfo  {journal} {Phys. Rev. A}\ }\textbf {\bibinfo
  {volume} {65}},\ \bibinfo {pages} {032314} (\bibinfo {year}
  {2002})}\BibitemShut {NoStop}%
\bibitem [{\citenamefont {Calabrese}\ and\ \citenamefont
  {Lefevre}(2008)}]{PhysRevA.78.032329}%
  \BibitemOpen
  \bibfield  {author} {\bibinfo {author} {\bibfnamefont {P.}~\bibnamefont
  {Calabrese}}\ and\ \bibinfo {author} {\bibfnamefont {A.}~\bibnamefont
  {Lefevre}},\ }\bibfield  {title} {\bibinfo {title} {Entanglement spectrum in
  one-dimensional systems},\ }\href
  {https://doi.org/10.1103/PhysRevA.78.032329} {\bibfield  {journal} {\bibinfo
  {journal} {Phys. Rev. A}\ }\textbf {\bibinfo {volume} {78}},\ \bibinfo
  {pages} {032329} (\bibinfo {year} {2008})}\BibitemShut {NoStop}%
\bibitem [{\citenamefont {Pollmann}\ and\ \citenamefont
  {Moore}(2010)}]{Pollmann_2010}%
  \BibitemOpen
  \bibfield  {author} {\bibinfo {author} {\bibfnamefont {F.}~\bibnamefont
  {Pollmann}}\ and\ \bibinfo {author} {\bibfnamefont {J.~E.}\ \bibnamefont
  {Moore}},\ }\bibfield  {title} {\bibinfo {title} {Entanglement spectra of
  critical and near-critical systems in one dimension},\ }\href
  {https://doi.org/10.1088/1367-2630/12/2/025006} {\bibfield  {journal}
  {\bibinfo  {journal} {New Journal of Physics}\ }\textbf {\bibinfo {volume}
  {12}},\ \bibinfo {pages} {025006} (\bibinfo {year} {2010})}\BibitemShut
  {NoStop}%
\bibitem [{\citenamefont {Pichler}\ \emph {et~al.}(2016)\citenamefont
  {Pichler}, \citenamefont {Zhu}, \citenamefont {Seif}, \citenamefont
  {Zoller},\ and\ \citenamefont {Hafezi}}]{ES_PichlerZoller}%
  \BibitemOpen
  \bibfield  {author} {\bibinfo {author} {\bibfnamefont {H.}~\bibnamefont
  {Pichler}}, \bibinfo {author} {\bibfnamefont {G.}~\bibnamefont {Zhu}},
  \bibinfo {author} {\bibfnamefont {A.}~\bibnamefont {Seif}}, \bibinfo {author}
  {\bibfnamefont {P.}~\bibnamefont {Zoller}},\ and\ \bibinfo {author}
  {\bibfnamefont {M.}~\bibnamefont {Hafezi}},\ }\bibfield  {title} {\bibinfo
  {title} {Measurement protocol for the entanglement spectrum of cold atoms},\
  }\href {https://journals.aps.org/prx/abstract/10.1103/PhysRevX.6.041033}
  {\bibfield  {journal} {\bibinfo  {journal} {Phys. Rev. X}\ }\textbf {\bibinfo
  {volume} {6}},\ \bibinfo {pages} {041033} (\bibinfo {year}
  {2016})}\BibitemShut {NoStop}%
\bibitem [{\citenamefont {Cirac}\ \emph {et~al.}(2011)\citenamefont {Cirac},
  \citenamefont {Poilblanc}, \citenamefont {Schuch},\ and\ \citenamefont
  {Verstraete}}]{PhysRevB.83.245134}%
  \BibitemOpen
  \bibfield  {author} {\bibinfo {author} {\bibfnamefont {J.~I.}\ \bibnamefont
  {Cirac}}, \bibinfo {author} {\bibfnamefont {D.}~\bibnamefont {Poilblanc}},
  \bibinfo {author} {\bibfnamefont {N.}~\bibnamefont {Schuch}},\ and\ \bibinfo
  {author} {\bibfnamefont {F.}~\bibnamefont {Verstraete}},\ }\bibfield  {title}
  {\bibinfo {title} {Entanglement spectrum and boundary theories with projected
  entangled-pair states},\ }\href
  {https://link.aps.org/doi/10.1103/PhysRevB.83.245134} {\bibfield  {journal}
  {\bibinfo  {journal} {Phys. Rev. B}\ }\textbf {\bibinfo {volume} {83}},\
  \bibinfo {pages} {245134} (\bibinfo {year} {2011})}\BibitemShut {NoStop}%
\bibitem [{\citenamefont {Schuch}\ \emph {et~al.}(2013)\citenamefont {Schuch},
  \citenamefont {Poilblanc}, \citenamefont {Cirac},\ and\ \citenamefont
  {P\'erez-Garc\'{\i}a}}]{PhysRevLett.111.090501}%
  \BibitemOpen
  \bibfield  {author} {\bibinfo {author} {\bibfnamefont {N.}~\bibnamefont
  {Schuch}}, \bibinfo {author} {\bibfnamefont {D.}~\bibnamefont {Poilblanc}},
  \bibinfo {author} {\bibfnamefont {J.~I.}\ \bibnamefont {Cirac}},\ and\
  \bibinfo {author} {\bibfnamefont {D.}~\bibnamefont {P\'erez-Garc\'{\i}a}},\
  }\bibfield  {title} {\bibinfo {title} {Topological order in the projected
  entangled-pair states formalism: Transfer operator and boundary
  hamiltonians},\ }\href {https://doi.org/10.1103/PhysRevLett.111.090501}
  {\bibfield  {journal} {\bibinfo  {journal} {Phys. Rev. Lett.}\ }\textbf
  {\bibinfo {volume} {111}},\ \bibinfo {pages} {090501} (\bibinfo {year}
  {2013})}\BibitemShut {NoStop}%
\bibitem [{\citenamefont {Ender}\ \emph {et~al.}(2023)\citenamefont {Ender},
  \citenamefont {ter Hoeven}, \citenamefont {Niehoff}, \citenamefont
  {Drieb-Sch{\"{o}}n},\ and\ \citenamefont {Lechner}}]{compilerpaper}%
  \BibitemOpen
  \bibfield  {author} {\bibinfo {author} {\bibfnamefont {K.}~\bibnamefont
  {Ender}}, \bibinfo {author} {\bibfnamefont {R.}~\bibnamefont {ter Hoeven}},
  \bibinfo {author} {\bibfnamefont {B.~E.}\ \bibnamefont {Niehoff}}, \bibinfo
  {author} {\bibfnamefont {M.}~\bibnamefont {Drieb-Sch{\"{o}}n}},\ and\
  \bibinfo {author} {\bibfnamefont {W.}~\bibnamefont {Lechner}},\ }\bibfield
  {title} {\bibinfo {title} {Parity {Q}uantum {O}ptimization: {C}ompiler},\
  }\href {https://doi.org/10.22331/q-2023-03-17-950} {\bibfield  {journal}
  {\bibinfo  {journal} {{Quantum}}\ }\textbf {\bibinfo {volume} {7}},\ \bibinfo
  {pages} {950} (\bibinfo {year} {2023})}\BibitemShut {NoStop}%
\bibitem [{\citenamefont {Lechner}\ \emph {et~al.}(2015)\citenamefont
  {Lechner}, \citenamefont {Hauke},\ and\ \citenamefont
  {Zoller}}]{lechner2015quantum}%
  \BibitemOpen
  \bibfield  {author} {\bibinfo {author} {\bibfnamefont {W.}~\bibnamefont
  {Lechner}}, \bibinfo {author} {\bibfnamefont {P.}~\bibnamefont {Hauke}},\
  and\ \bibinfo {author} {\bibfnamefont {P.}~\bibnamefont {Zoller}},\
  }\bibfield  {title} {\bibinfo {title} {A quantum annealing architecture with
  all-to-all connectivity from local interactions},\ }\href
  {http://advances.sciencemag.org/content/1/9/e1500838} {\bibfield  {journal}
  {\bibinfo  {journal} {Sci. Adv.}\ }\textbf {\bibinfo {volume} {1}},\ \bibinfo
  {pages} {1500838} (\bibinfo {year} {2015})}\BibitemShut {NoStop}%
\bibitem [{\citenamefont {Choi}(2008)}]{choi_minor-embedding_2008}%
  \BibitemOpen
  \bibfield  {author} {\bibinfo {author} {\bibfnamefont {V.}~\bibnamefont
  {Choi}},\ }\bibfield  {title} {\bibinfo {title} {Minor-embedding in adiabatic
  quantum computation: {I}. {The} parameter setting problem},\ }\href
  {https://doi.org/10.1007/s11128-008-0082-9} {\bibfield  {journal} {\bibinfo
  {journal} {Quantum Information Processing}\ }\textbf {\bibinfo {volume}
  {7}},\ \bibinfo {pages} {193} (\bibinfo {year} {2008})}\BibitemShut {NoStop}%
\bibitem [{\citenamefont {Dreier}\ and\ \citenamefont
  {Lechner}(2026)}]{dreier2024uniqueness}%
  \BibitemOpen
  \bibfield  {author} {\bibinfo {author} {\bibfnamefont {F.}~\bibnamefont
  {Dreier}}\ and\ \bibinfo {author} {\bibfnamefont {W.}~\bibnamefont
  {Lechner}},\ }\bibfield  {title} {\bibinfo {title} {On the uniqueness of
  compiling graphs under the parity transformation},\ }\href
  {https://doi.org/10.1063/5.0274696} {\bibfield  {journal} {\bibinfo
  {journal} {Journal of Mathematical Physics}\ }\textbf {\bibinfo {volume}
  {67}},\ \bibinfo {pages} {022201} (\bibinfo {year} {2026})}\BibitemShut
  {NoStop}%
\bibitem [{\citenamefont {Johansson}\ \emph {et~al.}(2013)\citenamefont
  {Johansson}, \citenamefont {Nation},\ and\ \citenamefont {Nori}}]{Qutip2}%
  \BibitemOpen
  \bibfield  {author} {\bibinfo {author} {\bibfnamefont {J.}~\bibnamefont
  {Johansson}}, \bibinfo {author} {\bibfnamefont {P.}~\bibnamefont {Nation}},\
  and\ \bibinfo {author} {\bibfnamefont {F.}~\bibnamefont {Nori}},\ }\bibfield
  {title} {\bibinfo {title} {Qutip 2: A python framework for the dynamics of
  open quantum systems},\ }\href
  {https://doi.org/https://doi.org/10.1016/j.cpc.2012.11.019} {\bibfield
  {journal} {\bibinfo  {journal} {Computer Physics Communications}\ }\textbf
  {\bibinfo {volume} {184}},\ \bibinfo {pages} {1234} (\bibinfo {year}
  {2013})}\BibitemShut {NoStop}%
\end{thebibliography}
%

\newpage
\setcounter{secnumdepth}{2}
\newpage
\appendix
\addcontentsline{toc}{section}{Appendices}
\section{Parity embedding example\label{app:exampleParityEmbedding}}
The initial optimization problem (Fig.~\ref{fig:parityTransform4q}, left) of the embedded problem analyzed at the beginning of Sec.~\ref{sec:numVis} is given by the Ising spin Hamiltonian
$$
H = \sum_{i=1}^{4} J_i s_i + \sum_{1\leq i<j\leq 4} J_{ij} s_is_j.
$$
To implement local fields, the logical problem is first mapped to the logical graph in Fig.~\ref{fig:parityTransform4q} (middle) by introducing the auxiliary vertex with index $0$. Here the local fields $J_i$ correspond to the edges $\set{0,i}$ and the interactions $J_{ij}$ to $\set{i,j}$ for $1\leq i<j\leq 4$ in the logical graph. As a second step, the logical graph is mapped via the parity embedding to the physical graph given in Fig.~\ref{fig:parityTransform4q} (right), where each of the 10 edges of the logical graph corresponds to a vertex in the physical graph. The  local field $\tilde{J}_k$ of each physical vertex equals the weight of its corresponding edge of the logical graph. The Hamiltonian after applying the parity encoding, represented by the physical graph in Fig.~\ref{fig:parityTransform4q} (right), reads $$H_p = \sum_{i=1}^{10} \tilde{J}_i \pqbsz{i} + H_{C},$$
where the penalty constraints are given in
\begin{align*}
    H_{C} &=  C_{1} \pqbsz{1} \pqbsz{2} \pqbsz{5} + C_{2} \pqbsz{5} \pqbsz{6}\pqbsz{8}\\
    &+ C_{3} \pqbsz{8} \pqbsz{9} \pqbsz{10} + C_{4} \pqbsz{2} \pqbsz{3} \pqbsz{5}\pqbsz{6}\\
    &+ C_{5} \pqbsz{3} \pqbsz{4} \pqbsz{6} \pqbsz{7} + C_{6} \pqbsz{6} \pqbsz{7}\pqbsz{8}\pqbsz{9}.
\end{align*}
The physical vertices are laid out in a square grid, where the parity constraints
are presented by three-body (red triangles) or four-body (blue
squares) plaquettes. Note that $H_{C}$, and therefore the physical graph, is not unique because different sets of constraints can in principle be chosen. The only requirement is that the chosen constraints span the entire constraint space.

\begin{figure}[t]
\centering
\includegraphics[width=1\columnwidth]{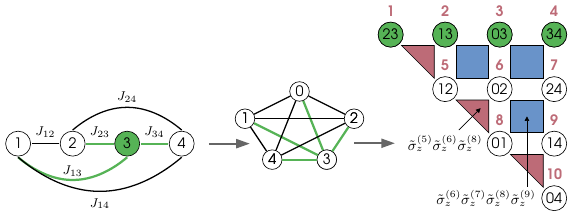}
\caption{\emph{Parity encoding} applied to the optimization problem used for the numerical simulations in Fig.~\ref{fig:GroupEntanglementDynamics}. The left graph illustrates the initial optimization problem, the middle graph the corresponding logical graph and the right picture the physical graph embedded by the parity encoding. The logical spins and the physical qubits are presented by black cycles. The logical vertices are labeled by indices ${0,1,2,3,4}$ inside the cycles and the vertices of the physical graph are labeled in red above the cycles with indices ${1, \ldots, 10}$. The indices inside the vertices of the physical graph indicate the corresponding edges of the logical problem graph.
The local fields of the Ising Hamiltonian are presented by interactions with a logical ancilla spin with index $0$. The red triangles (blue squares) are three-body (four-body) parity constraints. The logical line of vertex with index $3$ is marked in green. 
}
\label{fig:parityTransform4q}
\end{figure}
\section{Bundles of entanglement in the parity embedding}\label{app:bundlesEntanglement}
\begin{figure*}[th!]
\centering
    \includegraphics[width=1.8\columnwidth]{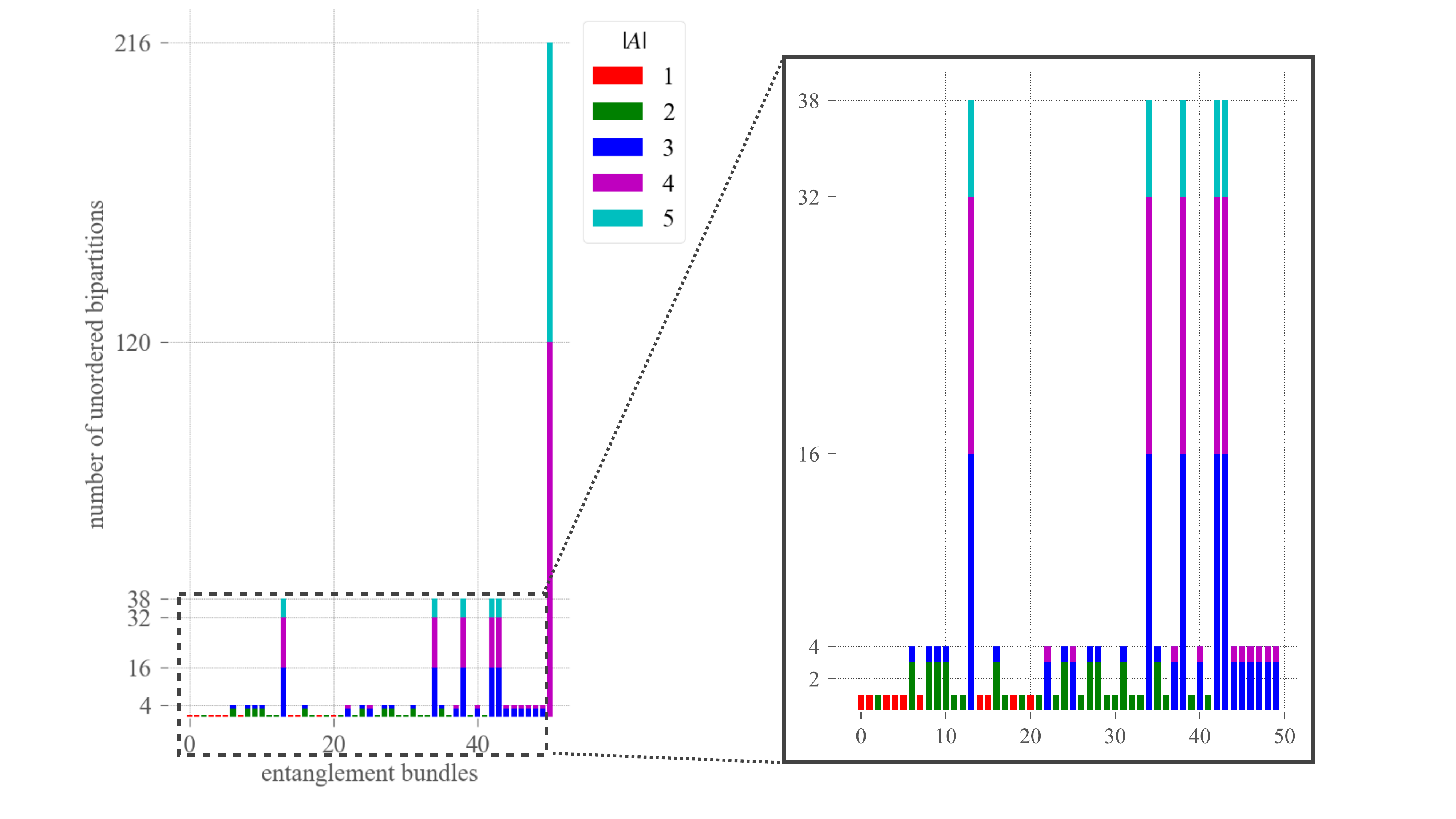}
\caption{\emph{Bundles of von Neumann entanglement entropy}.
On the left, all numerically detected entanglement entropy bundles for all $511$ bipartitions for the first numerical example in Sec.~\ref{sec:numVis} are represented with bars. The bundles are enumerated from $0$ to $50$ along the $x$ axis and sorted by their maximal entanglement entropy value in ascending order. The $y$ axis shows the number of bipartitions of each bundle. The different colors of the bars indicate the number of bipartitions of varying size $\abs{A}$. On the right, a zoomed-in view of the first $50$ bundles from the left figure is shown.
}
\label{fig:groupsOfFinalEntanglementValues}
\end{figure*}
Here, for the sake of completeness, we present the results of the entanglement entropy bundle analysis across all $511$ bipartitions of the numerical example in Fig.~\ref{fig:GroupEntanglementDynamics}~{(b)}.
We used \texttt{dbscan} to find all bundles of the final entanglement values ${S_{A}(t_f)}$ of all existing $511$ unordered bipartitions $\set{A,A^c}$. Because of numerical rounding errors and their propagation, we used a radius of $10^{-4}$ in the DBSCAN algorithm. The bundles of all $511$ unordered bipartitions are presented in Fig.~\ref{fig:groupsOfFinalEntanglementValues}, where the determined groups of the final entanglement entropy values are enumerated from the lowest to the highest entropy value. Bipartitions with equal size ${\abs{A}}$ are presented in the same color. We identified $26$ entropy values consisting of $20$ bundles with four bipartitions, five bundles with $38$ bipartitions and one bundle containing $216$ bipartitions. The remaining $25$ values correspond to bundles containing a single bipartition of size ${\abs{A}=1}$ (red boxes), or of size ${\abs{A}=2}$ (green boxes) if the subsystem $A$ is a set of two disjoint edges in the logical graph. In our example, where we consider a complete graph with five vertices, the number of pairs of disjoint edges is $15$ that we can count using the green boxes of height $1$.
We observe that there are two different types of bundles of size four. One bundle, which corresponds to type in Example~\ref{example:threeCycles}, consists of three bipartitions of size ${\abs{A}=2}$ and one bipartition of size ${\abs{A}=3}$. The other type consists of three bipartitions of size ${\abs{A}=3}$ and one of size ${\abs{A}=4}$, and corresponds to bundles of the form that are presented in Example~\ref{lemma:SameDisconnectedSubgraphs}. For both types of bundles we count 10 bundles.
The five bundles with $38$ bipartitions can be identified via logical lines and have been examined in Sec.~\ref{sec:embeddingExamples} in more detail. The bundle with the most bipartitions ($216$) corresponds to a type of bundle investigated in Example~\ref{lemma:highestEntanglementGroup}.

\section{Entanglement Spectrum}\label{app:entanglementSpectrum}
To underline our theoretical results in Theorem~\ref{thm:suff_cond_equivalence_spectrum_subspace_main}, we illustrate in Fig.~\ref{fig:entanglementSpectraLargestClass} the entanglement spectra of all $216$ bipartitions of the largest bipartition class for all three annealing processes presented in Fig.~\ref{fig:GroupEntanglementDynamics}. We observe in Fig.~\ref{fig:entanglementSpectraLargestClass}~{(b)}, where the annealing process is not adiabatic and restricted to the parity subspace $\Pi$ at the final annealing time, the final entanglement spectra of all bipartitions is the same.
\begin{figure*}[t]
\centering
\includegraphics[width=2\columnwidth]{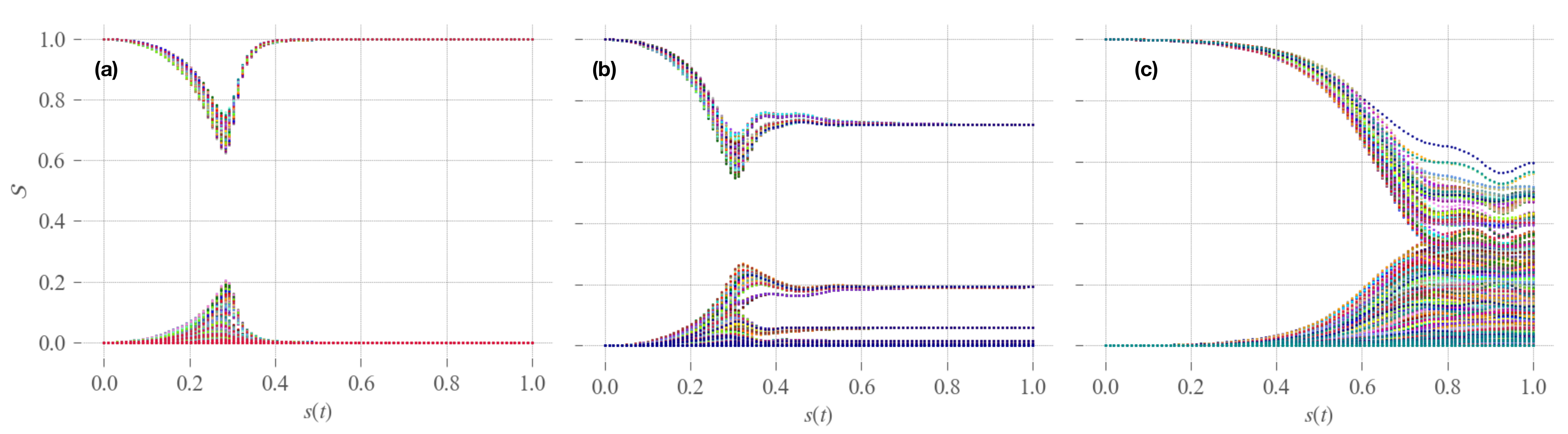}
\caption{\emph{Entanglement spectrum dynamics} for all bipartitions of the largest bipartition class given in the example at the beginning of Sec.~\ref{sec:numVis}, which contains all bipartitions of the highest entropy bundle in Fig.~\ref{fig:GroupEntanglementDynamics}~\textbf{(b)}. We present the entanglement spectra for all elements of this class for each of the different annealing process Fig.~\ref{fig:entanglementSpectraLargestClass}~\textbf{(a)} (adiabatic and restricted to $\Pi$), Fig.~\ref{fig:entanglementSpectraLargestClass}~\textbf{(b)} (nonadiabatic and restricted to $\Pi$) and Fig.~\ref{fig:entanglementSpectraLargestClass}~\textbf{(c)} (nonadiabatic and not restricted to the $\Pi$).
}
\label{fig:entanglementSpectraLargestClass}
\end{figure*}
\section{Dependence of bundling on the annealing implementation}\label{app:bundlingOnsetDependency}
Here, we examine the dependence of the bundling dynamics on implementation parameters of the annealing process, including the penalty strength $C$ of the embedding constraints and the selected annealing schedule. In addition, we investigate the impact of thermal noise on the bundling dynamics.
In the following, we consider an example with six parity embedded physical qubits. The Hamiltonian for the annealing processes is given by
$$
H(t) = A(t) H_{0} + B(t) \left(H_{J} - H_{C}\right),
$$
where we call the first term $A(t)H_{0}$ the \emph{driver Hamiltonian} and the second term with $B(t)\left(H_{J} - H_{C}\right)$ the \emph{problem Hamiltonian}.
The first part of the problem Hamiltonian is given by
$$
H_{J} = \sum^{5}_{m = 0} \tilde{J}_{m} \pqbsz{m}
$$
and the second part represents the constraint term
\begin{align*}
    H_{C} &= \\ &C \cdotp \left(\pqbsz{0}\pqbsz{1}\pqbsz{3} + \pqbsz{3}\pqbsz{4}\pqbsz{5} + \pqbsz{1}\pqbsz{2}\pqbsz{3}\pqbsz{4}\right).
\end{align*}
For the examples in Fig.~\ref{fig:onsetOfBundlingConstraintsPenaltystrength6qbParityExample}, Fig.~\ref{fig:DiffScheduleOnsetBundling} and Fig.\ref{fig:thermalNoise6qpParityExample}, we used the same logical problem, which, in the parity embedding, is characterized by the local fields
\begin{equation}
\label{eq:JprobExample}
    \tilde{J}=(0.96, 0.38, -0.56, -0.39, -0.54, 0.47)^{T}.
\end{equation}
\subsection{Varying penalty strengths for embedding constraints}\label{app:diffPenalties}
We investigate the dependence of the bundling dynamics on the penalty strength $C$ by simulating for each of $500$ random problem instances an annealing process for three different values of $C$.
As the annealing schedule, we again employ the linear schedule introduced in Sec.~\ref{sec:QAbackground} of the main text, defined by ${s(t)= t/t_f}$ and ${A(t) = 1- s(t)}$ and ${B(t) =s(t)}$.

For this class of annealing protocols, entanglement is known to develop in the regime where the driver and problem Hamiltonians compete and to peak shortly before the problem Hamiltonian takes over the dynamics~\cite{LantingEntanglementInQA}.
Only thereafter can the constraint term efficiently confine the system to $\QS\SSP$, giving rise to the emergence of bundling. Consequently, the entanglement maximum ${S_{A}(s^A_m) = \max_{s}S_{A}(s)}$ constitutes a natural reference point for identifying the onset of bundling.
Hence, we take the mean position ${\left<s^A_{\rm{m}}\right>_{A}}$ over the values $s^A_{\rm{m}}$ with maximal entanglement entropy
as a reference point for the onset of bundling and write for short $\left<s_{\rm{m}}\right>$.
The results are presented in Fig.~\ref{fig:onsetOfBundlingConstraintsPenaltystrength6qbParityExample} and Fig.~\ref{fig:C_AverageMaxEntropyInstances}.
While Fig.~\ref{fig:onsetOfBundlingConstraintsPenaltystrength6qbParityExample} demonstrates the shift of the bundling onset of the Hamiltonian using the local fields~\eqref{eq:JprobExample} toward earlier times for larger values of $C$, Fig.~\ref{fig:C_AverageMaxEntropyInstances} additionally shows  the difference between the onset times for various $C$ values across $500$ random instances with local fields $J_i$, which were randomly chosen from the interval $[-1, 1]$.
In Fig.~\ref{fig:onsetOfBundlingConstraintsPenaltystrength6qbParityExample}, we observe that the onset of bundling (indicated by the green dashed vertical line) shifts to earlier times as the penalty strength increases.
For Fig.~\ref{fig:C_AverageMaxEntropyInstances} we determined for each of the $500$ random instances the mean position ${\left<s_{\rm{m}}(C)\right>}$ of the maximal entanglement entropy for each of three different $C$ values and  calculated the differences between them.
As we can see, the consistently positive difference shows that increasing values of $C$ shift the onset of bundling toward earlier times.\\
A second noteworthy observation is that, although the $\left<s_{\rm{m}}\right>$ values of the individual random instances are distributed relatively uniformly over a broad range for each value of $C$, the corresponding shifts in $\mpme$ take only quantized values. Which of these discrete difference values arises for a given pair of $C$ values still appears to depend on the specific random instance.
This raises several further intriguing questions whose answers remain to be explored in future work.
\begin{figure}[t!]
\centering
\includegraphics[width=1.\columnwidth]{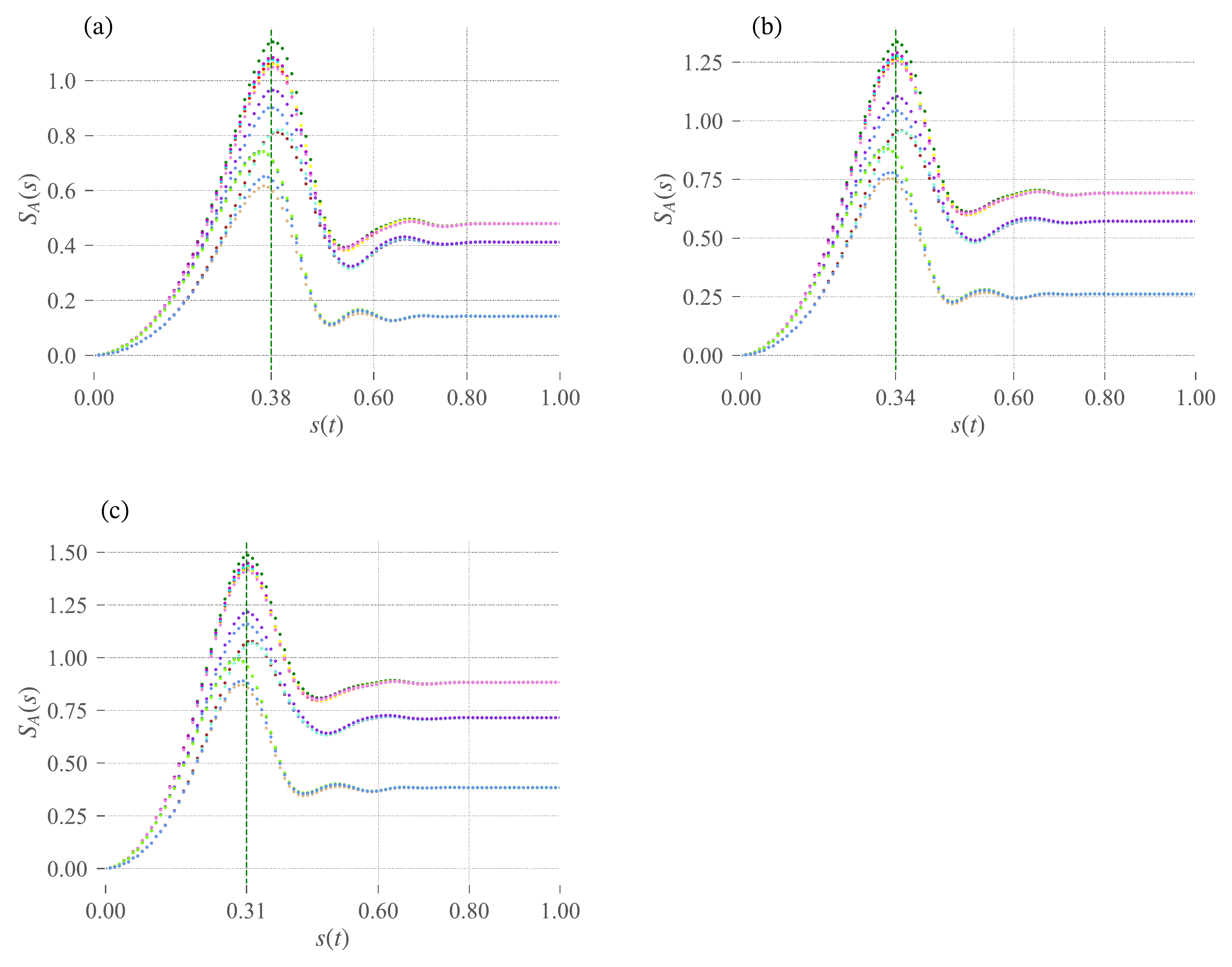}
\caption{\emph{Dependency of bundling dependency on constraint penalty strength.} Illustration of the von Neumann entropy of the bipartitions of three selected bundles of a six-qubit parity embedded problem with three different penalty strengths (a)~${C=3}$, (b)~${C=4}$ and (c)~${C=5}$. Each annealing processes is realized with a linear annealing schedule with ${t_f=50}$. The green dashed line marks the average time $\left<s_{\rm{m}}\right>$ of the entanglement entropy maxima, averaged over all bipartitions.}
\label{fig:onsetOfBundlingConstraintsPenaltystrength6qbParityExample}
\end{figure}
\begin{figure*}[t!]
\centering
\includegraphics[width=2.0\columnwidth]{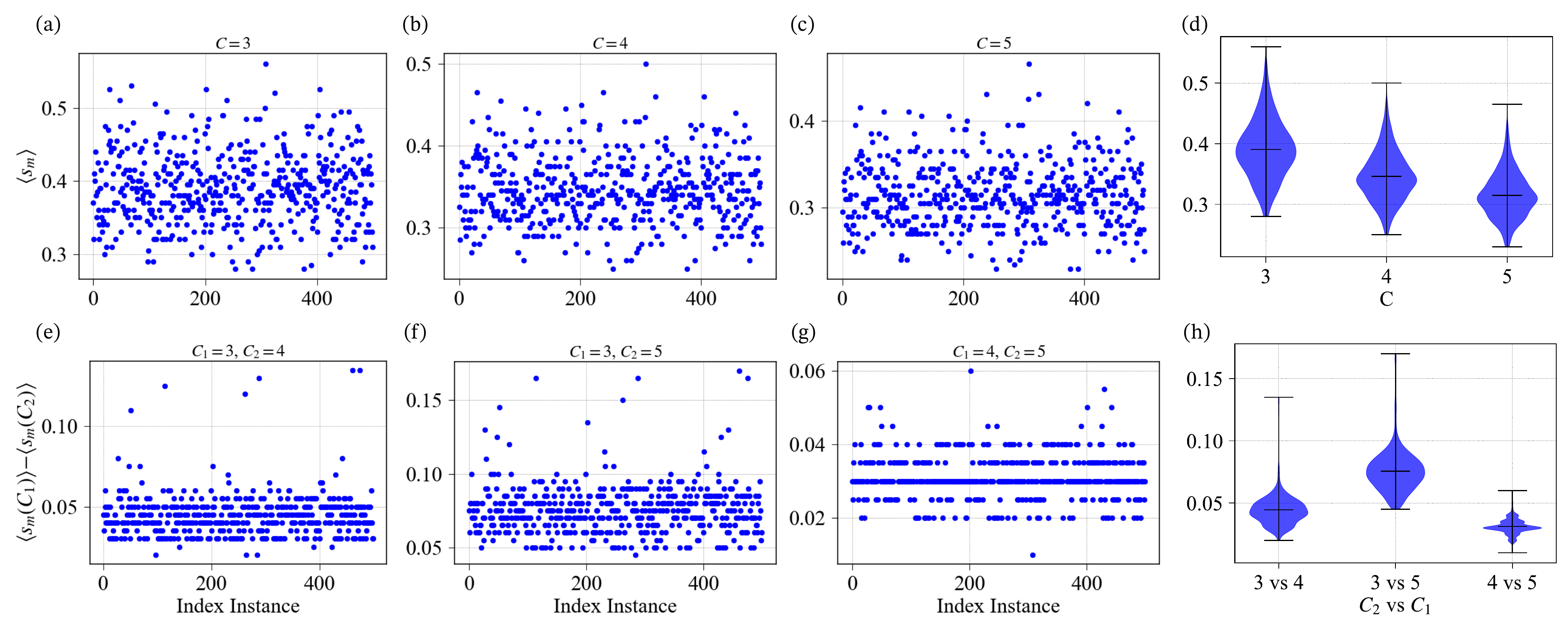}
\caption{\emph{Onset of bundling $\mpme$ for different penalty strengths $C$.} Illustration of the mean position $\mpme$ for each of the $500$ random instances is determined for the penalty strength (a)~${C=3}$, (b)~${C=4}$ and (c)~${C=5}$. (d)~Distribution of mean position values $\mpme$ for each of the three penalty strengths $C$. The lower row shows the differences in the mean values for two different $C$ values. Differences for the individual problem instances are shown for (e)~${C=3}$ versus ${C=4}$, (f)~${C=3}$ versus ${C=5}$, and (g)~${C=4}$ versus ${C=5}$. Panel (h)~compares the corresponding distributions. Each of the presented results uses the same set of $500$ random instances.
}
\label{fig:C_AverageMaxEntropyInstances}
\end{figure*}
\subsection{Different annealing schedules}\label{app:diffAnnealSchedules}
Similar to the investigations of the dependence of the mean position of the maximum entanglement $\mpme$ and the onset of the bundle on the penalty strength $C$, we analyze here the behavior of these entanglement characteristics for different annealing schedules. To this end, we consider different annealing schedules for a fixed ${C=4}$ by varying the function ${B(t) = s^{\alpha}(t)}$ through the parameter $\alpha$. The choice ${\alpha = 1}$ corresponds to the linear annealing schedule used for the results in Fig.~\ref{fig:onsetOfBundlingConstraintsPenaltystrength6qbParityExample}(b) and Fig.~\ref{fig:C_AverageMaxEntropyInstances}(b).

In Fig.~\ref{fig:DiffScheduleOnsetBundling}, we present the resulting entanglement dynamics for the same bipartitions of the same problem instance~\eqref{eq:JprobExample} as in Fig.~\ref{fig:onsetOfBundlingConstraintsPenaltystrength6qbParityExample}. We observe that increasing the value of $\alpha$ shifts the onset $\mpme$ to later time values. Compared with the variation across different $C$ values, a similar behavior is observed: when the constraint term $H_{C}$ of the Hamiltonian $H(t)$ reaches a certain relative weight with respect to the driver Hamiltonian at later times, both the onset of the bundling and the maximum entanglement are shifted to later times.\\
\begin{figure}[t!]
\centering
\includegraphics[width=1.\columnwidth]{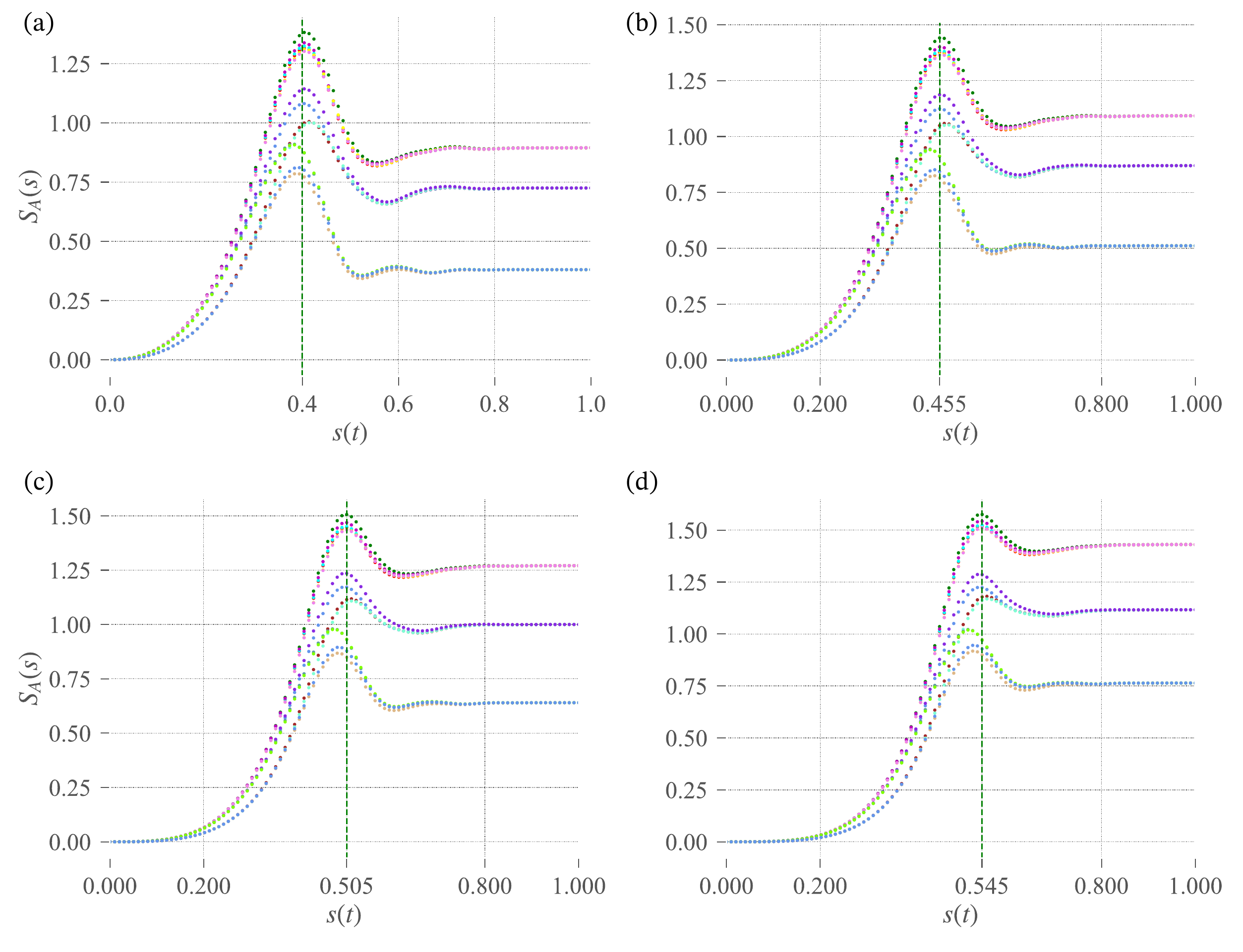}
\caption{\emph{Dependency of annealing schedule.} Entanglement dynamics for the same problem as in Fig.~\ref{fig:onsetOfBundlingConstraintsPenaltystrength6qbParityExample}~(b) with ${C=4}$. Each of the four plots shows on the $y$-axis the time dependent von Neumann Entanglement entropy and on the $x$-axis the normalized time scale ${s(t)=t/t_f}$.
As the parameter $\alpha$, we chose (a)~${\alpha=1.25}$, (b)~${\alpha=1.5}$, (c)~${\alpha=1.75}$ and (d)~${\alpha=2}$.
}
\label{fig:DiffScheduleOnsetBundling}
\end{figure}
For the $500$ random instances, we observe again that the differences between the $\mpme$ of two different $\alpha$ values cluster to a few distinct values [Fig.~\ref{fig:ScheduleAverageMaxEntropyInstances}(e)-(g)], although the individual $\mpme$ values of the random instances are broadly distributed across a certain time interval [Fig.~\ref{fig:ScheduleAverageMaxEntropyInstances}(a)-(c)].
In this case, the clustering of the $\mpme$ shifts is even more pronounce than that observed for different penalty strengths $C$.
The medians of these time ranges increase with increasing $\alpha$, while the widths of the ranges remain approximately constant [Fig.~\ref{fig:ScheduleAverageMaxEntropyInstances}(d)]. A similar behavior is observed for the ranges of the time differences [Fig.~\ref{fig:ScheduleAverageMaxEntropyInstances}(h)]: their widths remain nearly unchanged, whereas their medians exhibit a clear shift. Moreover, all time differences between a larger $\alpha$ value and a comparatively smaller one are positive.
\begin{figure*}[t!]
\centering
\includegraphics[width=2.0\columnwidth]{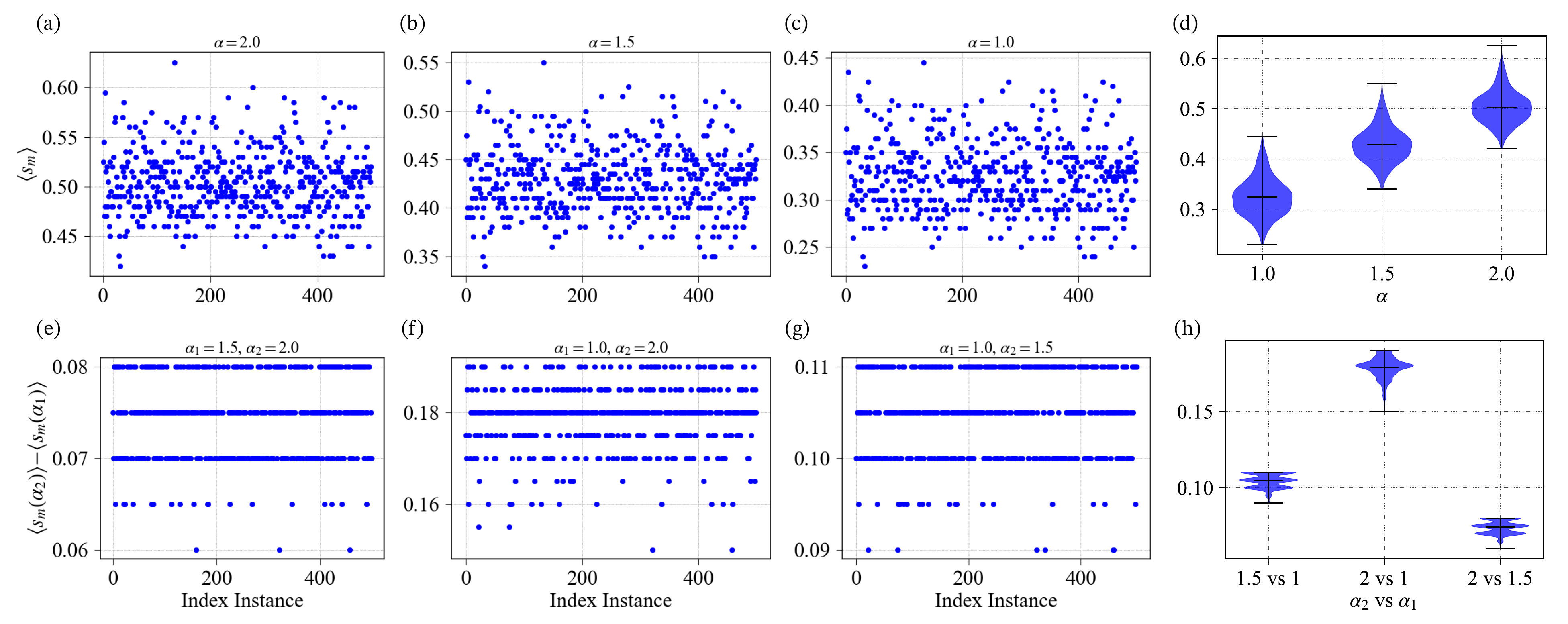}
\caption{\emph{Onset of bundling $\mpme$ for different annealing schedule parameters $\alpha$.} The mean position $\mpme$ for each of the $500$ random instances is determined for the schedule parameter (a)~${\alpha = 2.0}$, (b)~${\alpha = 1.5}$, (c)~${\alpha = 1.0}$. (d)~Distribution of mean position values $\mpme$ for each of the three schedule parameters $\alpha$. The lower row shows the differences for two different $\alpha$ values.
Differences for the individual problem instances are shown for (e)~${\alpha=1.5}$ versus ${\alpha=2.0}$, (f)~${\alpha=1.0}$ versus ${\alpha=2.0}$, and (g)~${\alpha=1.0}$ versus ${\alpha=1.5}$. Panel (h)~compares the corresponding distributions. Each of the presented results uses the same set of $500$ random instances.}
\label{fig:ScheduleAverageMaxEntropyInstances}
\end{figure*}
\subsection{Thermal Noise}\label{app:thermalNoise}
In this section, we investigate the effect of thermal noise on the entanglement bundles. To this end, we consider the entanglement dynamics for selected bipartitions of Example~\eqref{eq:JprobExample} under weak to moderate thermal noise in a linear annealing process with ${C=4}$. Thermal noise was simulated using the \texttt{mesolve} function of the \texttt{python} package \texttt{QuTiP} together with the \texttt{c\_ops} formalism for thermal dissipation. The results are shown in Fig.~\ref{fig:thermalNoise6qpParityExample}.
For thermal noise, we apply for each physical qubit the emission operator
$
    \sqrt{\gamma \cdotp (\delta_{th} + 1)} \sigma_{+}
$
and the absorption 
$
    \sqrt{\gamma \cdotp \delta_{th}} \sigma^{\dagger}_{+}
$,
where we use
$
    \delta^{-1}_{th} = e^{\frac{\omega}{T}}-1.
$
In this example, we used ${\omega = 1}$, ${T=0.05}$ and ${\gamma=0.001}$ for very low, ${\gamma=0.005}$ for low and ${\gamma=0.01}$ for medium noise.

As expected, coupling the system to a thermal bath leads to a progressive reduction of the total entanglement entropy within the system due to the loss of quantum information to the environment. We observe that the onset of the bundling of the entanglement still emerges around the point of maximal entropy. However, in the presence of thermal noise, the bundles gradually disperse into distinct individual values during the later stages of the annealing process. Moreover, the final width of the distribution around the mean bundle value increases with increasing noise strength.

\begin{figure}[t!]
\centering
\includegraphics[width=1.\columnwidth]{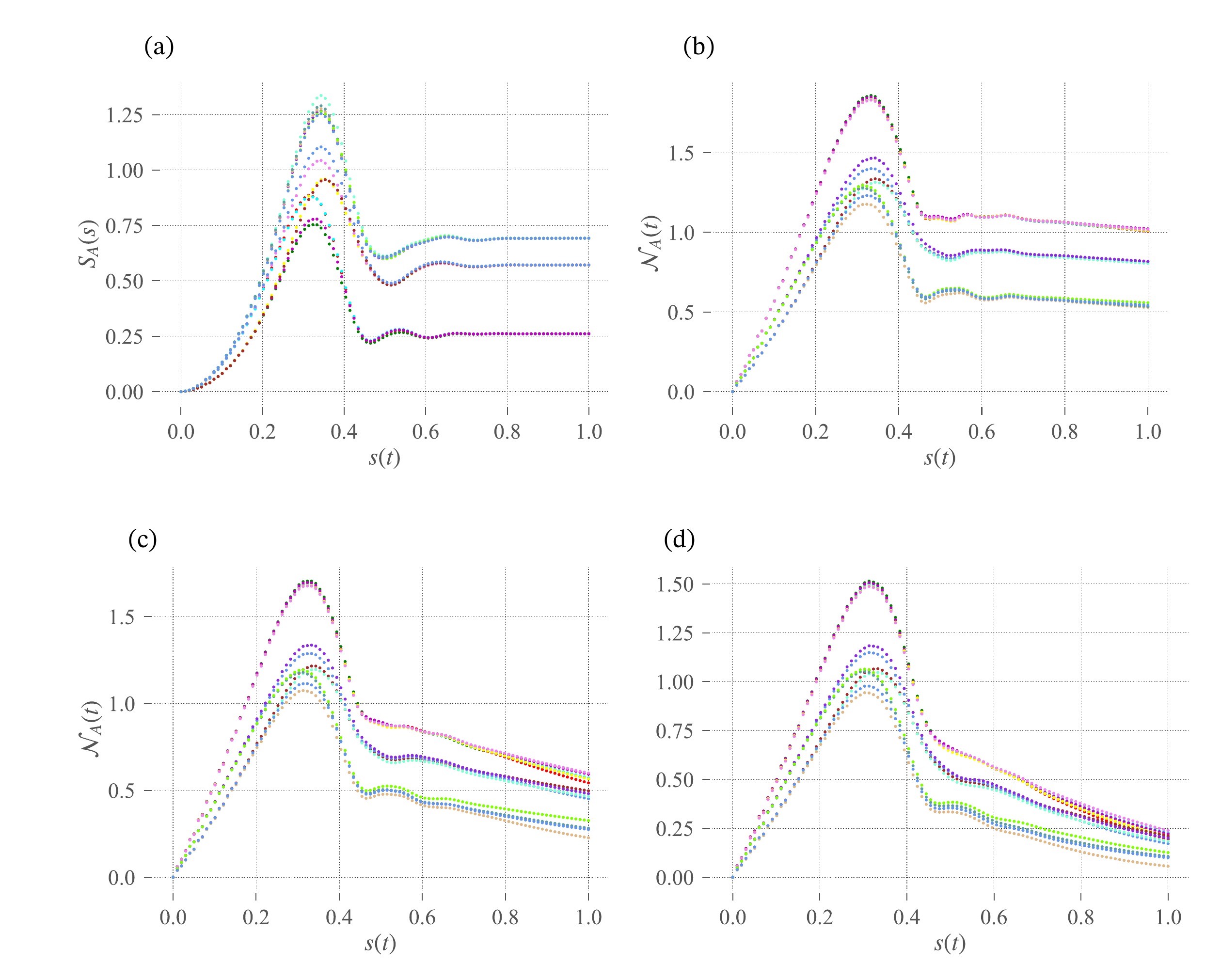}
\caption{\emph{Bundling with thermal noise.} Illustration of the entanglement entropy of an annealing process with linear schedule for a six-qubit parity embedding (a) without noise, (b) with very low (c) with low noise (d) with medium thermal noise. For the case without noise, we used the von Neumann entanglement entropy as a measure. For the three other cases with thermal noise we used the logarithmic negativity $\mathcal{N}_{A}$ as a measure for bipartite entanglement entropy.
}
\label{fig:thermalNoise6qpParityExample}
\end{figure}
\end{document}